\documentclass[a4paper,fleqn,numbers]{cas-sc}

\usepackage{graphicx}
\usepackage{subcaption}
\usepackage{amssymb}

\usepackage{longtable}
\usepackage{booktabs}
\usepackage{array}
\usepackage{amsmath,amssymb}

\usepackage{xcolor}

\stmAuthorSetup{color=black}

\ExplSyntaxOn
\RenewDocumentCommand \firstname {} 
{
  \textcolor{\l_stm_augroup_color_tl}
  {\seq_use:Nn \l_stm_au_seq { ~ }}
}
\ExplSyntaxOff

\usepackage[numbers,sort&compress]{natbib}

\usepackage{float}
\usepackage{placeins}
\def\tsc#1{\csdef{#1}{\textsc{\lowercase{#1}}\xspace}}
\tsc{WGM}
\tsc{QE}

\newtheorem{theorem}{Theorem}
\newtheorem{lemma}[theorem]{Lemma}
\newdefinition{rmk}{Remark}
\newproof{pf}{Proof}
\newdefinition{property}{Property}
\newdefinition{assumption}{Assumption}
\newcommand{\norm}[1]{\left\lVert #1 \right\rVert}

\begin{document}
\let\WriteBookmarks\relax
\renewcommand{\topfraction}{0.95}
\renewcommand{\bottomfraction}{0.90}
\renewcommand{\textfraction}{0.05}
\renewcommand{\floatpagefraction}{0.75}
\setcounter{topnumber}{4}
\setcounter{bottomnumber}{3}
\setcounter{totalnumber}{6}

\shorttitle{}    

\shortauthors{A. Haghighat et al.}  

\title [mode = title]{Adaptive Finite-Time Position-Force Control of Teleoperation Systems With Time-Varying Delays Using a Liquid State Machine Uncertainty Estimator}  



%

\author[1]{Shayan Akbari Haghighat}
\ead{S_akabari@elec.iust.ac.ir}

\author[1]{Mohammadali Ghaemifar}
\ead{m_ghaemifar@elec.iust.ac.ir}

\author[1]{Armin Attarzadeh}
\ead{armin_attarzadeh@elec.iust.ac.ir}

\author[1]{Mohammadreza Piri Sangdeh}
\ead{piri_m@elec.iust.ac.ir}



\affiliation[1]{organization={Department of Electrical Engineering, Iran University of Science and Technology (IUST) },
                city={Tehran},
                country={Iran}}


\begin{abstract}
Teleoperation systems are increasingly used in medical, rehabilitation, and remote manipulation applications, where accurate position/force tracking and stable interaction are essential. In such applications, the remote environment may exhibit viscoelasticity, frictional memory, contact transitions, and other dynamic interaction effects, causing the system response to depend not only on the current state but also on its previous evolution. This history dependence, together with communication delays and uncertain nonlinear dynamics, makes accurate uncertainty compensation particularly challenging. Conventional feedforward neural approximators do not inherently retain temporal information, while fully recurrent architectures may introduce additional computational and online training complexity. To address this limitation, this article introduces the first application of a liquid state machine (LSM) to bilateral teleoperation control. A finite-time adaptive controller is developed using a hybrid position/force auxiliary error system with velocity and force filters, while the LSM is employed to estimate uncertain dynamics by exploiting its intrinsic temporal processing and fading-memory capabilities with a simple adaptation mechanism. Closed-loop stability and finite-time convergence are established through a Lyapunov–Krasovskii framework. Simulations in spring–damper and generalized Maxwell viscoelastic environments demonstrate improved position and force tracking and lower mean execution time compared with an RBFNN-based controller.
\end{abstract}




\maketitle


\section*{}
\label{sec:nomenclature}

\begingroup
\renewcommand{\arraystretch}{1.15}

\begin{table}[pos=!htbp]
\centering
\caption{Nomenclature of abbreviations and variables used throughout this paper.}
\label{tab:nomenclature}

\begin{tabular}{@{}>{\raggedright\arraybackslash}p{0.20\linewidth}%
                  >{\raggedright\arraybackslash}p{0.73\linewidth}@{}}

\toprule
\textbf{Symbol} & \textbf{Description} \\
\midrule

$q_j,\ \dot q_j,\ \ddot q_j$ & Joint position, velocity and acceleration \\
$M_j,\ C_j,\ G_j,\ B_j$ & Inertia matrix, Coriolis matrix, gravity and friction vectors \\
$M_j^{0},\ C_j^{0},\ G_j^{0}$ & Nominal parts of $M_j$, $C_j$, $G_j$ \\
$\Delta M_j,\ \Delta C_j,\ \Delta G_j$ & Unknown parts of $M_j$, $C_j$, $G_j$ \\
$\lambda_{j,\min},\ \lambda_{j,\max}$ & Smallest and largest eigenvalue bounds of $M_j$ \\
$\tau_j$ & Control input torque \\
$\tau_h,\ \tau_e$ & Operator and environment torques \\
$F_h^{*},\ F_e^{*}$ & Exogenous operator and environment torques \\
$D_{he,j},\ S_{he,j}$ & Damping and stiffness matrices of the operator/environment model \\
$T_j(t)$ & Time-varying communication delay \\
$\bar T_j,\ \delta_j$ & Bounds of $T_j(t)$ and $|\dot T_j(t)|$ \\
$e_j$ & Position tracking error \\
$e_{\tau j}$ & Force tracking error \\
$\phi_j$ & Force error filter output \\
$\nu_j$ & Velocity feedback filter output \\
$L_j^{\tau},\ L_j^{v}$ & Force and velocity filter gains \\
$\xi_j$ & Hybrid auxiliary error system \\
$\chi_{j1},\ \chi_{j2}$ & Weighting matrices in $\xi_j$ \\
$e_{jv}$ & Delay-consistent velocity error \\
$e_j^{\nu}$ & Filtered auxiliary velocity error \\
$\psi_j$ & Finite-time term \\
$\lambda_{j1},\ \lambda_{j2}$ & Weighting matrices in $\psi_j$ \\
$\sigma_{j1},\ \sigma_{j2}$ & Fractional powers in $\psi_j$ \\
$\zeta_j$ & Integral auxiliary variable \\
$\operatorname{sig}(x)^{r}$ & Componentwise signed power \\
$u_j$ & Estimator input vector \\
$\pi_{j1},\ \pi_{j2},\ \pi_{j3}$ & Grouping terms in $u_j$ and $\Pi_j$ \\
$\Pi_j$ & Lumped model uncertainty \\
$s_j[k]$ & Reservoir spike indicator \\
$x_{L,j}$ & Filtered reservoir state \\
$X_j$ & Augmented feature vector \\
$\tau^{\mathrm{flt}}$ & Postsynaptic filter time constant \\
$\rho_j$ & Adaptation leakage coefficient \\
$\kappa$ & Retention threshold \\
$W_j^{*},\ \hat W_j,\ \tilde W_j$ & Ideal readout weights, estimate and error \\
$\varepsilon_j^{r}$ & Reservoir approximation error \\
$\omega_j^{r}$ & Bound on $\lVert \varepsilon_j^{r}-F_{he,j}^{*}\rVert^{2}$ \\
$K_{j1},\ K_{j2}$ & Linear and finite-time control gains \\
$\sigma$ & Fractional power of the finite-time control term \\
$a_{j1},\ a_{j2}$ & Young's inequality constants \\
$\gamma_j^{W},\ \gamma_j^{\delta},\ \gamma_j^{\omega}$ & Adaptation gains \\
$\hat\delta_{j'},\ \tilde\delta_{j'}$ & Estimate of $\delta_{j'}$ and its error \\
$\hat\omega_j^{r},\ \tilde\omega_j^{r}$ & Estimate of $\omega_j^{r}$ and its error \\
$\mathcal{D}$ & Residual set of $\zeta_j$ \\
$T_{\mathrm{fin}}$ & Reaching time bound \\

\bottomrule

\end{tabular}
\end{table}

\endgroup


\section{Introduction}
\label{sec:introduction}
Teleoperation continues to expand in applications requiring accurate bilateral interaction under uncertain and networked conditions. Recent examples include bilateral surgical teleoperation with spatial constraints and communication delays \cite{NEWAPP1}, haptic tele-echocardiography under network latency \cite{NEWAPP2}, cloud-based upper-limb telerehabilitation with bilateral control \cite{NEWAPP3}, and therapist-in-the-loop haptic teleoperation for remote arm-exoskeleton guidance \cite{NEWAPP4}. In such systems, communication delay arises from network transmission and can degrade performance or destabilize the closed loop. Anderson and Spong \cite{NEW1} introduced a scattering controller for delayed teleoperators, while Niemeyer and Slotine \cite{NEW2} developed stable adaptive teleoperation using wave variables. Munir and Book \cite{NEW3} subsequently combined wave variables with prediction to improve Internet teleoperation, and Chen et al. \cite{NEW4} proposed a delay-compensated four-channel wave-variable architecture for multilateral teleoperation. Since these approaches may reduce transparency and tracking performance, Lyapunov methods were also developed.
 Slawinski and Mut \cite{NEW5} studied PD-like control for delayed bilateral manipulators; Wais et al. \cite{NEW6} proposed a robust $H_\infty$ integral sliding-mode controller for nonlinear teleoperation with variable delay; and Yang et al. \cite{NEW7} considered neural output-feedback synchronization. These studies \cite{NEW5,NEW6} established important delay-robust control frameworks, but they generally retain explicit robot models, prescribed uncertainty bounds, or passivity/dissipation conditions. Consequently, uncertain manipulator dynamics and poorly known operator/environment interactions motivated adaptive teleoperation control.

Classical adaptive control has been widely used when the uncertain dynamics can be represented through known regressors. Liu and Khong \cite{NEW8} developed bilateral adaptive control for uncertain robot kinematics and dynamics. Chen et al. \cite{NEW9} considered nonlinear uncertain teleoperation while imposing guaranteed transient performance. Zhai and Xia \cite{NEW10} addressed a semi-autonomous teleoperation system with input uncertainties through structured adaptation. More recently, Lu et al. \cite{ID1} incorporated uncertain robot, operator, and environment parameters into a position-error-based adaptive force-feedback architecture without acceleration measurement. Khanzadeh et al. \cite{ID105} extended regressor-based adaptation to a single-master/multiple-slave teleoperation system with time-varying formation and communication delays. These controllers \cite{ID1,ID105} obtain their adaptive structure from a prescribed regression model; accordingly, friction, contact-dependent effects, actuator nonlinearities, and other unmodeled terms outside the selected parameterization cannot be reconstructed directly by the adaptive law. This limitation has led to extensive use of fuzzy systems and neural networks as nonlinear uncertainty approximators.

Adaptive fuzzy and neural control has therefore become a major branch of delayed teleoperation research. Li et al. \cite{NEW11} used multiple adaptive fuzzy systems for nonlinear teleoperators with stochastic time-varying communication delays, whereas Sun et al. \cite{NEW12} employed type-2 fuzzy modeling for dynamic uncertainty and time-varying delay. Yang et al. \cite{NEW13} combined wave variables with neural-network compensation for uncertain delayed teleoperation. Wang et al. \cite{NEW14} developed adaptive neural synchronization for bilateral systems with backlash-like hysteresis. Huang et al. \cite{ID7} subsequently combined RBFNN approximation with adaptive sliding-mode tracking and delayed force-feedback reference generation. Chen and Zhang \cite{ID3} used RBFNN residual learning together with position-only communication and velocity reconstruction, obtaining an acceleration-free design under asymmetric time-varying delays. Non-RBF alternatives have also been investigated: Kebria et al. \cite{ID87} developed interval type-2 fuzzy-neural adaptation for Internet teleoperation, while an MLP-based estimator \cite{ID88} was used for delayed reference/control reconstruction under random Internet delays. Moran-Armenta et al. \cite{ID4} more recently embedded a functional-link neural network in a P+d framework under asymmetric time-varying delays. Nevertheless, the guarantees in these representative adaptive schemes \cite{ID3,ID4} are mainly asymptotic or uniformly ultimately bounded rather than finite/fixed-time. Hence, finite-time designs have been introduced to improve transient convergence.

Finite-time teleoperation control has developed through several distinct routes. Yang et al. \cite{NEW15} proposed adaptive fuzzy finite-time coordination for networked nonlinear bilateral teleoperation. Zhai and Xia \cite{NEW16} incorporated input saturation and varying time delays into a finite-time controller. Terminal-sliding formulations were developed for delayed bilateral synchronization by Yang et al. \cite{NEW17}, and position-error constraints \cite{NEW18} were subsequently incorporated into a finite-time design. Wang et al. \cite{NEW19} proposed an adaptive finite-time bilateral controller, while Zhang et al. \cite{NEW20} developed adaptive finite-time synchronization for varying delays. Jittering time delays were considered by Wang et al. \cite{NEW21} using an adaptive finite-time architecture. More recent learning-based studies further diversified this line. Wang et al. \cite{ID11} combined RBFNN uncertainty approximation with nonsingular terminal sliding control and obtained finite-time position synchronization under unequal constant delays. Zhang et al. \cite{ID13} established practical fixed-time RBFNN synchronization under asymmetric time-varying delays. Li et al. \cite{ID69} introduced a zeroing-neural-dynamics/interval-type-2 fuzzy-neural finite-time design in 2026, again focusing on bilateral position synchronization. Force-capable finite/fixed-time schemes are fewer. Yang et al. \cite{ID72} studied fixed-time position--force synchronization for a rigid-master/flexible-joint-slave system; Li et al. \cite{ID75} combined finite-time synchronization with force reflection and output constraints; and their fixed-time extension \cite{ID76} imposed time-varying output constraints while reconstructing environment interaction. Thus, although finite/fixed-time synchronization is now well represented, simultaneous fast position--force behavior \cite{ID72,ID75} remains a narrower research line than position-only coordination.

Force/torque feedback has also been addressed independently of finite-time synchronization. Hashemzadeh and Tavakoli \cite{NEW22} developed position and force tracking for nonlinear teleoperators under varying delays, and Ganjefar et al. \cite{NEW23} extended position/force tracking to actuator nonlinearities with time-varying delay. Observer-based force reconstruction was studied by Amini et al. \cite{NEW24} for delayed bilateral systems, while Chan et al. \cite{NEW25} introduced an improved extended active observer for nonlinear position/force tracking. Sun et al. \cite{NEW26} used a type-2 fuzzy-neural moving-horizon force observer to improve force compliance and human perception, and Yuan et al. \cite{NEW27} considered force-reflecting control under time-varying delays. Recent learning-assisted force estimators include an RBFNN controller with a sensorless environment-force observer \cite{ID10} and an RBF environment model \cite{ID81} combined with disturbance-observer/sliding-mode control. However, these observer-oriented schemes \cite{ID10,ID81} depend on assumed robot or interaction models and do not by themselves remove all implementation-sensitive derivatives. Acceleration-level quantities remain explicit in representative terminal/RBF controllers such as the finite-time design of Wang et al. \cite{ID11}, and the RBF sliding-mode architecture of Huang et al. \cite{ID7} requires filtered reference acceleration. In contrast, Zhang et al. \cite{ID5} constructed force- and velocity-filtered auxiliary errors to obtain acceleration-free practical finite-time position--force tracking; Chen and Zhang \cite{ID3} used local velocity reconstruction for acceleration-free position synchronization; and non-neural finite-time controllers \cite{ID52,ID56} have also avoided acceleration sensing through terminal-sliding/regressor structures. Therefore, acceleration avoidance is achievable, but the strongest acceleration-free neural position--force result \cite{ID5} still relies on a conventional RBFNN uncertainty estimator.

RBFNNs remain particularly common in adaptive finite/fixed-time teleoperation, and the detailed designs reveal several narrower structural limitations. In the acceleration-free synchronization controller of Chen and Zhang \cite{ID3}, the Gaussian basis structure is selected in advance and online adaptation is concentrated on the output weights. Zhang et al. \cite{ID5} explicitly specify the number of hidden nodes together with Gaussian centers and widths in their finite-time position--force implementation. Huang et al. \cite{ID7} likewise use fixed RBF basis functions inside the adaptive sliding-mode controller while proving boundedness of the weight errors rather than identification of the ideal dynamics. Another P+d design \cite{ID64} retains the same pattern: fixed RBF basis parameters are used while the network weights and residual bounds are updated online. This group of works therefore places the burden of covering the operating region on a designer-selected Gaussian dictionary rather than adapting the representation itself. A separate limitation concerns the relationship between fast tracking and neural adaptation. Wang et al. \cite{ID11} establish finite-time tracking while the RBF weight errors remain bounded; Zhang et al. \cite{ID13} obtain practical fixed-time synchronization without ideal-weight convergence. Thus, finite/fixed-time error dynamics do not imply finite/fixed-time learning of the uncertainty model in these architectures. The burden increases further in hybrid force designs: the RBFNN in Zhang et al. \cite{ID5} approximates a composite uncertainty built from multiple auxiliary and interaction variables, whereas the constrained designs of Li et al. \cite{ID75,ID76} employ separate neural estimators for robot uncertainty and environment-force reconstruction. Non-RBF networks demonstrate that alternative representations are possible--including the functional-link network in \cite{ID4}, fuzzy-neural backstepping in \cite{ID74}, and MLP estimation in \cite{ID88}--but these alternatives have not simultaneously provided acceleration-free finite-time position--force tracking under asymmetric time-varying delays.

The dynamics of bilateral teleoperation are influenced by time-varying communication delays, nonlinear dynamics, and coupled interactions among the manipulators, human operator, and environment. Since the resulting uncertainties may depend on preceding state, force, and contact trajectories, they are more appropriately represented as history-dependent dynamic mappings rather than purely instantaneous nonlinear functions. These characteristics motivate the introduction of an LSM-based uncertainty approximator into the considered teleoperation framework. In this article, an LSM-based dynamic uncertainty approximator is introduced to estimate the aggregated uncertain dynamics of bilateral teleoperation. Unlike conventional RBFNN-based and feedforward approximators \cite{ID3,ID5,ID88,Ghaemifar2026}, whose temporal information must be supplied through manually constructed delayed or filtered regressors, the proposed liquid layer intrinsically preserves relevant motion and interaction history. This is important for representing frictional memory, contact transitions, variations in operator--environment behavior, and delayed bilateral coupling. The novelty lies in incorporating this fading-memory representation into a Lyapunov-based adaptive framework, enabling the controller to distinguish operating conditions that have similar instantaneous states but different preceding trajectories. Moreover, keeping the input and recurrent weights fixed and updating only the linear readout avoids online recurrent-network training, simplifies the adaptation mechanism, and improves the representation of history-dependent uncertainties.

Compared with RBFNN, fuzzy systems, and deep recurrent architectures such as LSTMs, LSM offers a practical balance between temporal modeling capability, computational efficiency, and suitability for online parameter adjustment. RBF networks and conventional fuzzy systems are mainly designed to approximate static nonlinear mappings. As a result, capturing dynamic or history-dependent uncertainties often requires delayed inputs, filtered signals, or an enlarged regressor vector, which rapidly expands the input dimensionality. Deep recurrent networks can learn temporal relationships directly; however, adapting their large number of nonlinear and recurrent parameters online can impose a considerable computational burden and make controller tuning and stability analysis difficult. This issue becomes even more significant when the network parameters must be updated online and can limit their suitability for real-time control applications. In contrast, an LSM uses a fixed recurrent reservoir to transform the input history into a rich fading-memory representation, while only the linear readout weights are updated online \cite{NEWLSM1,NEWLSM2,NEWLSM3}. This structure preserves the temporal information needed to approximate dynamic nonlinearities without requiring online training of the recurrent network. It also reduces the number of adaptive parameters and results in a linear formulation that is more convenient for real-time implementation and theoretical analysis. Efficient implementation of LSM temporal dynamics on dedicated neuromorphic hardware has also been demonstrated in \cite{NEWLSM4}. Therefore, LSMs provide an efficient and analytically tractable alternative for the adaptive approximation of nonlinear systems with time-varying and history-dependent dynamics. These properties make LSMs particularly suitable for bilateral teleoperation systems, where asymmetric communication delays, operator dynamics, and unstructured environment interactions induce complex, history-dependent dynamic uncertainties.

Taken together, the literature leaves three closely connected gaps. First, finite/fixed-time teleoperation has been developed much more extensively for position synchronization than for simultaneous position and force tracking, as summarized in Table~\ref{tab:literature_comparison}. Second, acceleration-free implementations exist, but only a limited subset combines acceleration avoidance with finite-time position--force objectives and asymmetric time-varying delays. Third, the most directly comparable adaptive finite/fixed-time controllers still use fixed-basis RBF representations whose tracking guarantees are stronger than their neural-parameter convergence guarantees. More importantly, these feedforward approximation structures do not intrinsically retain the preceding motion and interaction history that may distinguish different uncertainty realizations at similar instantaneous operating points. Accordingly, the main contributions of this article are summarized as follows.

\begin{itemize}
    \item A finite-time hybrid position/force tracking controller is developed for bilateral teleoperation with asymmetric time-varying delays and uncertain dynamics. Auxiliary error systems together with velocity and force filters are introduced to avoid the direct use of acceleration measurements and explicit time-delay derivatives while improving position synchronization and force tracking performance.

    \item A Lyapunov--Krasovskii framework is established for the complete bilateral closed-loop system. By accounting for the delayed master--slave dynamics and adaptive compensation terms, boundedness of the closed-loop signals and finite-time convergence of the position/force tracking errors are derived under the stated assumptions.

    \item A novel LSM-based uncertainty approximator is incorporated into the teleoperation controller to represent history-dependent uncertain dynamics. The recurrent spiking liquid naturally preserves fading-memory information from previous motion, interaction, and delayed coupling conditions. With fixed input and recurrent weights and adaptation applied only to the linear readout, the proposed structure avoids online recurrent-network training while retaining a simple form suitable for Lyapunov-based adaptive control.
\end{itemize}

\begin{table*}[pos=!htbp]
\centering
\caption{Comparison of the Proposed Method With Existing Teleoperation Controllers}
\label{tab:literature_comparison}
\renewcommand{\arraystretch}{1.15}
\resizebox{\textwidth}{!}{%
\begin{tabular}{lllcccc}
\hline
Reference
& Method (Year)
& Estimator Type
& Finite-Time?
& Force Tracking?
& Acceleration-Free?
& Time-Varying Delay? \\
\hline

\cite{ID1}
& Position-Error Adaptive Control (2022)
& Linear regressor adaptation
& $\times$
& $\times$
& $\checkmark$
& $\times$ \\

\cite{ID3}
& Velocity-Reconstructed Adaptive Synchronization (2022)
& RBFNN
& $\times$
& $\times$
& $\checkmark$
& $\checkmark$ \\

\cite{ID4}
& P+d Neural Compensation (2025)
& Functional-link NN
& $\times$
& $\times$
& $\checkmark$
& $\checkmark$ \\

\cite{ID7}
& Adaptive Sliding-Mode Control (2019)
& RBFNN
& $\times$
& $\times$
& $\times$
& $\checkmark$ \\

\cite{ID11}
& Nonsingular Terminal Sliding-Mode Control (2022)
& RBFNN
& $\checkmark$
& $\times$
& $\times$
& $\times$ \\

\cite{ID13}
& Adaptive Neural Fixed-Time Control (2022)
& RBFNN
& $\checkmark^{\mathrm{F}}$
& $\times$
& --
& $\checkmark$ \\

\cite{ID52}
& Robust Terminal-Sliding Coordination (2021)
& Structured regressor / robust compensation
& $\checkmark$
& $\times$
& $\checkmark$
& $\checkmark$ \\

\cite{ID69}
& ZND-Based Fuzzy-Neural Control (2026)
& Interval type-2 FNN
& $\checkmark$
& $\times$
& --
& $\checkmark$ \\

\cite{ID72}
& Composite Neural Fixed-Time Control (2024)
& Composite RBFNN
& $\checkmark^{\mathrm{F}}$
& $\checkmark$
& --
& $\checkmark$ \\

\cite{ID75}
& Funnel-Constrained Neural Control (2022)
& RBFNN + environment estimator
& $\checkmark$
& $\checkmark$
& --
& $\checkmark$ \\

\cite{ID76}
& Fixed-Time Output-Constrained Neural Control (2022)
& RBFNN + environment estimator
& $\checkmark^{\mathrm{F}}$
& $\checkmark$
& --
& $\checkmark$ \\

\cite{ID5}
& Hybrid Position--Force Finite-Time Control (2022)
& RBFNN
& $\checkmark$
& $\checkmark$
& $\checkmark$
& $\checkmark$ \\

Proposed
& Proposed Hybrid Finite-Time Control
& LSM-based dynamic approximator
& $\checkmark$
& $\checkmark$
& $\checkmark$
& $\checkmark$ \\

\hline
\end{tabular}}
\vspace{1mm}

\footnotesize{
$\checkmark$: supported;
$\times$: not supported;
$\checkmark^{\mathrm{F}}$: fixed-time convergence;
--: not explicitly reported.
}
\end{table*}

The remainder of this article is organized as follows. Section~\ref{sec:preliminaries} presents the bilateral teleoperation dynamics and preliminary results. Section~\ref{sec:control} introduces the proposed LSM-based uncertainty estimation framework and control design. Section~\ref{sec:lsm-stability} presents the finite-time control law and the Lyapunov--Krasovskii stability and convergence analysis. Section~\ref{sec:simulation} presents the simulation results. Finally, Section~\ref{sec:conclusion} concludes this article.

\section{Preliminaries}
\label{sec:preliminaries}

\subsection{Dynamics of the Bilateral Teleoperation}
\label{ssec:dyn}

Throughout, $j\in\{m,s\}$ indexes the master and the slave side, and
$j'$ denotes the opposite side, i.e. $j'=s$ when $j=m$ and $j'=m$ when
$j=s$. Both manipulators are serial chains of $n$ revolute joints, whose
rigid-body dynamics are written in the common indexed form
\begin{equation}\label{eq:dyn}
M_j(q_j)\,\ddot q_j
 + C_j(\dot q_j,q_j)\,\dot q_j
 + G_j(q_j) + B_j
 \;=\; \tau_j - \tau_{he,j},
\end{equation}
where $q_j,\dot q_j,\ddot q_j\in\mathbb{R}^{n}$ are the joint position,
velocity and acceleration vectors; $M_j(q_j)\in\mathbb{R}^{n\times n}$ is the
inertia matrix; $C_j(\dot q_j,q_j)\in\mathbb{R}^{n\times n}$ collects the
centripetal and Coriolis terms; $G_j(q_j)\in\mathbb{R}^{n}$ is the
gravitational torque; $B_j\in\mathbb{R}^{n}$ is the joint friction torque; and
$\tau_j\in\mathbb{R}^{n}$ is the commanded actuator torque. The term
$\tau_{he,j}\in\mathbb{R}^{n}$ is the torque exerted \emph{on} the
manipulator by its external partner,
\begin{equation}\label{eq:tauhe}
\tau_{he,m}\triangleq\tau_h,
\qquad
\tau_{he,s}\triangleq\tau_e,
\end{equation}
with $\tau_h$ the human-operator torque applied at the master and
$\tau_e$ the environment torque applied at the slave. 


Two structural properties of \eqref{eq:dyn} are used in the
stability analysis of Section~\ref{sec:lsm-stability}.

\begin{property}[Uniformly bounded inertia]\label{prop:inertia}
For each $j\in\{m,s\}$ the inertia matrix $M_j(q_j)$ is symmetric and
positive definite for every $q_j$, and admits configuration-independent
bounds
\begin{equation}\label{eq:inertiabound}
0<\lambda_{j,\min}I_n \;\preceq\; M_j(q_j) \;\preceq\; \lambda_{j,\max}I_n,
\qquad \forall\,q_j\in\mathbb{R}^{n},
\end{equation}
where
$\lambda_{j,\min}\triangleq\inf_{q_j}\lambda_{\min}\!\big(M_j(q_j)\big)>0$
and
$\lambda_{j,\max}\triangleq\sup_{q_j}\lambda_{\max}\!\big(M_j(q_j)\big)<\infty$,
with $\lambda_{\min}(\cdot)$ and $\lambda_{\max}(\cdot)$ the smallest and
largest eigenvalue of a matrix.
\end{property}

\begin{property}[Skew symmetry]\label{prop:skew}
For each $j\in\{m,s\}$ the matrix $\dot M_j-2C_j$ is skew symmetric,
hence
\begin{equation}\label{eq:skew}
x^{\mathrm T}\big(\dot M_j(q_j)-2C_j(\dot q_j,q_j)\big)x=0,
\qquad \forall\,x\in\mathbb{R}^{n}.
\end{equation}
\end{property}

\subsection{Known and Unknown Parts of the Dynamics}
\label{ssec:split}

Exact values of the inertial, Coriolis and gravitational terms in
\eqref{eq:dyn} are not available in practice. Each term is therefore
separated into a nominal component, computable online from the
identified model, and an unknown residual component:
\begin{equation}\label{eq:split}
M_j = M_j^{0} + \Delta M_j,
\qquad
C_j = C_j^{0} + \Delta C_j,
\qquad
G_j = G_j^{0} + \Delta G_j,
\end{equation}
for $j\in\{m,s\}$, where $M_j^{0},C_j^{0},G_j^{0}$ are the nominal parts
and $\Delta M_j,\Delta C_j,\Delta G_j$ the unknown parts. The friction
torque $B_j$ is treated as entirely unknown. The unknown terms, together with the interaction terms introduced
next, are lumped into a single uncertainty $\Pi_j$, assembled in
\eqref{eq:Pi}.

\subsection{Operator and Environment Interaction}
\label{ssec:interaction}
 
The operator and the environment are modelled as passive spring--damper
couplings,
\begin{equation}\label{eq:interaction}
\tau_{he,j}=F^{*}_{he,j}+D_{he,j}\,\dot q_j+S_{he,j}\,q_j,
\qquad j\in\{m,s\},
\end{equation}
with $(F^{*}_{he,m},D_{he,m},S_{he,m})=(F^{*}_{h},D_h,S_h)$ and
$(F^{*}_{he,s},D_{he,s},S_{he,s})=(F^{*}_{e},D_e,S_e)$. The
damping matrices $D_h,D_e\in\mathbb{R}^{n\times n}$ and the stiffness matrices
$S_h,S_e\in\mathbb{R}^{n\times n}$ are constant, diagonal and
positive definite, which renders both couplings passive, and
$F^{*}_{h},F^{*}_{e}\in\mathbb{R}^{n}$ are the exogenous operator and
environment torques.

\begin{assumption}
\label{as:exo}
The exogenous torques $F_h^{*}$ and $F_e^{*}$ are bounded.
\end{assumption}

The stiffness and damping terms in \eqref{eq:interaction} are functions of
the measured state and are absorbed into $\Pi_j$, so no knowledge of
$D_{he,j}$ or $S_{he,j}$ is required.

\subsection{Communication Delays}
\label{ssec:delay}

Let $T_m(t)$ denote the forward communication delay from master to slave
and $T_s(t)$ the backward delay from slave to master. The channel is
asymmetric, i.e. $T_m(t)\neq T_s(t)$ in general.

\begin{assumption}
\label{as:delay}
For $j\in\{m,s\}$ the delay $T_j(t)$ is non-negative, and there exist
constants $\bar T_j,\delta_j>0$ such that
\begin{equation}\label{eq:delaybound}
0\le T_j(t)\le \bar T_j,
\qquad
\big|\dot T_j(t)\big|\le \delta_j .
\end{equation}
\end{assumption}

\subsection{Standing Regularity Condition}
\label{ssec:compact}

\begin{assumption}
\label{as:compact}
For each $j\in\{m,s\}$, there exists a compact admissible operating
domain $\Omega_j$ on which the lumped uncertainty $\Pi_j$ is a bounded
causal fading-memory functional of the closed-loop history and of the
bounded communication-delay rate.
\end{assumption}

\begin{rmk}\label{rem:history}
$\Pi_j$ is not a pointwise function of an instantaneous state. Through the
auxiliary error construction it depends on $q_{j'}(t-T_{j'})$ and
$\dot q_{j'}(t-T_{j'})$, and it contains the factor $\dot T_{j'}$
multiplying an unknown inertia term, so $\Pi_j(t)$ is determined by the
closed-loop trajectory on $[t-\bar T_{j'},t]$. The delayed signals are
received over the channel and are available pointwise; the delay rate is
not. It is nonetheless reflected in the time course of the received stream,
since a varying delay rescales the apparent time base of the arriving
signal. An estimator evaluated at the instantaneous input cannot recover
this information, and a memoryless basis therefore requires lagged
regressors whose lags are fixed at design time. The reservoir of
Section~\ref{ssec:lsm} carries the history in its own state instead.
\end{rmk}

\subsection{Problem Formulation and Control Objectives}
\label{ssec:problem}

\subsubsection{Tracking errors}

Position synchronisation is measured on each side by comparing the local
joint position against the delayed position received from the partner:
\begin{equation}\label{eq:poserr}
\begin{aligned}
e_m(t) &= q_m(t) - q_s\!\left(t-T_s\right),\\
e_s(t) &= q_s(t) - q_m\!\left(t-T_m\right).
\end{aligned}
\end{equation}
With the interaction torques
$\tau_h,\tau_e$ obtained from the force sensors, the
torque tracking errors are defined as
\begin{equation}\label{eq:frcerr}
\begin{aligned}
e_{\tau m}(t) &= \tau_h(t) - \tau_e\!\left(t-T_s\right),\\
e_{\tau s}(t) &= \tau_e(t) - \tau_h\!\left(t-T_m\right).
\end{aligned}
\end{equation}
Each side is penalised for its own
mismatch against the delayed partner signal, so \eqref{eq:poserr} and
\eqref{eq:frcerr} are defined entirely from locally available data.

\subsubsection{Control objectives}

Given the teleoperator \eqref{eq:dyn} with the uncertainty structure
\eqref{eq:split}, the interaction model \eqref{eq:interaction}, and the
asymmetric time-varying delays of Assumption~\ref{as:delay}, the control
torques $\tau_m$ and $\tau_s$ are to be designed so as to meet the
following five objectives.

\begin{enumerate}
\item[O1)] \emph{Stability with finite-time reaching.}
The auxiliary closed-loop variables and adaptive estimation errors are
uniformly ultimately bounded in the presence of asymmetric time-varying
delay, dynamic uncertainty and bounded exogenous torques, and the
auxiliary variables reach a residual neighbourhood in finite time
$T_{\mathrm{fin}}$ rather than only asymptotically.

\item[O2)] \emph{Position tracking.} The bilateral controller is designed
to reduce the delayed position-synchronisation errors $e_j$ defined in
\eqref{eq:poserr}, with the position-tracking performance assessed by the
steady-state and transient magnitude of $\norm{e_j(t)}$ for
$j\in\{m,s\}$.

\item[O3)] \emph{Force feedback tracking.} The bilateral controller is
designed to reduce the delayed interaction-torque errors $e_{\tau j}$
defined in \eqref{eq:frcerr}, with the force-feedback performance assessed
by the steady-state and transient magnitude of
$\norm{e_{\tau j}(t)}$ for $j\in\{m,s\}$.

\item[O4)] \emph{Guarantee-preserving reconstruction of delay-dependent
uncertainty.} The lumped uncertainty $\Pi_j$ is to be reconstructed
online by an estimator that (i) remains linear in its adaptive
parameters, so that the adaptive law admits a Lyapunov certificate of the
same form, and (ii) represents the history dependence of $\Pi_j$ noted in
Remark~\ref{rem:history} without hand-crafted delayed regressors whose
lags must be fixed at design time.

\item[O5)] \emph{Sparse, event-driven computation.} Objectives O1--O4 are
to be met while the estimator activates only a small fraction of its
units per control step. The cost is quantified by the synaptic operation
count per control step,
\begin{equation}\label{eq:synops}
\mathcal S_j(t_k) \;=\; \sum_{i\in\mathcal A_j(t_k)} \mathrm{fan\text{-}out}(i),
\end{equation}
where $\mathcal A_j(t_k)$ is the set of units active at step $t_k$, so
that the estimator is comparable against a dense basis at
matched estimation accuracy.
\end{enumerate}

\begin{rmk}
\label{rem:conflict}
Position tracking does not imply force tracking. Substituting the
interaction model \eqref{eq:interaction} into \eqref{eq:frcerr} and
evaluating at exact position synchronisation with negligible delay gives
\begin{equation}\label{eq:conflict}
e_{\tau m}
 \;=\; \big(F_h^{*}-F_e^{*}\big)
 + \big(D_h-D_e\big)\dot q_m
 + \big(S_h-S_e\big) q_m ,
\end{equation}
which vanishes only if the operator and the environment present matched
impedances and matched exogenous torques. The slave-side counterpart follows identically and satisfies
$e_{\tau s}=-e_{\tau m}$ under the same idealisation. Objectives~O2 and~O3 are therefore generally competing:
$\chi_{j1}$ and $\chi_{j2}$ determine the relative emphasis placed on
position synchronisation and force feedback in the hybrid error
\eqref{eq:xi}. Their achieved tracking performance is evaluated
separately in the numerical results and is not inferred solely from the
finite-time bound on $\zeta_j$.
\end{rmk}

\begin{figure}[pos=H]
    \centering
    \includegraphics[width=\linewidth]{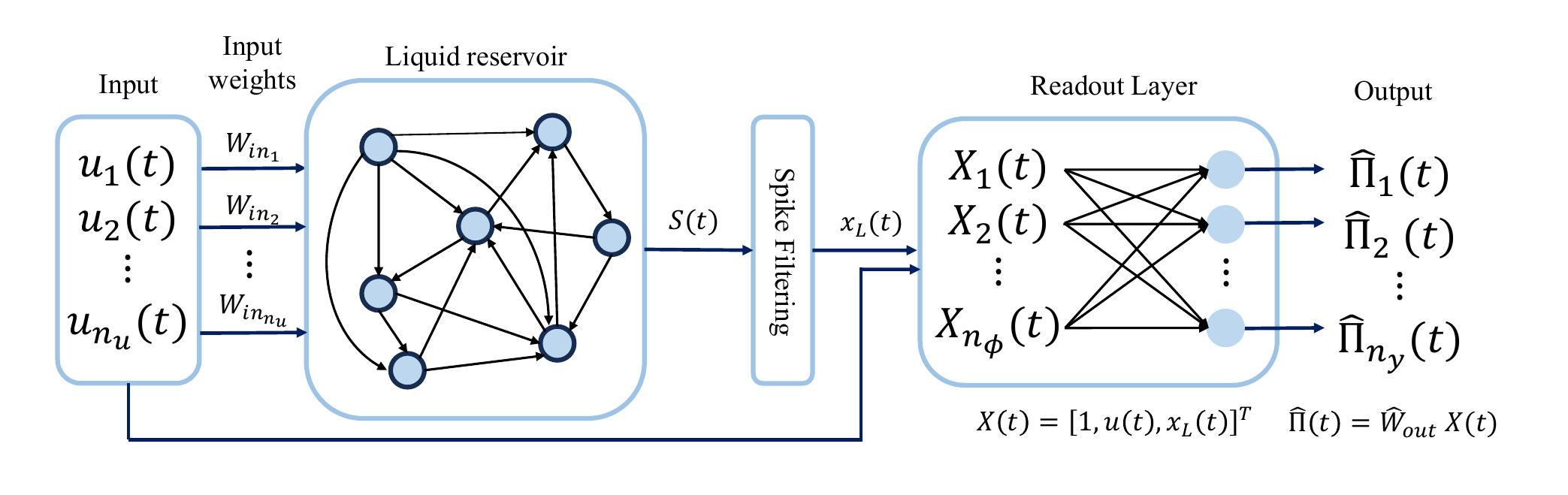}
    \caption{Structure of the LSM-based uncertainty estimator: fixed spiking
reservoir with an adapted linear readout.}
    \label{fig:lsm}
\end{figure}

\subsection{Liquid State Machine as Uncertainty Estimator}
\label{ssec:lsm}
The uncertainty $\Pi_j$ is a functional of the
closed-loop history rather than a map of its instantaneous value. It is
estimated here by a liquid state machine (LSM) \cite{NEWLSM1}: a
recurrent reservoir of spiking neurons with fixed internal connectivity
and a single adapted linear readout, instantiated independently on each
side $j\in\{m,s\}$. The estimator comprises four stages, shown in
Fig.~\ref{fig:lsm}: an input layer, which injects $u_j$ into the
reservoir through the randomly drawn weights $W_j^{\mathrm{in}}$; the
liquid itself, a sparsely and recurrently connected population of
spiking neurons whose weights $W_j^{\mathrm{rec}}$ are likewise drawn at
initialization and held fixed; a filtering layer, which converts the
binary spike trains into the continuous traces $x_{L,j}$ that a linear
map can act on; and the readout layer, the only stage adapted online,
which forms $\hat\Pi_j$ from those traces.

\subsubsection{Operator form}

Let $u_j(t)\in\mathbb{R}^{n_u}$ denote the estimator input on side $j$ and
$u_{j,t}$ its history up to time $t$. The LSM is the composition of a
\emph{liquid filter} $\mathcal L_j$, which expands the input history into
a high-dimensional transient state, and a \emph{readout} $\mathcal R_j$,
which extracts the task-relevant component of that state:
\begin{equation}\label{eq:lsmcomp}
X_j(t) = \mathcal L_j\big(u_{j,t}\big),
\qquad
\hat\Pi_j(t) = \mathcal R_j\big(X_j(t)\big).
\end{equation}
The liquid is a recurrent network of spiking neurons whose input and
recurrent weights $W_j^{\mathrm{in}},W_j^{\mathrm{rec}}$ are drawn at
initialization and held fixed, so all nonlinearity and all memory reside
in a part of the estimator that is never adapted; the readout is linear
and is the only part adapted online. Accordingly the ideal
representation of the uncertainty is written
\begin{equation}\label{eq:idealrep}
\Pi_j(t) = W_j^{*}X_j(t) + \varepsilon_j^{r}(t),
\end{equation}
where $W_j^{*}$ is a fixed ideal readout matrix and
$\varepsilon_j^{r}(t)$ is the residual, i.e. the part of $\Pi_j$ not
captured by the reservoir together with unmodelled effects.
\begin{rmk}
The readout $\mathcal R_j$ is linear in its parameters, so
\eqref{eq:idealrep} enters the control law exactly as a
linear-in-parameters basis expansion does. This is why the reservoir is
held fixed and only the readout adapted: it is what allows the adaptive
law and the Lyapunov certificate of the underlying scheme to be carried
over unchanged.
\end{rmk}

\subsubsection{Spiking liquid reservoir}

The liquid contains $N$ leaky integrate-and-fire (LIF) neurons
\cite{manna2023plsm,eshraghian2023}. The membrane potential $v_{j,i}$ of
neuron $i$ on side $j$ evolves as
\begin{equation}\label{eq:lif}
\tau^{\mathrm{mem}}\,\dot v_{j,i}(t)
 = -\big(v_{j,i}(t)-V_{\mathrm{rest}}\big) + R\,I_{j,i}(t),
\qquad i=1,\dots,N,
\end{equation}
where $\tau^{\mathrm{mem}}$ is the membrane time constant,
$V_{\mathrm{rest}}$ the resting potential, $R$ the membrane resistance,
and $I_{j,i}(t)$ the total input current of the neuron, comprising the
external drive and the recurrent synaptic contribution. When
$v_{j,i}$ reaches the firing threshold $V_{\mathrm{th}}$ the neuron emits
a spike and its potential is reset to $V_{\mathrm{reset}}$.
The reservoir is integrated at the controller step $\Delta t$. 

Firing events are represented by the bounded spike indicator
\begin{equation}\label{eq:spike}
s_{j,i}[k] \;=\; \sum_{f}\delta\bigl[k-k_{j,i,f}\bigr]
\;=\; \Theta\bigl(v_{j,i}[k]-V_{\mathrm{th}}\bigr),
\qquad
s_j[k]\in\{0,1\}^{N},
\end{equation}
where $k$ indexes the controller step $\Delta t$, $k_{j,i,f}$ are the
firing steps of neuron $i$, $\delta[\cdot]$ is the Kronecker delta and
$\Theta$ the unit step. The potential is reset to $V_{\mathrm{reset}}$ at
step $k+1$ whenever $s_{j,i}[k]=1$.
%
\begin{rmk}
\label{rem:indicator}
Spiking activity is often written as a train of Dirac impulses,
$\sum_f\delta(t-t_{i,f})$. Such a drive makes the filtered state
\eqref{eq:filter} jump by $1/\tau^{\mathrm{flt}}$ at each spike, so its
bound depends on the minimum inter-spike interval and the unit interval
is no longer invariant. The bounded indicator \eqref{eq:spike} carries
the same information at the controller sampling rate, matches the
implemented model, and keeps the filtered state in $[0,1]$ by
construction.
\end{rmk}

\subsubsection{Input Injection and Synaptic Dynamics}

The estimator input is continuous-valued and is delivered to the liquid
through \emph{current injection}: the numerical input directly modulates the
current supplied to the neurons, so no separate spike-encoding stage is
required. To guarantee bounded excitation of the reservoir, define the
saturated estimator input
\begin{equation}\label{eq:usat}
u_j^{s}(t)
=
\operatorname{sat}_{\bar u_j}\!\big(u_j(t)\big),
\end{equation}
where $\operatorname{sat}_{\bar u_j}(\cdot)$ denotes a componentwise
saturation with prescribed positive bound vector
$\bar u_j\in\mathbb{R}^{n_u}_{>0}$. Hence,
\begin{equation}\label{eq:usat-bound}
\norm{u_j^{s}(t)}
\le
\norm{\bar u_j},
\qquad \forall t\ge0 .
\end{equation}
With a first-order current-based synapse model, the synaptic current vector
obeys
\begin{equation}\label{eq:current}
\tau^{\mathrm{syn}}\,\dot I_j(t)
=
-I_j(t)
+
W_j^{\mathrm{in}}u_j^{s}(t)
+
W_j^{\mathrm{rec}}s_j\!\left(t-T_{\mathrm{syn}}\right),
\end{equation}
where $\tau^{\mathrm{syn}}$ is the synaptic time constant,
$W_j^{\mathrm{in}}\in\mathbb{R}^{N\times n_u}$ is the input weight matrix,
$W_j^{\mathrm{rec}}\in\mathbb{R}^{N\times N}$ is the recurrent weight matrix,
and $T_{\mathrm{syn}}=\Delta t$ is the recurrent transmission delay,
corresponding to one integration step. The saturation in
\eqref{eq:usat} is applied only to the copy of $u_j$ entering the LSM and
does not alter the plant dynamics or the definitions of the control
regressors.

\subsubsection{Spike Filtering and Feature Vector}

Binary spike signals are not directly suitable for a linear readout, so each
spike channel is passed through a first-order low-pass filter:
\begin{equation}\label{eq:filter}
\tau^{\mathrm{flt}}\,\dot x_{j,i}(t)
=
-x_{j,i}(t)
+
s_{j,i}(t),
\qquad
x_{L,j}(t)
=
\big[
x_{j,1}(t)\ \cdots\ x_{j,N}(t)
\big]^{\mathrm T},
\end{equation}
where $\tau^{\mathrm{flt}}$ is the filter time constant and
$x_{L,j}\in\mathbb{R}^N$ collects the filtered spike states of all $N$ reservoir
neurons.

The linear readout acts on the augmented feature vector
\begin{equation}\label{eq:feature}
X_j(t)
=
\Big[
\,1\quad
\big(u_j^{s}(t)\big)^{\mathrm T}\quad
x_{L,j}^{\mathrm T}(t)
\,\Big]^{\mathrm T}
\in\mathbb{R}^{n_X},
\qquad
n_X=1+n_u+N.
\end{equation}
The leading entry provides a bias term, while $u_j^{s}$ supplies a bounded
direct input-to-output feedthrough path alongside the recurrent reservoir
features.

Since $s_{j,i}(t)\in[0,1]$, Lemma~\ref{lem:leaky} implies
$x_{j,i}(t)\in[0,1]$ for all $i=1,\dots,N$, and therefore
\begin{equation}\label{eq:xLbound}
\norm{x_{L,j}(t)}^2
=
\sum_{i=1}^{N}x_{j,i}^2(t)
\le
N.
\end{equation}
Combining \eqref{eq:usat-bound} and \eqref{eq:xLbound}, the augmented
feature vector satisfies the uniform bound
\begin{equation}\label{eq:Xbound}
\norm{X_j(t)}
\le
\sqrt{
1+\norm{\bar u_j}^{2}+N
}
=: \bar X_j,
\qquad
\forall t\ge0 .
\end{equation}
Thus the feature vector used by the adaptive readout is bounded by
construction, independently of any compactness assumption on the
closed-loop trajectory.
\begin{assumption}
\label{as:approx}
For each $j\in\{m,s\}$ and for the compact admissible operating domain
$\Omega_j$ of Assumption~\ref{as:compact}, there exist a reservoir size
$N$ and a constant readout
$W_j^{*}\in\mathbb{R}^{n\times n_X}$ such that the residual in
\eqref{eq:idealrep} satisfies
\[
\norm{\varepsilon_j^{r}(t)}\le \bar\varepsilon_j
\]
for all admissible input histories associated with $\Omega_j$.
\end{assumption}

Assumption~\ref{as:approx} is the standard approximation hypothesis for
reservoir estimators and is imposed on the compact admissible operating
domain $\Omega_j$ introduced in Assumption~\ref{as:compact}. Only boundedness of $\varepsilon_j^{r}$ is used in
Section~\ref{sec:lsm-stability}; the tolerance enters through the
size of the residual set.

\subsubsection{Adapted linear readout}

The uncertainty estimate is produced by the linear map
\begin{equation}\label{eq:readout}
\hat\Pi_j(t) = \hat W_j(t)\,X_j(t),
\qquad
\hat W_j\in\mathbb{R}^{n\times n_X},
\end{equation}
where $\hat W_j$ is the adapted readout matrix, i.e. the estimate of
$W_j^{*}$ in \eqref{eq:idealrep}. The matrices
$W_j^{\mathrm{in}}$ and $W_j^{\mathrm{rec}}$ are drawn randomly at
initialization, with $W_j^{\mathrm{rec}}$ sparse and scaled to a
prescribed spectral radius, and are held fixed throughout operation.
The update law for $\hat W_j$ is not postulated but follows from the
Lyapunov analysis of Section~\ref{sec:lsm-stability}.

\begin{rmk}
\label{rem:vsrbf}
A feedforward radial-basis network approximates a static map of its
present input, so representing a history-dependent uncertainty requires
that the history be supplied explicitly, as additional lagged copies of
the input whose lags are fixed at design time. The liquid instead embeds the input history in its recurrent
state, whose fading memory furnishes a nonlinear temporal feature
representation. The dynamic uncertainty can therefore be recovered by a
purely linear adaptive readout.
\end{rmk}

\begin{rmk}
\label{rem:sparsity}
Spiking neurons communicate only when their membrane potential crosses
$V_{\mathrm{th}}$, so in general only a subset of the liquid fires at any
step and $s_j(t)$ is sparse. On hardware that exploits this sparsity,
the number of synaptic operations and the associated data movement scale
with the count of active units rather than with $N$, which is the basis
of the per-step cost measure $\mathcal S_j$ introduced in
\eqref{eq:synops} and evaluated in Section~\ref{sec:sim}. Sparsity here
refers to per-step activity, not to a reduction in network size.
\end{rmk}

\subsection{Preliminary Lemmas}
\label{ssec:lemmas}

Three lemmas are used in the stability and finite-time analysis of
Section~\ref{sec:lsm-stability}. 

\begin{lemma}[Power-sum inequality]\label{lem:powersum}
For any $a_1,a_2,\dots,a_k\in\mathbb{R}$ and any $p\in(0,1)$,
\begin{equation}\label{eq:powersum}
\big(|a_1|+|a_2|+\cdots+|a_k|\big)^{p}
 \;\le\; |a_1|^{p}+|a_2|^{p}+\cdots+|a_k|^{p}.
\end{equation}
\end{lemma}

\begin{lemma}[Finite-time and practical finite-time convergence]
\label{lem:finitetime}
Consider $\dot x=f(x)$ with $x\in\mathbb{R}^{n}$ and $x(0)=x_0$, and let
$\mathcal V:\mathbb{R}^{n}\to\mathbb{R}_{\ge0}$ be continuously differentiable and
positive definite. Let $\alpha,\beta>0$ and $\sigma\in(0,1)$.
\begin{enumerate}
\item[(i)] If
\begin{equation}\label{eq:ftineq}
\dot{\mathcal V}(x)\;\le\;-\alpha\,\mathcal V(x)-\beta\,\mathcal V^{\sigma}(x)
\end{equation}
holds for all $t\ge0$, then $\mathcal V$ reaches the origin in finite
time, and the settling time satisfies
\begin{equation}\label{eq:tfin}
T_{\mathrm{fin}}
 \;\le\;
 \frac{1}{\alpha(1-\sigma)}
 \ln\!\frac{\alpha\,\mathcal V^{1-\sigma}(x_0)+\beta}{\beta}.
\end{equation}
\item[(ii)] If \eqref{eq:ftineq} holds only while $x\notin\mathcal D$ for
some compact set $\mathcal D$ containing the origin, then the trajectory
reaches $\mathcal D$ within the time \eqref{eq:tfin}; behaviour after
that instant is not asserted by this lemma.
\end{enumerate}
\end{lemma}

\begin{pf}
For (i), set $\mathcal W=\mathcal V^{1-\sigma}$. Then
$\dot{\mathcal W}=(1-\sigma)\mathcal V^{-\sigma}\dot{\mathcal V}$, and
substituting \eqref{eq:ftineq} gives
\begin{equation}\label{eq:wineq}
\dot{\mathcal W}\;\le\;-(1-\sigma)\big(\alpha\mathcal W+\beta\big),
\end{equation}
a linear differential inequality. Integrating,
\begin{equation}\label{eq:wsol}
\mathcal W(t)\;\le\;
\Big(\mathcal W(0)+\tfrac{\beta}{\alpha}\Big)
 e^{-\alpha(1-\sigma)t}-\tfrac{\beta}{\alpha},
\end{equation}
and the right-hand side vanishes exactly at the time given in
\eqref{eq:tfin}, which therefore upper bounds the instant at which
$\mathcal W$, and hence $\mathcal V$, reaches zero. Part (ii) follows by
applying the same estimate on the maximal interval during which the
trajectory remains outside $\mathcal D$.
\end{pf}

The filtered spike channels in \eqref{eq:filter} are driven by a
discontinuous but bounded signal. The following elementary bound confines
them to the unit interval.

\begin{lemma}[Boundedness of a leaky integrator with bounded drive]
\label{lem:leaky}
Let $z:\mathbb{R}_{\ge0}\to\mathbb{R}$ be absolutely continuous and satisfy
\begin{equation}\label{eq:leaky}
\tau\,\dot z(t) = -z(t)+w(t)
\qquad\text{for almost every } t\ge0,
\end{equation}
with $\tau>0$ and $w$ measurable and bounded, $0\le w(t)\le\bar w$. Then
\begin{equation}\label{eq:leakybound}
0\;\le\; z(t)\;\le\;\max\{z(0),\,\bar w\},
\qquad \forall t\ge0,
\end{equation}
provided $z(0)\ge0$. In particular, if $z(0)\in[0,\bar w]$ then
$z(t)\in[0,\bar w]$ for all $t\ge0$, so $[0,\bar w]$ is positively
invariant.
\end{lemma}

\begin{pf}
Variation of constants applied to \eqref{eq:leaky} gives
\begin{equation}\label{eq:leakysol}
z(t)=z(0)e^{-t/\tau}
 +\frac{1}{\tau}\int_{0}^{t}e^{-(t-r)/\tau}\,w(r)\,dr .
\end{equation}
The kernel is non-negative and satisfies
$\tau^{-1}\!\int_{0}^{t}e^{-(t-r)/\tau}dr=1-e^{-t/\tau}$. Using
$0\le w\le\bar w$ in \eqref{eq:leakysol} therefore yields
$z(t)\ge0$ and
$z(t)\le z(0)e^{-t/\tau}+\bar w\big(1-e^{-t/\tau}\big)$, which is a
convex combination of $z(0)$ and $\bar w$ and hence bounded by their
maximum.
\end{pf}

\section{Control Design}
\label{sec:control}

\begin{figure}[pos=H]
    \centering
    \includegraphics[width=\linewidth]{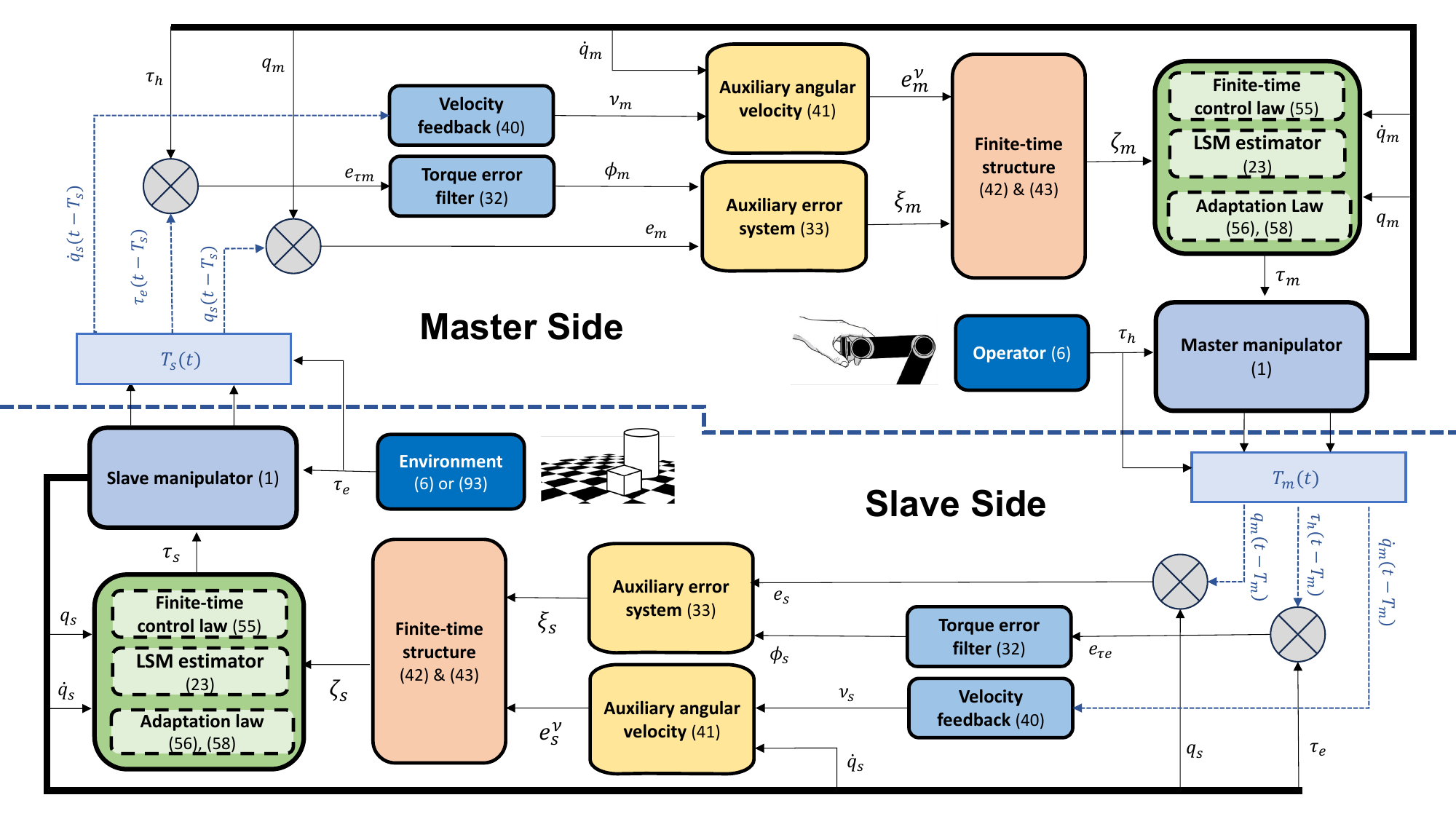}
    \caption{Overall architecture of the proposed bilateral teleoperation
    system and the LSM-based adaptive control scheme.}
    \label{fig:main_diagram}
\end{figure}

The design proceeds in four steps. The overall architecture of the proposed control scheme is depicted in Fig.~\ref{fig:main_diagram}. Section~\ref{ssec:hybrid-error}
combines the position and force tracking errors into one hybrid error
per side; Section~\ref{ssec:aux-var} augments it with fractional-power
terms to form the integral auxiliary variable;
Section~\ref{ssec:estimation} reconstructs the lumped uncertainty
$\Pi_j$ by the reservoir of Section~\ref{ssec:lsm}; and Section~\ref{ssec:control-law} states the control
law and the adaptive laws. The signals defined in
Section~\ref{ssec:hybrid-error} enter the estimator input $u_j$ through
the shorthand $\pi_{j1}$--$\pi_{j3}$.
\subsection{Hybrid Position--Force Error and Delay-Consistent Velocity
Errors}
\label{ssec:hybrid-error}

Objectives~O2 and~O3 are stated on quantities of different physical
dimension, the position error $e_j$ in \eqref{eq:poserr} and the force
error $e_{\tau j}$ in \eqref{eq:frcerr}, and are combined here into a
single error per side. The design differentiates that combined error, so using $e_{\tau j}$ in it directly would introduce $\dot e_{\tau j}$,
hence the rates of the operator and environment torques, which are neither
modelled in \eqref{eq:interaction} nor measured. The force channel is
therefore admitted through the state of the filter  
\begin{equation}
\dot\phi_j = -L_j^{\tau}\phi_j + L_j^{\tau}e_{\tau j},
\qquad \phi_j(0)=0,
\qquad j\in\{m,s\},
\label{eq:tau-filter}
\end{equation}
with $L_j^{\tau}\in\mathbb{R}^{n\times n}$ symmetric positive definite.
Define the hybrid position--force error
\begin{equation}
\xi_j = \chi_{j1}e_j + \chi_{j2}\phi_j ,
\label{eq:xi}
\end{equation}
where $\chi_{j1},\chi_{j2}\in\mathbb{R}^{n\times n}$ are positive
definite and diagonal. Then
\begin{equation}
\dot\xi_j = \chi_{j1}\dot e_j + \chi_{j2}\dot\phi_j ,
\end{equation}
in which $\dot\phi_j$ is supplied algebraically by
\eqref{eq:tau-filter}. 
The filter \eqref{eq:tau-filter} has unit dc gain, so a constant
force error reaches $\xi_j$ unattenuated.
Differentiating \eqref{eq:poserr} gives
\begin{equation}
\dot e_j = \dot q_j(t)
- \bigl(1-\dot T_{j'}(t)\bigr)\,
\dot q_{j'}\bigl(t-T_{j'}(t)\bigr),
\label{eq:edot-raw}
\end{equation}
which contains the delay rate $\dot T_{j'}$. This quantity is not
available online in an asymmetric communication channel. Define the
delay-consistent velocity error
\begin{equation}
e_{jv} = \dot q_j(t) - \dot q_{j'}\bigl(t-T_{j'}(t)\bigr),
\label{eq:ejv}
\end{equation}
formed from the same two signals as $e_j$. Substituting \eqref{eq:ejv}
into \eqref{eq:edot-raw} gives the exact decomposition
\begin{align}
\dot e_j &= e_{jv}
+ \dot T_{j'}\,\dot q_{j'}(t-T_{j'}),
\label{eq:edot-split}\\
\dot\xi_j &= \chi_{j1}e_{jv} + \chi_{j2}\dot\phi_j
+ \chi_{j1}\,\dot T_{j'}\,\dot q_{j'}(t-T_{j'}),
\label{eq:xidot-split}
\end{align}
and eliminating $\dot\phi_j$ by \eqref{eq:tau-filter},
\begin{equation}
\dot\xi_j = \chi_{j1}e_{jv}
+ \chi_{j2}L_j^{\tau}\bigl(e_{\tau j}-\phi_j\bigr)
+ \chi_{j1}\,\dot T_{j'}\,\dot q_{j'}(t-T_{j'}),
\label{eq:xidot-measurable}
\end{equation}
in which only the last term is unavailable for feedback. Its unknown
factor $\dot T_{j'}$ is bounded by $\delta_{j'}$ under
Assumption~\ref{as:delay}; the term is retained and carried into $\Pi_j$
in Section~\ref{ssec:estimation}, where its magnitude is dominated
adaptively.

%
%
%
%

\subsection{Integral Auxiliary Variable with Finite-Time Terms}
\label{ssec:aux-var}

The hybrid error \eqref{eq:xi} is defined on positions and interaction
torques, whereas the control torque enters the dynamics \eqref{eq:dyn}
through $\ddot q_j$. The variable on which the controller acts must
therefore involve the joint velocity, and must do so without requiring
either the acceleration or the delay rate at implementation. To that end
define the velocity feedback filter
\begin{equation}
\dot\nu_j = -L_j^{v}\nu_j
+ L_j^{v}\,\dot q_{j'}\bigl(t-T_{j'}(t)\bigr),
\qquad \nu_j(0)=0,
\label{eq:v-filter}
\end{equation}
with $L_j^{v}\in\mathbb{R}^{n\times n}$ symmetric positive definite, and
the filtered velocity error
\begin{equation}
e_j^{\nu} = \dot q_j - \nu_j .
\label{eq:enu}
\end{equation}
Since \eqref{eq:v-filter} is a stable low-pass with unit dc gain,
$e_j^{\nu}$ is a smoothed counterpart of the delay-consistent velocity
error \eqref{eq:ejv}, and coincides with it in steady state.

The finite-time term is
\begin{equation}
\psi_j = \lambda_{j1}\,\mathrm{sig}(\xi_j)^{\sigma_{j1}}
+ \lambda_{j2}\!\int_{0}^{t}\!
\mathrm{sig}\bigl(\xi_j(\varsigma)\bigr)^{\sigma_{j2}}\,
\mathrm{d}\varsigma ,
\label{eq:psi}
\end{equation}
where $\lambda_{j1},\lambda_{j2}\in\mathbb{R}^{n\times n}$ are positive
definite and diagonal, $1<\sigma_{j1}<2$ and $0<\sigma_{j2}<1$. For
$x\in\mathbb{R}^{n}$ and $r>0$,
$\mathrm{sig}(x)^{r}=[\,|x_1|^{r}\mathrm{sign}(x_1),\dots,
|x_n|^{r}\mathrm{sign}(x_n)\,]^{\mathsf T}$ and
$\mathrm{diag}(|x|)^{r}=\mathrm{diag}(|x_1|^{r},\dots,|x_n|^{r})$. The
integral auxiliary variable is then
\begin{equation}
\zeta_j = e_j^{\nu} + \psi_j ,
\qquad j\in\{m,s\},
\label{eq:zeta}
\end{equation}
and the controller of Section~\ref{ssec:control-law} is designed on
$\zeta_m$ and $\zeta_s$.

Differentiating \eqref{eq:zeta},
\begin{equation}
\dot\zeta_j = \dot e_j^{\nu} + \dot\psi_j ,
\label{eq:zetadot}
\end{equation}
with
\begin{align}
\dot e_j^{\nu} &= \ddot q_j - \dot\nu_j ,
\label{eq:enudot}\\
\dot\psi_j &= \sigma_{j1}\lambda_{j1}\,
\mathrm{diag}\bigl(|\xi_j|\bigr)^{\sigma_{j1}-1}\dot\xi_j
+ \lambda_{j2}\,\mathrm{sig}(\xi_j)^{\sigma_{j2}} ,
\label{eq:psidot}
\end{align}
in which $\dot\nu_j$ follows from \eqref{eq:v-filter} and $\dot\xi_j$
from \eqref{eq:xidot-split}. The acceleration in \eqref{eq:enudot} is an
intermediate only: $\zeta_j$ enters the analysis premultiplied by the
inertia matrix, and $M_j\ddot q_j$ is replaced by the dynamics
\eqref{eq:dyn} when $M_j\dot\zeta_j$ is formed in
Section~\ref{ssec:estimation}, so that no acceleration appears in the
control law or in the estimator input.

\begin{rmk}
\label{rem:aux-comparison}
Finite-time designs for teleoperation commonly take the auxiliary
variable as $\zeta_j=\dot e_j+\lambda_j\beta_j(e_j)$ with $\beta_j$ a
fractional-power function \cite{REF-FT-1,REF-FT-2}. Beyond requiring
$\ddot q_j$, that choice carries the delay rate into the design through
$\dot e_j$, since $\dot e_j$ depends on $\dot T_{j'}$ by
\eqref{eq:edot-raw}. The filter \eqref{eq:v-filter} avoids both: it is
driven by the delayed velocity itself, which is received over
the channel.
\end{rmk}

\begin{rmk}
\label{rem:no-singularity}
The two fractional powers occupy structurally different positions in
\eqref{eq:psi}. Because $\sigma_{j1}>1$, the matrix
$\mathrm{diag}(|\xi_j|)^{\sigma_{j1}-1}$ in \eqref{eq:psidot} is
continuous and vanishes at $\xi_j=0$, so no negative power arises when
$\psi_j$ is differentiated. The exponent $\sigma_{j2}<1$, which supplies
the finite-time drive, is placed under the integral in \eqref{eq:psi}
and is therefore never differentiated; it enters \eqref{eq:psidot}
undifferentiated as $\mathrm{sig}(\xi_j)^{\sigma_{j2}}$. The control law
of Section~\ref{ssec:control-law} is consequently free of the
singularity associated with terminal sliding-mode surfaces.
\end{rmk}

\subsection{Reservoir-Based Reconstruction of the Lumped Uncertainty}
\label{ssec:estimation}

Premultiplying \eqref{eq:zetadot} by the inertia matrix and using
\eqref{eq:enudot} and \eqref{eq:psidot},
\begin{equation}
M_j\dot\zeta_j = M_j\ddot q_j
+ M_j\Bigl[-\dot\nu_j
+ \sigma_{j1}\lambda_{j1}\mathrm{diag}\bigl(|\xi_j|\bigr)^{\sigma_{j1}-1}
\dot\xi_j
+ \lambda_{j2}\mathrm{sig}(\xi_j)^{\sigma_{j2}}\Bigr].
\label{eq:Mzeta-raw}
\end{equation}
Substituting $\dot\xi_j$ from \eqref{eq:xidot-split} separates the
available part of $\dot\xi_j$ from the delay-rate term. With the
shorthand
\begin{align}
\pi_{j1} &= -\dot\nu_j
+ \sigma_{j1}\lambda_{j1}
\mathrm{diag}\bigl(|\xi_j|\bigr)^{\sigma_{j1}-1}
\bigl(\chi_{j1}e_{jv}+\chi_{j2}\dot\phi_j\bigr)
\nonumber\\
&\quad + \lambda_{j2}\mathrm{sig}(\xi_j)^{\sigma_{j2}},
\label{eq:pi1}\\
\pi_{j2} &= \sigma_{j1}\lambda_{j1}
\mathrm{diag}\bigl(|\xi_j|\bigr)^{\sigma_{j1}-1}\chi_{j1}\,
\dot q_{j'}\bigl(t-T_{j'}\bigr),
\label{eq:pi2}\\
\pi_{j3} &= \zeta_j-\dot q_j ,
\label{eq:pi3}
\end{align}
equation \eqref{eq:Mzeta-raw} becomes
$M_j\dot\zeta_j = M_j\ddot q_j + M_j\pi_{j1}
+ M_j\pi_{j2}\dot T_{j'}$. Eliminating $M_j\ddot q_j$ by the dynamics
\eqref{eq:dyn}, applying the nominal--uncertain split \eqref{eq:split}
and the interaction model \eqref{eq:interaction}, and writing
$-C_j\dot q_j = C_j\pi_{j3}-C_j\zeta_j$ through \eqref{eq:pi3} gives
\begin{equation}
M_j\dot\zeta_j = \tau_j - F^{*}_{he,j}
+ M_j^{0}\pi_{j1} + M_j^{0}\pi_{j2}\dot T_{j'}
+ C_j^{0}\pi_{j3} - G_j^{0} + \Pi_j - C_j\zeta_j ,
\label{eq:Mzeta}
\end{equation}
in which $F^{*}_{he,j}$, $D_{he,j}$ and $\ S_{he,j}$ denote
$F^{*}_{h},D_h,\ S_h$ for $j=m$ and $F^{*}_{e},D_e,\ S_e$ for
$j=s$, and the lumped uncertainty is
\begin{equation}
\Pi_j = \Delta M_j\pi_{j1} + \Delta M_j\pi_{j2}\dot T_{j'}
+ \Delta C_j\pi_{j3} - \Delta G_j - B_j
- D_{he,j}\dot q_j - \ S_{he,j}q_j .
\label{eq:Pi}
\end{equation}
The delay rate enters \eqref{eq:Mzeta} twice, and the two occurrences
are treated differently. In $M_j^{0}\pi_{j2}\dot T_{j'}$ both factors
$M_j^{0}$ and $\pi_{j2}$ are known and only the scalar $\dot T_{j'}$ is
not; this term is retained explicitly and dominated in
Section~\ref{ssec:control-law} by an adaptive estimate of the bound
$\delta_{j'}$ of Assumption~\ref{as:delay}. In
$\Delta M_j\pi_{j2}\dot T_{j'}$ the multiplying matrix is itself
unknown, so the product is absorbed into $\Pi_j$ and reconstructed
together with the remaining uncertain terms.

The estimator input is taken as
\begin{equation}
u_j = \bigl[\pi_{j1}^{\mathsf T},\ \pi_{j2}^{\mathsf T},\
\pi_{j3}^{\mathsf T},\ q_j^{\mathsf T},\
\dot q_j^{\mathsf T}\bigr]^{\mathsf T}\in\mathbb{R}^{n_u},
\qquad n_u = 5n ,
\label{eq:uinput}
\end{equation}
every block of which is formed from filter states, measured positions
and velocities, and the delayed velocity received over the channel; no
acceleration and no delay rate appear. Substituting the split
\eqref{eq:split} directly into the dynamics \eqref{eq:dyn} instead of
forming \eqref{eq:Mzeta} would place $\Delta M_j\ddot q_j$ in the
uncertain term and require the joint acceleration as a regressor. With
$X_j(t)\in
\mathbb{R}^{n_X}$ the augmented reservoir feature vector of
Section~\ref{ssec:lsm}, the uncertainty is parameterised as
\begin{equation}
\Pi_j(t) = W_j^{*}X_j(t) + \varepsilon_j^{r},
\label{eq:approx}
\end{equation}
where $W_j^{*}\in\mathbb{R}^{n\times n_X}$ is the ideal readout matrix
and $\varepsilon_j^{r}(t)\in\mathbb{R}^{n}$ the reconstruction residual,
bounded on the considered compact admissible operating domain by Assumption~\ref{as:approx}. Only the readout is
adapted, through the estimate $\hat W_j$ of \eqref{eq:readout}; the
input and recurrent weights of the reservoir are drawn once and held
fixed, so \eqref{eq:approx} is linear in the adapted parameters.
With Assumptions~\ref{as:exo} and~\ref{as:approx}, there exists a
constant $\omega_j^{r}>0$ such that
$\lVert \varepsilon_j^{r}(t)-F^{*}_{he,j}(t)\rVert^{2}
\le\omega_j^{r}$, which is the only property of the combined
reconstruction residual and exogenous interaction term used in
Section~\ref{sec:lsm-stability}.

\begin{rmk}
\label{rem:structural-fit}
The quantity to be reconstructed is not a static map of $u_j$. Both
$\pi_{j2}$ and the products $\Delta M_j\pi_{j2}\dot T_{j'}$ in
\eqref{eq:Pi} depend on the delayed velocity
$\dot q_{j'}(t-T_{j'})$ and on the rate at which the delay itself
varies, so the value of $\Pi_j$ at time $t$ is determined by the
closed-loop trajectory over $[t-\bar T_{j'},t]$ rather than by $u_j(t)$
alone; this is the property established in
Remark~\ref{rem:history}. A kernel expansion over the instantaneous
input can represent such a functional only to the extent that the
required history is supplied to it explicitly, as hand-selected lag
taps. The reservoir supplies it instead through its own state, whose
leaky dynamics retain a fading trace of past inputs. No additional lag
taps are introduced in \eqref{eq:uinput}: the delayed regressor it
contains is the one the communication channel already places in the
closed loop.
\end{rmk}

\section{Finite-Time Control Law and Stability Analysis}
\label{sec:lsm-stability}

The control law is stated first, then the closed-loop certificate. The
two are presented together because the adaptive laws are not postulated:
each is read off the Lyapunov derivative as the choice that cancels an
otherwise indefinite term.

\subsection{Control and Adaptive Laws}
\label{ssec:control-law}

For $j\in\{m,s\}$ and $j'$ the opposite side, the control torque is
\begin{equation}
\begin{aligned}
\tau_j = {}& -M_j^{0}\pi_{j1} - C_j^{0}\pi_{j3} + G_j^{0}
 - \hat W_j X_j(t)\\
 &- \frac{\hat\delta_{j'}\bigl\lVert M_j^{0}\pi_{j2}\bigr\rVert^{2}}
        {2a_{j1}^{2}}\,\zeta_j
  - \frac{\hat\omega_j^{r}}{2a_{j2}^{2}}\,\zeta_j
  - K_{j1}\zeta_j - K_{j2}\,\mathrm{sig}(\zeta_j)^{\sigma},
\end{aligned}
\label{eq:control-law}
\end{equation}
where $K_{j1},K_{j2}\in\mathbb{R}^{n\times n}$ are positive definite and
diagonal, $a_{j1},a_{j2}>0$, and $\sigma\in(0,1)$. The first line
cancels the nominal feedforward of \eqref{eq:Mzeta} and compensates the
lumped uncertainty through the reservoir estimate \eqref{eq:readout};
the second dominates the two unmeasured quantities and supplies the
asymptotic and finite-time feedback terms.

The adaptive laws are
\begin{align}
\dot{\hat W}_j
 &= \gamma_j^{W}\Bigl(\zeta_j X_j^{\mathsf T}(t) - \rho_j\hat W_j\Bigr),
\label{eq:adapt-W}\\[2pt]
\dot{\hat\delta}_{j'}
 &= \gamma_j^{\delta}\left(
    \frac{\zeta_j^{\mathsf T}\zeta_j
          \bigl\lVert M_j^{0}\pi_{j2}\bigr\rVert^{2}}{2a_{j1}^{2}}
    - \rho_j\hat\delta_{j'}\right),
\label{eq:adapt-delta}\\[2pt]
\dot{\hat\omega}_j^{r}
 &= \gamma_j^{\omega}\left(
    \frac{\zeta_j^{\mathsf T}\zeta_j}{2a_{j2}^{2}}
    - \rho_j\hat\omega_j^{r}\right),
\label{eq:adapt-omega}
\end{align}
with $\gamma_j^{W},\gamma_j^{\delta},\gamma_j^{\omega}>0$ the adaptation
gains and $\rho_j>0$ a leakage coefficient. Here $\hat\delta_{j'}$
estimates the delay-rate bound of Assumption~\ref{as:delay} and
$\hat\omega_j^{r}$ the residual bound introduced in
Section~\ref{ssec:estimation}.

\begin{rmk}
\label{rem:readout-only}
Only $\hat W_j$ is adapted in \eqref{eq:adapt-W}; $W_j^{*}$ is the fixed
unknown of Assumption~\ref{as:approx}, and the reservoir weights
$W_j^{\mathrm{in}},W_j^{\mathrm{rec}}$ are drawn once and never updated.
The adapted quantity therefore enters \eqref{eq:control-law} linearly,
exactly as a weight vector over a fixed basis would, which is what
allows the certificate below to retain the structure of the
non-spiking scheme. The gain in \eqref{eq:adapt-W} is scalar rather than
matrix valued because $\hat W_j\in\mathbb{R}^{n\times n_X}$ with
$n_X=1+n_u+N$, so a full gain matrix would carry $n_X^{2}$ parameters
for no additional freedom in the proof.
\end{rmk}

\subsection{Closed-Loop Stability and Finite-Time Reaching}
\label{ssec:theorem}

\begin{theorem}
\label{thm:main}
Consider the teleoperator \eqref{eq:dyn} with the uncertainty split
\eqref{eq:split}, the interaction model \eqref{eq:interaction} and the
asymmetric delays of Assumption~\ref{as:delay}, driven by the control
law \eqref{eq:control-law} and the adaptive laws
\eqref{eq:adapt-W}--\eqref{eq:adapt-omega}, with the torque-error filter
\eqref{eq:tau-filter}, hybrid error \eqref{eq:xi}, velocity filter
\eqref{eq:v-filter} and auxiliary variable \eqref{eq:zeta}, and with the
lumped uncertainty reconstructed by the reservoir of
Section~\ref{ssec:lsm}. Let Assumptions~\ref{as:exo}--\ref{as:approx}
hold, let $K_{j1},K_{j2}$ be positive definite and diagonal, let
$a_{j1},a_{j2},\gamma_j^{W},\gamma_j^{\delta},\gamma_j^{\omega},\rho_j>0$
and $\sigma\in(0,1)$. Then $\zeta_j$, $\tilde W_j$, $\tilde\delta_{j'}$, and
$\tilde\omega_j^r$, $j\in\{m,s\}$, are uniformly ultimately bounded,
and $\zeta_m$ and $\zeta_s$ reach the residual set $\mathcal D$ defined
in \eqref{eq:residual-set} within the finite time $T_{\mathrm{fin}}$
bounded by \eqref{eq:Tfin}.
\end{theorem}

\begin{pf}
Take the Lyapunov--Krasovskii functional
\begin{equation}
\mathcal V = \mathcal V_1 + \mathcal V_2 + \mathcal V_3 ,
\label{eq:lyap}
\end{equation}
where, summing over $j\in\{m,s\}$ with $j'$ the opposite side,
\begin{align}
\mathcal V_1 &= \sum_{j}\tfrac12\,\zeta_j^{\mathsf T}M_j\zeta_j ,
\label{eq:V1}\\
\mathcal V_2 &= \sum_{j}\left[
 \frac{\mathrm{tr}\bigl(\tilde W_j\tilde W_j^{\mathsf T}\bigr)}
      {2\gamma_j^{W}}
 + \frac{\tilde\delta_{j'}^{2}}{2\gamma_j^{\delta}}
 + \frac{\bigl(\tilde\omega_j^{r}\bigr)^{2}}{2\gamma_j^{\omega}}
 \right],
\label{eq:V2}\\
\mathcal V_3 &= \sum_{j}\eta_j\!\int_{t-\bar T_j}^{t}\!
 \frac{\bar T_j-t+\alpha}{\bar T_j}\,
 \zeta_j^{\mathsf T}(\alpha)\zeta_j(\alpha)\,\mathrm d\alpha ,
\label{eq:V3}
\end{align}
with $\tilde W_j = W_j^{*}-\hat W_j$,
$\tilde\delta_{j'} = \delta_{j'}-\hat\delta_{j'}$,
$\tilde\omega_j^{r} = \omega_j^{r}-\hat\omega_j^{r}$, and
$\eta_j\ge0$ constant. Each term is positive definite:
$\mathcal V_1$ by Property~\ref{prop:inertia}, $\mathcal V_2$ by the
positivity of the adaptation gains, $\mathcal V_3$ by construction.

Differentiating \eqref{eq:V1},
\begin{equation}
\dot{\mathcal V}_1
 = \sum_{j}\Bigl(
   \zeta_j^{\mathsf T}M_j\dot\zeta_j
 + \tfrac12\,\zeta_j^{\mathsf T}\dot M_j\zeta_j\Bigr).
\label{eq:V1dot-raw}
\end{equation}
Substituting $M_j\dot\zeta_j$ from \eqref{eq:Mzeta}, the term
$-\zeta_j^{\mathsf T}C_j\zeta_j$ combines with
$\tfrac12\zeta_j^{\mathsf T}\dot M_j\zeta_j$ into
$\tfrac12\zeta_j^{\mathsf T}(\dot M_j-2C_j)\zeta_j$, which vanishes by
Property~\ref{prop:skew}. Replacing $\Pi_j$ by its reservoir
representation \eqref{eq:approx},
\begin{equation}
\begin{aligned}
\dot{\mathcal V}_1 = \sum_{j}\zeta_j^{\mathsf T}\Bigl[&
 \tau_j + M_j^{0}\pi_{j1} + C_j^{0}\pi_{j3} - G_j^{0}
 + W_j^{*}X_j(t)\\
 &+ M_j^{0}\pi_{j2}\dot T_{j'}
 + \varepsilon_j^{r} - F^{*}_{he,j}\Bigr].
\end{aligned}
\label{eq:V1dot}
\end{equation}
The last two groups contain the only quantities not available to the
controller. Bounding them by the Cauchy--Schwarz inequality,
\begin{equation}
\begin{aligned}
\dot{\mathcal V}_1 \le \sum_{j}\Bigl\{&
 \zeta_j^{\mathsf T}\bigl[\tau_j + M_j^{0}\pi_{j1}
 + C_j^{0}\pi_{j3} - G_j^{0} + W_j^{*}X_j(t)\bigr]\\
 &+ \bigl\lVert\zeta_j\bigr\rVert
    \bigl\lVert M_j^{0}\pi_{j2}\bigr\rVert
    \bigl|\dot T_{j'}\bigr|
  + \bigl\lVert\zeta_j\bigr\rVert
    \bigl\lVert \varepsilon_j^{r}-F^{*}_{he,j}\bigr\rVert
 \Bigr\}.
\end{aligned}
\label{eq:V1dot-bound}
\end{equation}
By Assumption~\ref{as:delay} the delay rate satisfies
$|\dot T_{j'}|\le\delta_{j'}$, and by
Assumptions~\ref{as:exo} and~\ref{as:approx} the combined residual
satisfies
$\lVert\varepsilon_j^{r}-F^{*}_{he,j}\rVert\le\sqrt{\omega_j^{r}}$.
Applying Young's inequality in the form
$e\le e^{2}/2p^{2}+p^{2}/2$ to each of the two remaining products, with
$p=a_{j1}$ and $p=a_{j2}$ respectively,
\begin{equation}
\begin{aligned}
\dot{\mathcal V}_1 \le \sum_{j}\Bigl\{&
 \zeta_j^{\mathsf T}\bigl[\tau_j + M_j^{0}\pi_{j1}
 + C_j^{0}\pi_{j3} - G_j^{0} + W_j^{*}X_j(t)\bigr]\\
 &+ \frac{\zeta_j^{\mathsf T}\zeta_j
          \bigl\lVert M_j^{0}\pi_{j2}\bigr\rVert^{2}}
         {2a_{j1}^{2}}\,\delta_{j'}
  + \frac{\zeta_j^{\mathsf T}\zeta_j}{2a_{j2}^{2}}\,\omega_j^{r}
  + \frac{a_{j1}^{2}+a_{j2}^{2}}{2}\Bigr\}.
\end{aligned}
\label{eq:V1dot-young}
\end{equation}
Each of the two unmeasured bounds $\delta_{j'}$ and $\omega_j^{r}$ now
multiplies a computable quantity, which is what
\eqref{eq:adapt-delta} and \eqref{eq:adapt-omega} are constructed to
estimate.

Substituting the control law \eqref{eq:control-law} into
\eqref{eq:V1dot-young}, the nominal feedforward cancels term by term,
the estimate $\hat W_jX_j$ combines with $W_j^{*}X_j$ into the parameter
error $\tilde W_j = W_j^{*}-\hat W_j$, and the two domination terms
convert the unknown bounds into their estimation errors
$\tilde\delta_{j'}$ and $\tilde\omega_j^{r}$:
\begin{equation}
\begin{aligned}
\dot{\mathcal V}_1 \le \sum_{j}\Bigl\{&
 -\zeta_j^{\mathsf T}K_{j1}\zeta_j
 -\zeta_j^{\mathsf T}K_{j2}\,\mathrm{sig}(\zeta_j)^{\sigma}
 + \zeta_j^{\mathsf T}\tilde W_j X_j(t)\\
 &+ \frac{\zeta_j^{\mathsf T}\zeta_j
          \bigl\lVert M_j^{0}\pi_{j2}\bigr\rVert^{2}}{2a_{j1}^{2}}\,
   \tilde\delta_{j'}
 + \frac{\zeta_j^{\mathsf T}\zeta_j}{2a_{j2}^{2}}\,\tilde\omega_j^{r}
 + \frac{a_{j1}^{2}+a_{j2}^{2}}{2}\Bigr\}.
\end{aligned}
\label{eq:V1dot-closed}
\end{equation}
The first two terms are negative definite in $\zeta_j$ and carry the
asymptotic and finite-time action respectively. The three remaining
sign-indefinite terms are each linear in one estimation error, and each
is cancelled by the corresponding adaptive law.

Differentiating \eqref{eq:V2} and using
$\dot{\tilde W}_j=-\dot{\hat W}_j$,
$\dot{\tilde\delta}_{j'}=-\dot{\hat\delta}_{j'}$ and
$\dot{\tilde\omega}_j^{r}=-\dot{\hat\omega}_j^{r}$, which hold because
$W_j^{*}$, $\delta_{j'}$ and $\omega_j^{r}$ are constants,
\begin{equation}
\dot{\mathcal V}_2 = \sum_{j}\left[
 -\frac{\mathrm{tr}\bigl(\tilde W_j\dot{\hat W}_j^{\mathsf T}\bigr)}
       {\gamma_j^{W}}
 -\frac{\tilde\delta_{j'}\dot{\hat\delta}_{j'}}{\gamma_j^{\delta}}
 -\frac{\tilde\omega_j^{r}\dot{\hat\omega}_j^{r}}{\gamma_j^{\omega}}
 \right].
\label{eq:V2dot-raw}
\end{equation}
Inserting \eqref{eq:adapt-W}--\eqref{eq:adapt-omega} and using the
identity
$\mathrm{tr}\bigl(\tilde W_jX_j\zeta_j^{\mathsf T}\bigr)
=\zeta_j^{\mathsf T}\tilde W_jX_j$, valid since the argument is a
rank-one matrix,
\begin{equation}
\begin{aligned}
\dot{\mathcal V}_2 = \sum_{j}\Bigl[&
 -\zeta_j^{\mathsf T}\tilde W_jX_j(t)
 + \rho_j\,\mathrm{tr}\bigl(\tilde W_j\hat W_j^{\mathsf T}\bigr)\\
 &-\frac{\zeta_j^{\mathsf T}\zeta_j
        \bigl\lVert M_j^{0}\pi_{j2}\bigr\rVert^{2}}{2a_{j1}^{2}}\,
  \tilde\delta_{j'}
 + \rho_j\tilde\delta_{j'}\hat\delta_{j'}\\
 &-\frac{\zeta_j^{\mathsf T}\zeta_j}{2a_{j2}^{2}}\,\tilde\omega_j^{r}
 + \rho_j\tilde\omega_j^{r}\hat\omega_j^{r}\Bigr].
\end{aligned}
\label{eq:V2dot}
\end{equation}
The three leading terms of \eqref{eq:V2dot} are the exact negatives of
the three indefinite terms of \eqref{eq:V1dot-closed}. This cancellation
is the reason the adaptive laws take the form they do, and it is
unaffected by the internal structure of the estimator: it requires only
that $X_j$ be independent of the adapted parameters, which holds because
the reservoir weights are fixed.

The leakage terms are handled by the elementary bounds
\begin{equation}
\begin{aligned}
\rho_j\,\mathrm{tr}\bigl(\tilde W_j\hat W_j^{\mathsf T}\bigr)
 &\le -\frac{\rho_j}{2}\,
      \mathrm{tr}\bigl(\tilde W_j\tilde W_j^{\mathsf T}\bigr)
      +\frac{\rho_j}{2}\,
      \mathrm{tr}\bigl(W_j^{*}W_j^{*\mathsf T}\bigr),\\
\rho_j\tilde\delta_{j'}\hat\delta_{j'}
 &\le -\frac{\rho_j}{2}\tilde\delta_{j'}^{2}
      +\frac{\rho_j}{2}\delta_{j'}^{2},\\
\rho_j\tilde\omega_j^{r}\hat\omega_j^{r}
 &\le -\frac{\rho_j}{2}\bigl(\tilde\omega_j^{r}\bigr)^{2}
      +\frac{\rho_j}{2}\bigl(\omega_j^{r}\bigr)^{2},
\end{aligned}
\label{eq:leakage}
\end{equation}
each obtained by writing the estimate as the true value minus its error
and applying $2ab\le a^{2}+b^{2}$. Adding
\eqref{eq:V1dot-closed} and \eqref{eq:V2dot} and applying
\eqref{eq:leakage},
\begin{equation}
\begin{aligned}
\dot{\mathcal V}_1+\dot{\mathcal V}_2 \le \sum_{j}\Bigl[&
 -\zeta_j^{\mathsf T}K_{j1}\zeta_j
 -\zeta_j^{\mathsf T}K_{j2}\,\mathrm{sig}(\zeta_j)^{\sigma}
 -\frac{\rho_j}{2}\,
   \mathrm{tr}\bigl(\tilde W_j\tilde W_j^{\mathsf T}\bigr)\\
 &-\frac{\rho_j}{2}\tilde\delta_{j'}^{2}
 -\frac{\rho_j}{2}\bigl(\tilde\omega_j^{r}\bigr)^{2}
 +\Psi_j\Bigr],
\end{aligned}
\label{eq:V12dot}
\end{equation}
where
\begin{equation}
\Psi_j = \frac{\rho_j}{2}
 \Bigl[\mathrm{tr}\bigl(W_j^{*}W_j^{*\mathsf T}\bigr)
 +\delta_{j'}^{2}+\bigl(\omega_j^{r}\bigr)^{2}\Bigr]
 +\frac{a_{j1}^{2}+a_{j2}^{2}}{2}
\label{eq:Psi}
\end{equation}
is a positive constant, finite by Assumptions~\ref{as:delay}
and~\ref{as:approx}.

\begin{rmk}
\label{rem:cancellation}
The parameter-error term $\zeta_j^{\mathsf T}\tilde W_jX_j$ cancels
identically, without any bound on $\lVert X_j\rVert$ and without any
assumption on how $X_j$ is generated. Fixing the reservoir and adapting
only the readout is therefore not merely convenient: it is what places
the estimator inside the class for which this cancellation holds, and it
is the sense in which the finite-time certificate of the underlying
scheme is preserved rather than re-established. A bound on
$\lVert X_j\rVert$ is required later, but only to size the residual set.
\end{rmk}

Differentiating \eqref{eq:V3},
\begin{equation}
\dot{\mathcal V}_3 = \sum_{j}\left[
 \eta_j\,\zeta_j^{\mathsf T}\zeta_j
 -\frac{\eta_j}{\bar T_j}\!\int_{t-\bar T_j}^{t}\!
  \zeta_j^{\mathsf T}(\alpha)\zeta_j(\alpha)\,\mathrm d\alpha\right].
\label{eq:V3dot-raw}
\end{equation}
Since $(\bar T_j-t+\alpha)/\bar T_j\le1$ on the interval of
integration, the integral in \eqref{eq:V3dot-raw} dominates the one
appearing in $\mathcal V_3$, so that
\begin{equation}
\dot{\mathcal V}_3 \le \sum_{j}\left[
 \eta_j\,\zeta_j^{\mathsf T}\zeta_j
 -\frac{\eta_j}{\bar T_j}\!\int_{t-\bar T_j}^{t}\!
  \frac{\bar T_j-t+\alpha}{\bar T_j}\,
  \zeta_j^{\mathsf T}(\alpha)\zeta_j(\alpha)\,\mathrm d\alpha\right].
\label{eq:V3dot}
\end{equation}
Combining \eqref{eq:V12dot} with \eqref{eq:V3dot} and using
Property~\ref{prop:inertia} to write
$\zeta_j^{\mathsf T}\zeta_j\ge
(1/\lambda_{j,\max})\,\zeta_j^{\mathsf T}M_j\zeta_j$,
\begin{equation}
\begin{aligned}
\dot{\mathcal V} \le \sum_{j}\Bigl[&
 -\frac{k_{j}}{2}\,\zeta_j^{\mathsf T}M_j\zeta_j
 -\zeta_j^{\mathsf T}K_{j2}\,\mathrm{sig}(\zeta_j)^{\sigma}
 -\frac{\rho_j}{2}\,
  \mathrm{tr}\bigl(\tilde W_j\tilde W_j^{\mathsf T}\bigr)\\
 &-\frac{\rho_j}{2}\tilde\delta_{j'}^{2}
 -\frac{\rho_j}{2}\bigl(\tilde\omega_j^{r}\bigr)^{2}
 -\frac{\eta_j}{\bar T_j}\!\int_{t-\bar T_j}^{t}\!
  \frac{\bar T_j-t+\alpha}{\bar T_j}\,
  \zeta_j^{\mathsf T}(\alpha)\zeta_j(\alpha)\,\mathrm d\alpha
 +\Psi_j\Bigr],
\end{aligned}
\label{eq:Vdot-assembled}
\end{equation}
where
\begin{equation}
k_{j} = \frac{2\bigl(\lambda_{\min}(K_{j1})-\eta_j\bigr)}
             {\lambda_{j,\max}} .
\label{eq:kj}
\end{equation}
The gain $K_{j1}$ is chosen so that
$\lambda_{\min}(K_{j1})>\eta_j$, which is always possible since
$\eta_j\ge0$ is a free constant of the functional \eqref{eq:V3} and may
be taken arbitrarily small; then $k_{j}>0$. Discarding the
non-positive finite-time and integral terms and setting
\begin{equation}
\Lambda = \min_{j\in\{m,s\}}
 \Bigl\{k_{j},\ \rho_j\gamma_j^{W},\ \rho_j\gamma_j^{\delta},\
 \rho_j\gamma_j^{\omega},\ 1/\bar T_j\Bigr\},
\qquad
\Psi = \Psi_m+\Psi_s ,
\label{eq:LambdaPsi}
\end{equation}
each surviving term of \eqref{eq:Vdot-assembled} is bounded above by
$-\Lambda$ times the corresponding term of \eqref{eq:lyap}, so that
\begin{equation}
\dot{\mathcal V} \le -\Lambda\,\mathcal V + \Psi ,
\qquad \Lambda>0,\ \Psi>0 .
\label{eq:Vdot-final}
\end{equation}
The comparison lemma applied to \eqref{eq:Vdot-final} gives
\begin{equation}
\mathcal V(t) \le
 \Bigl(\mathcal V(0)-\frac{\Psi}{\Lambda}\Bigr)e^{-\Lambda t}
 +\frac{\Psi}{\Lambda},
\qquad \forall t\ge0 ,
\label{eq:Vsol}
\end{equation}
so $\mathcal V$ is bounded for all $t$ and
$\limsup_{t\to\infty}\mathcal V(t)\le\Psi/\Lambda$. Since $\mathcal V$
is positive definite and radially unbounded in $\zeta_j$, $\tilde W_j$,
$\tilde\delta_{j'}$ and $\tilde\omega_j^{r}$, all four are uniformly
ultimately bounded, and with them $\hat W_j$, $\hat\delta_{j'}$ and
$\hat\omega_j^{r}$. Boundedness of $\zeta_j$ propagates to the remaining
closed-loop signal. The filtered reservoir state satisfies
$x_{L,j}(t)\in[0,1]^N$ by Lemma~\ref{lem:leaky}, and the augmented
feature vector $X_j(t)$ is uniformly bounded by \eqref{eq:Xbound}.
Since $\hat W_j$ is bounded, the uncertainty estimate
$\hat\Pi_j=\hat W_jX_j$ is also bounded.

\begin{rmk}
\label{rem:residual-size}
The ultimate bound $\Psi/\Lambda$ separates into a part fixed by the
plant and the specification and a part the designer controls. The
contributions $\delta_{j'}^{2}$ and $a_{j1}^{2}+a_{j2}^{2}$ in
\eqref{eq:Psi} come from the delay rate and from the Young's-inequality
step; the contribution
$\mathrm{tr}(W_j^{*}W_j^{*\mathsf T})+(\omega_j^{r})^{2}$ is the price
of the leakage $\rho_j$ and of the reconstruction residual. By
Assumption~\ref{as:approx} the latter shrinks as the reservoir size $N$
grows, so the estimator quality enters the guarantee only through the
size of the residual set, never through its existence. Enlarging
$K_{j1}$ increases $\Lambda$ and contracts the set; reducing $\rho_j$
reduces $\Psi$ but slows the decay of the parameter errors.
\end{rmk}

The analysis so far establishes boundedness. Finite-time reaching is
shown on the functional
\begin{equation}
\bar{\mathcal V} = \mathcal V_1 + \mathcal V_2 ,
\label{eq:Vbar}
\end{equation}
which omits the Krasovskii term \eqref{eq:V3}; the latter has no
fractional-power counterpart and is not required here, since
\eqref{eq:V12dot} already holds independently of $\mathcal V_3$.

Write the constituents of \eqref{eq:Vbar} as
$\mathcal V_{1j}=\tfrac12\zeta_j^{\mathsf T}M_j\zeta_j$,
$P_j^{W}=\mathrm{tr}(\tilde W_j\tilde W_j^{\mathsf T})/2\gamma_j^{W}$,
$P_j^{\delta}=\tilde\delta_{j'}^{2}/2\gamma_j^{\delta}$ and
$P_j^{\omega}=(\tilde\omega_j^{r})^{2}/2\gamma_j^{\omega}$, and set
$\bar k_{j1}=\lambda_{\min}(K_{j1})$,
$\bar k_{j2}=\lambda_{\min}(K_{j2})$. Since $K_{j2}$ is diagonal and
positive definite and $\sigma+1\in(1,2)$,
\begin{equation}
\zeta_j^{\mathsf T}K_{j2}\,\mathrm{sig}(\zeta_j)^{\sigma}
 = \sum_{i=1}^{n}K_{j2,ii}\bigl|\zeta_{j,i}\bigr|^{\sigma+1}
 \;\ge\; \bar k_{j2}\bigl\lVert\zeta_j\bigr\rVert^{\sigma+1},
\label{eq:sig-lower}
\end{equation}
the last step because $\lVert\cdot\rVert_{\sigma+1}\ge\lVert\cdot\rVert_2$
for exponents below $2$. Using Property~\ref{prop:inertia} to convert
both quadratic forms,
\begin{equation}
\bigl\lVert\zeta_j\bigr\rVert^{2}\ge
 \frac{2}{\lambda_{j,\max}}\,\mathcal V_{1j},
\qquad
\bigl\lVert\zeta_j\bigr\rVert^{\sigma+1}\ge
 \Bigl(\frac{2}{\lambda_{j,\max}}\Bigr)^{\frac{\sigma+1}{2}}
 \mathcal V_{1j}^{\frac{\sigma+1}{2}} .
\label{eq:zeta-convert}
\end{equation}
Subtracting and adding the fractional powers of the parameter-error
terms in \eqref{eq:V12dot} and applying
\eqref{eq:sig-lower}--\eqref{eq:zeta-convert},
\begin{equation}
\begin{aligned}
\dot{\bar{\mathcal V}} \le \sum_{j}\Bigl\{&
 -\frac{2\bar k_{j1}}{\lambda_{j,\max}}\mathcal V_{1j}
 -\rho_j\gamma_j^{W}P_j^{W}
 -\rho_j\gamma_j^{\delta}P_j^{\delta}
 -\rho_j\gamma_j^{\omega}P_j^{\omega}\\
 &-\bar k_{j2}\Bigl(\frac{2}{\lambda_{j,\max}}\Bigr)^{\frac{\sigma+1}{2}}
   \mathcal V_{1j}^{\frac{\sigma+1}{2}}
 -\bigl(P_j^{W}\bigr)^{\frac{\sigma+1}{2}}
 -\bigl(P_j^{\delta}\bigr)^{\frac{\sigma+1}{2}}
 -\bigl(P_j^{\omega}\bigr)^{\frac{\sigma+1}{2}}
 +\Xi_j\Bigr\},
\end{aligned}
\label{eq:Vbardot}
\end{equation}
where
\begin{equation}
\Xi_j = \Psi_j
 + \bigl(P_j^{W}\bigr)^{\frac{\sigma+1}{2}}
 + \bigl(P_j^{\delta}\bigr)^{\frac{\sigma+1}{2}}
 + \bigl(P_j^{\omega}\bigr)^{\frac{\sigma+1}{2}} .
\label{eq:Xi}
\end{equation}
The parameter errors were shown bounded in \eqref{eq:Vsol}, so
$\Xi_j$ is bounded: there exist constants $\Upsilon_m,\Upsilon_s>0$
with $0\le\Xi_j\le\Upsilon_j$. Defining
\begin{equation}
\bar\Lambda_1 = \min_{j}\Bigl\{
 \frac{2\bar k_{j1}}{\lambda_{j,\max}},\
 \rho_j\gamma_j^{W},\ \rho_j\gamma_j^{\delta},\
 \rho_j\gamma_j^{\omega}\Bigr\},
\quad
\bar\Lambda_2 = \min_{j}\Bigl\{
 \bar k_{j2}\Bigl(\frac{2}{\lambda_{j,\max}}\Bigr)^{\frac{\sigma+1}{2}},\
 1\Bigr\},
\label{eq:Lambda12}
\end{equation}
and applying Lemma~\ref{lem:powersum} with $p=(\sigma+1)/2\in(0,1)$ to
collect the eight fractional-power terms of \eqref{eq:Vbardot} into a
single power of their sum,
\begin{equation}
\dot{\bar{\mathcal V}}
 \;\le\; -\bar\Lambda_1\bar{\mathcal V}
 -\bar\Lambda_2\bar{\mathcal V}^{\frac{\sigma+1}{2}}
 +\Upsilon ,
\qquad \Upsilon=\Upsilon_m+\Upsilon_s .
\label{eq:Vbardot-final}
\end{equation}
Fix $\theta\in(0,1)$ and define the residual set
\begin{equation}
\mathcal D = \left\{\;\bar{\mathcal V}\;\le\;
 \Bigl(\frac{\Upsilon}{(1-\theta)\bar\Lambda_2}\Bigr)
 ^{\frac{2}{\sigma+1}}\right\}.
\label{eq:residual-set}
\end{equation}
Outside $\mathcal D$ the constant $\Upsilon$ is dominated by a fraction
of the fractional-power term, giving
$\dot{\bar{\mathcal V}}\le-\bar\Lambda_1\bar{\mathcal V}
-\theta\bar\Lambda_2\bar{\mathcal V}^{(\sigma+1)/2}$, which is
\eqref{eq:ftineq} with $\alpha=\bar\Lambda_1$,
$\beta=\theta\bar\Lambda_2$. Lemma~\ref{lem:finitetime}(ii) then gives
reaching of $\mathcal D$ within
\begin{equation}
T_{\mathrm{fin}} \;\le\;
 \frac{2}{\bar\Lambda_1(1-\sigma)}
 \ln\frac{\bar\Lambda_1\bar{\mathcal V}^{\frac{1-\sigma}{2}}(0)
          +\theta\bar\Lambda_2}{\theta\bar\Lambda_2}\,,
\label{eq:Tfin}
\end{equation}
and \eqref{eq:Vsol} guarantees that the trajectory remains bounded
thereafter. This completes the proof.
\end{pf}

\begin{rmk}
\label{rem:gain-conditions}
Both $\bar\Lambda_1$ and $\bar\Lambda_2$ in \eqref{eq:Lambda12} are
positive for any positive definite $K_{j1},K_{j2}$ and any positive
adaptation gains, so no lower bound on the feedback gains is needed for
finite-time reaching; the gains set the size of $\mathcal D$ through
\eqref{eq:residual-set} and the reaching time through \eqref{eq:Tfin}.
The parameter $\theta$ trades the two against each other: taking
$\theta$ near one shrinks $\mathcal D$ and lengthens $T_{\mathrm{fin}}$.
\end{rmk}

\begin{rmk}
\label{rem:preserved}
Neither \eqref{eq:Vbardot-final} nor \eqref{eq:Tfin} refers to the
internal structure of the estimator. The reservoir enters the argument
at exactly two points: the cancellation of
Remark~\ref{rem:cancellation}, which requires only that $X_j$ be
independent of $\hat W_j$, and the constant $\omega_j^{r}$ inside
$\Psi_j$, which sizes $\mathcal D$. The spiking dynamics, the recurrent
connectivity and the fading memory of Section~\ref{ssec:lsm} therefore
affect how large the residual set is, and not whether the finite-time
guarantee holds.
\end{rmk}

\begin{rmk}
\label{rem:tuning}
Theorem~\ref{thm:main} constrains the design parameters only by sign, so
the remaining freedom is used to shape the residual set and the
transient. A workable order is the following.

The filter gains $L_j^{\tau}$ and $L_j^{v}$ are set first. Both filters
are driven by signals received over the delayed channel, so smaller
gains yield smoother $\phi_j$ and $\nu_j$ at the cost of a slower
response; $L_j^{v}$ additionally governs how closely $e_j^{\nu}$
tracks the delay-consistent velocity error \eqref{eq:ejv}.

The weights $\chi_{j1},\chi_{j2}$ then place the operating point on the
position--force trade-off of Remark~\ref{rem:conflict}, and
$\lambda_{j1},\lambda_{j2},\sigma_{j1},\sigma_{j2}$ shape the
finite-time term \eqref{eq:psi}.

The feedback gains follow. Larger $K_{j1}$ increases $\bar\Lambda_1$ and
larger $K_{j2}$ increases $\bar\Lambda_2$, contracting $\mathcal D$ in
\eqref{eq:residual-set} and shortening $T_{\mathrm{fin}}$ in
\eqref{eq:Tfin}. The constants $a_{j1},a_{j2}$ act in opposition:
reducing them removes $(a_{j1}^{2}+a_{j2}^{2})/2$ from $\Psi_j$ in
\eqref{eq:Psi}, but they appear in the denominators of the two
domination terms of \eqref{eq:control-law}, so small values produce a
stiff response to the delay rate and to the reconstruction residual.
Larger $\rho_j$ raises $\bar\Lambda_1$ and accelerates the decay of the
parameter errors, while enlarging $\Psi_j$ through the leakage
contribution.

The estimator parameters have no counterpart in a memoryless design and
are set last. The reservoir size $N$ enters the guarantee only through
$\omega_j^{r}$, by Assumption~\ref{as:approx}, so it is chosen against
measured reconstruction error rather than against a stability condition.
The recurrent matrix $W_j^{\mathrm{rec}}$ is scaled to a spectral radius
below unity, which is what makes the liquid's dependence on its input
history fade. The time constants
$\tau^{\mathrm{mem}},\tau^{\mathrm{syn}},\tau^{\mathrm{flt}}$ set the
depth of that memory: by \eqref{eq:leakysol} the filtered state weights
a spike at lag $r$ by $e^{-r/\tau^{\mathrm{flt}}}$, so retaining a
fraction $\kappa$ of the weight across the full delay window requires
\begin{equation}
\tau^{\mathrm{flt}} \;\gtrsim\;
 \frac{\bar T_{j'}}{\ln(1/\kappa)} .
\label{eq:memory-depth}
\end{equation}
This is the one design rule specific to the reservoir: the memory depth
is matched to the delay bound of Assumption~\ref{as:delay}, not tuned by
trial.
\end{rmk}

\section{Simulation Results}
\label{sec:simulation}

\setlength{\intextsep}{2pt plus 2pt minus 2pt}
\setlength{\textfloatsep}{2pt plus 2pt minus 2pt}
\setlength{\floatsep}{2pt plus 2pt minus 2pt}

In this section, we use MATLAB/Simulink software to validate the effectiveness of our proposed method. 
The sampling rate in simulation is chosen as $0.001$s. The master and slave manipulators are modeled as 
two-link planar rigid robotic manipulators same as \eqref{eq:dyn}, the description of mass inertia matrix 
$M_j$, centripetal-coriolis matrix $C_j$, and gravitational torque $G_j$ can be attained in \cite{spong2020}. 
The vector of friction torque $B_j$ is given as:

\begin{equation}
B_{j}=
{\left[
b_{j1}\dot{q}_{j1}+b_{j2}sign\left(\dot{q}_{j1}\right),\ 
b_{j3}\dot{q}_{j2}+b_{j4}sign\left(\dot{q}_{j2}\right)
\right]}^{T}
\label{eq:friction_torque}
\end{equation}

The parameters of master and slave manipulators are shown in Table~\ref{tab:robot_parameters} and
parameter of $B_{j}$ are chosen as
$b_{m1}=b_{m3}=0.5,\ b_{m2}=b_{m4}=0.2,\ 
b_{s1}=b_{s2}=b_{s3}=b_{s4}=0.3$ and
$g=9.8\ m/s^{2}$.
The initial conditions of joint positions are set as
$q_{m}={\left[\frac{\pi}{12},\frac{\pi}{6}\right]}^{T}$,
${\dot{q}}_{m}={\left[0,0\right]}^{T}$,
$q_{s}={\left[\frac{\pi}{4},\frac{\pi}{6}\right]}^{T}$,
${\dot{q}}_{s}={\left[0,0\right]}^{T}$.
The damping and spring matrices in operator and environment model are
$D_{h}=\mathrm{diag}(0.1,\ 0.1)$,
$S_{h}=\mathrm{diag}(10.0,\ 10.0)$ and
$D_{e}=\mathrm{diag}(0.5,\ 0.5)$,
$S_{e}=\mathrm{diag}(10.0,\ 10.0)$.
The external operator interface torque $F_{h}^{\ast}$, which is applied on
both joints equally, is shown in Fig.~\ref{fig:fig1}, and we assume that it only exists
in the x direction.
\begin{figure}[pos=!htbp]
    \centering
    \includegraphics[width=\linewidth]{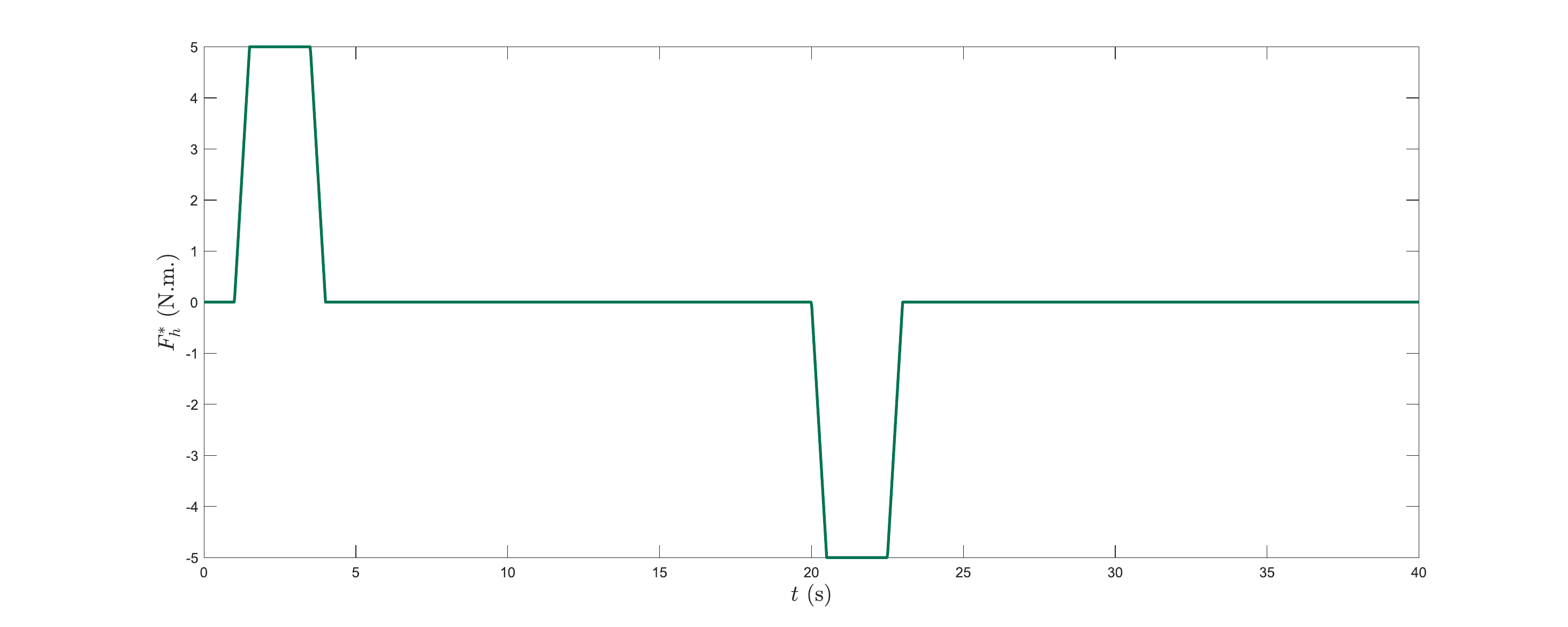}
    \caption{External operator torque $F_h^{*}$.}
    \label{fig:fig1}
\end{figure}
The external torque of the environment is 0 in both directions.
Operator and environments torque $\tau_{h}$ and $\tau_{e}$ are also depicted
in Fig.~\ref{fig:fig3} and Fig.~\ref{fig:fig4}.
\begin{figure}[pos=!htbp]
    \centering
    \includegraphics[width=0.85\linewidth]{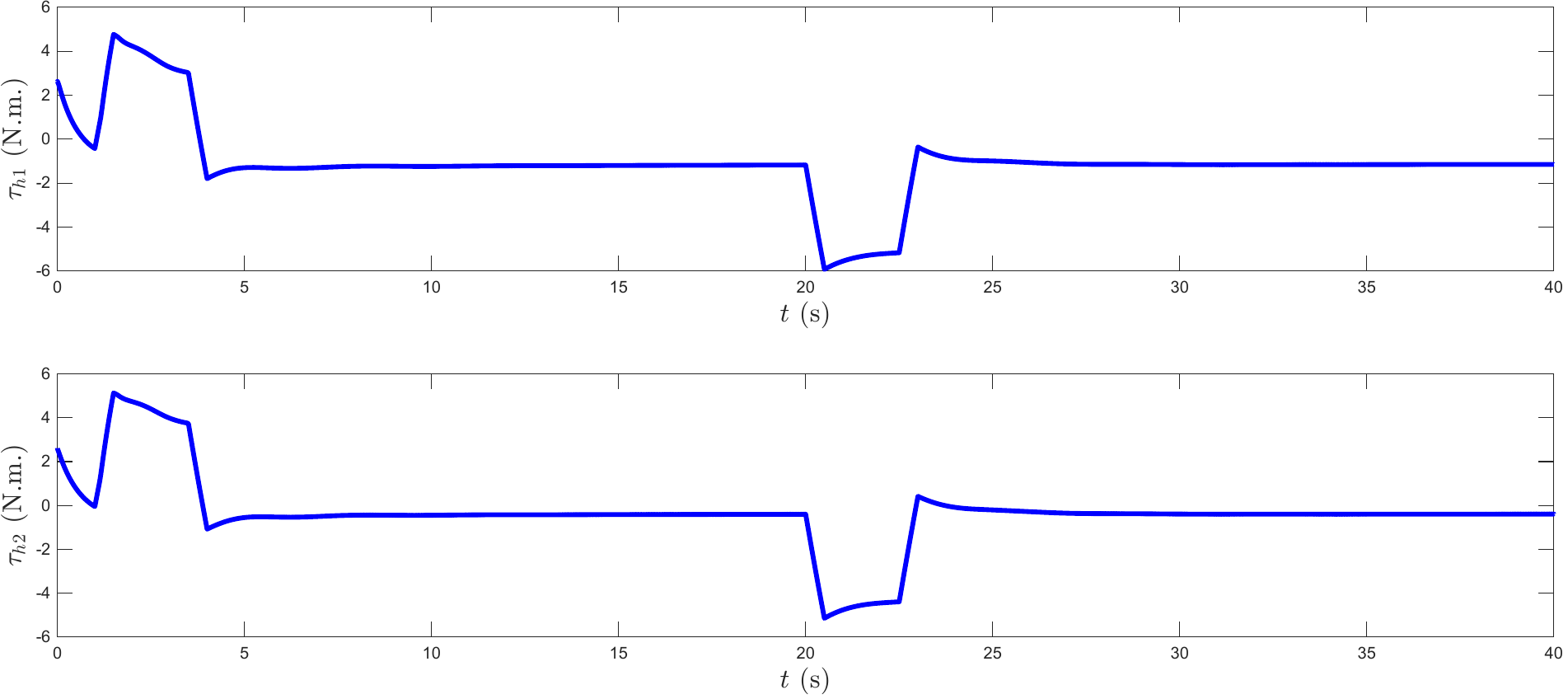}
    \caption{Operator torque $\tau_h$.}
    \label{fig:fig3}
\end{figure}

\begin{figure}[pos=!htbp]
    \centering
    \includegraphics[width=0.85\linewidth]{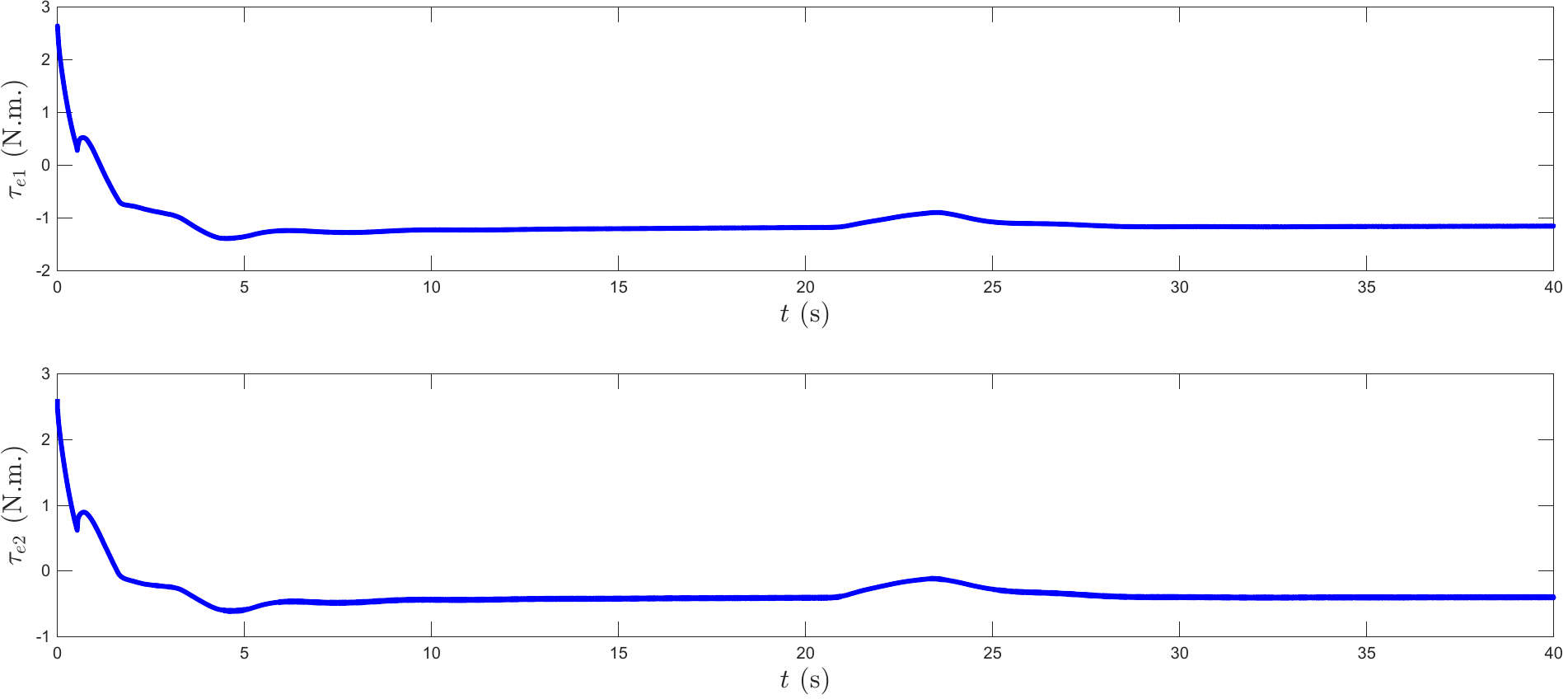}
    \caption{Environment torque $\tau_e$.}
    \label{fig:fig4}
\end{figure}
The master-slave asymmetric time-varying communication
$T_{m}(t),\ T_{s}(t)$ delays are shown in Fig.~\ref{fig:fig2}.
\begin{figure}[pos=!htbp]
    \centering
    \includegraphics[width=0.85\linewidth]{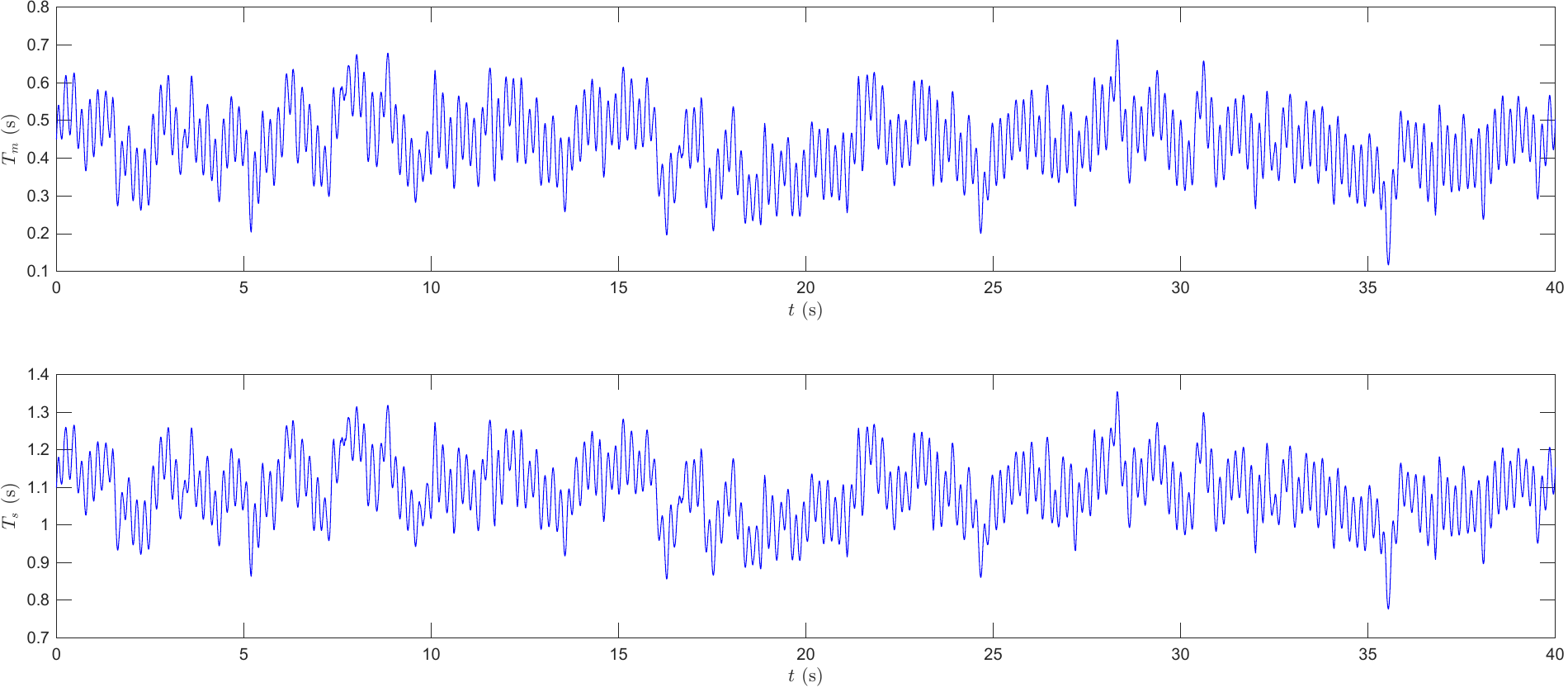}
    \caption{Communication time delays $T_m(t)$, $T_s(t)$.}
    \label{fig:fig2}
\end{figure}
The delay signals are obtained using equation below:
\begin{equation}
T_m(t)=0.45+0.08\sin(30t)+d_r
\label{eq:delay_master}
\end{equation}

\begin{equation}
T_s(t)=1.1+0.07\sin(30t+0.6)+d_r
\label{eq:delay_slave}
\end{equation}

$d_r$ is filtered band-limited white noise. We assume that the uncertain
parts of bilateral teleoperation can be denoted as
$\Delta M_j=0.02\sin(4t)M_j$,
$\Delta C_j=0.02\sin(4t)C_j$,
$\Delta G_j=0.02\sin(4t)G_j$.
The parameter of nominal models $M_j^{0}$, $C_j^{0}$ and $G_j^{0}$ are set
according to Table~\ref{tab:robot_parameters}. The parameters of the proposed controller are chosen in Table~\ref{tab:controller_parameters}.

For both the master and slave manipulators, identical LSM architectures and
design parameters were employed. Each LSM consisted of $N=50$ LIF reservoir
neurons. The membrane and synaptic time constants were selected as
$\tau_m=0.05s$ and $\tau_s=0.01s$, respectively. The firing threshold was
chosen as $V_{th}=0.06V$. The input and recurrent weights were initialized
randomly and recurrent weight matrix was subsequently normalized. The
adaptive readout consisted of a weight matrix initialized to zero and updated
online according \eqref{eq:adapt-W} with adaption gain
$\Gamma^{W}=\mathrm{diag}(0.9,0.9)$. Table~\ref{tab:lsm_parameters} summarizes the qualitative
influence of the main LSM design parameters on the estimator response and
closed-loop control performance.

\begin{table}[pos=!htbp]
\caption{Master and slave robot parameters}
\label{tab:robot_parameters}
\begin{tabular}{ccccc}
\toprule
Parameter & Value & Nominal parameter & Value & Physical description \\
\midrule
$m_{j1}$ & $3.5\,Kg$ & $m^{0}_{j1}$ & $2.0\,Kg$ & Mass of linkage 1 \\
$m_{j2}$ & $2.5\,Kg$ & $m^{0}_{j2}$ & $2.0\,Kg$ & Mass of linkage 2 \\
$l_{j1}$ & $0.3\,m$ & $l^{0}_{j1}$ & $0.3\,m$ & Length of linkage 1 \\
$l_{j2}$ & $0.35\,m$ & $l^{0}_{j2}$ & $0.3\,m$ & Length of linkage 2 \\
\bottomrule
\end{tabular}
\end{table}

\begin{table}[pos=!htbp]
\caption{Controller parameters}
\label{tab:controller_parameters}
\begin{tabular}{cccc}
\toprule
Master & Value & Slave & Value \\
\midrule
$L^{v}_{m1}$ & $\mathrm{diag}(0.2,0.2)$ & $L^{v}_{s1}$ & $\mathrm{diag}(0.2,0.2)$ \\
$L^{v}_{m2}$ & $\mathrm{diag}(0.05,0.05)$ & $L^{v}_{s2}$ & $\mathrm{diag}(0.05,0.05)$ \\
$\chi_{m1}$ & $\mathrm{diag}(1,1)$ & $\chi_{s1}$ & $\mathrm{diag}(1,1)$ \\
$\chi_{m2}$ & $\mathrm{diag}(0.02,0.02)$ & $\chi_{s2}$ & $\mathrm{diag}(0.02,0.02)$ \\
$\psi_{m1}$ & $\mathrm{diag}(1.8,1.8)$ & $\psi_{s1}$ & $\mathrm{diag}(1.8,1.8)$ \\
$\psi_{m2}$ & $\mathrm{diag}(0.7,0.7)$ & $\psi_{s2}$ & $\mathrm{diag}(0.7,0.7)$ \\
$\sigma_{m1}$ & $1.1$ & $\sigma_{s1}$ & $1.1$ \\
$\sigma_{m2}$ & $0.9$ & $\sigma_{s2}$ & $0.9$ \\
$\sigma$ & $0.8$ & $-$ & $-$ \\
$a_{m1}$ & $0.5$ & $a_{s1}$ & $0.5$ \\
$a_{m2}$ & $0.5$ & $a_{s2}$ & $0.5$ \\
$K_{m1}$ & $\mathrm{diag}(12,12)$ & $K_{s1}$ & $\mathrm{diag}(17,17)$ \\
$K_{m2}$ & $\mathrm{diag}(12,12)$ & $K_{s2}$ & $\mathrm{diag}(17,17)$ \\
\bottomrule
\end{tabular}
\end{table}

\begin{table}[pos=!htbp]
\centering
\caption{Effect of LSM parameters}
\label{tab:lsm_parameters}
\renewcommand{\arraystretch}{1.2}

\begin{tabular}{@{}
    >{\centering\arraybackslash}p{0.10\linewidth}
    >{\centering\arraybackslash}p{0.14\linewidth}
    >{\raggedright\arraybackslash}p{0.24\linewidth}
    >{\raggedright\arraybackslash}p{0.23\linewidth}
    >{\raggedright\arraybackslash}p{0.23\linewidth}
@{}}

\toprule
\textbf{Parameter} &
\textbf{Value} &
\textbf{Role} &
\textbf{If increased} &
\textbf{If decreased} \\
\midrule

$N$ &
$50$ &
Liquid dimensionality and approximation capacity &
Better representation, higher computation &
Lower cost, weaker approximation \\

$\tau^{mem}$ &
$0.01\,s$ &
Membrane integration time &
Smoother output but slower neuron response &
Faster response, more spike fluctuations \\

$\tau^{flt}$ &
$0.05\,s$ &
Fading-memory time of liquid state &
Longer memory, smoother estimate, more lag &
Faster estimate, shorter memory \\

$V_{th}$ &
$0.05\,V$ &
Spike-generation threshold &
Sparser firing, lower sensitivity &
Denser firing, possible output jitter \\

$\Gamma$ &
$\mathrm{diag}[0.9,0.9]$ &
Adaptive readout learning rate &
Faster adaptation, possible oscillation &
Slower but smoother adaptation \\

\bottomrule
\end{tabular}
\end{table}
```

\subsection{Stability and tracking verification}

The first set of simulations evaluates the stability and position/force tracking
performance of the proposed LSM-based adaptive teleoperation controller. The
corresponding results are presented in Figs.~\ref{fig:fig5}--\ref{fig:fig15}. Fig.~\ref{fig:fig5} shows the joint
position responses of the master and slave manipulators
\begin{figure}[pos=!htbp]
    \centering
    \includegraphics[width=0.85\linewidth]{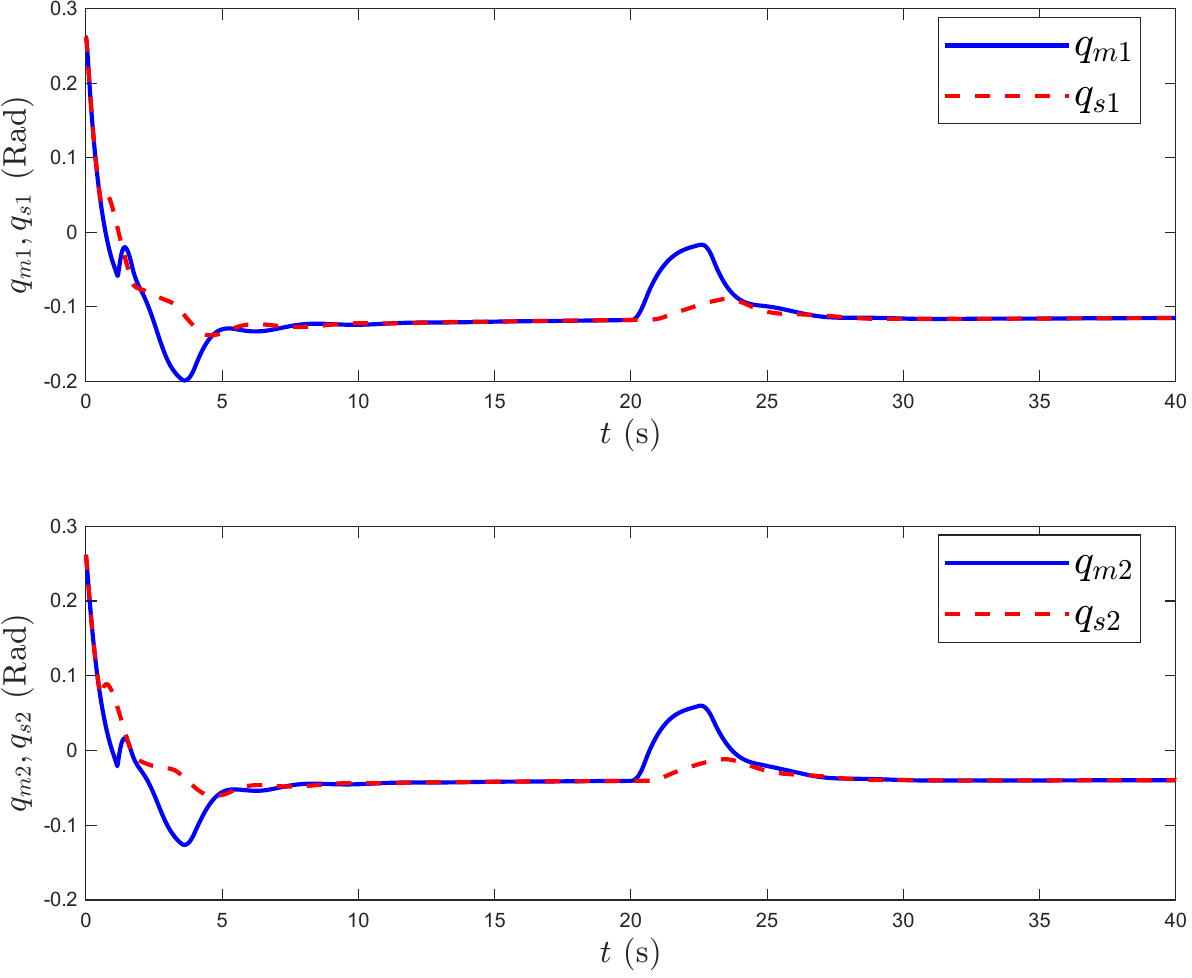}
    \caption{Joint position of master and slave manipulators $q_m,q_s$.}
    \label{fig:fig5}
\end{figure}
and Fig.~\ref{fig:fig6} the
human--environment interaction torques.
\begin{figure}[pos=!htbp]
    \centering
    \includegraphics[width=0.85\linewidth]{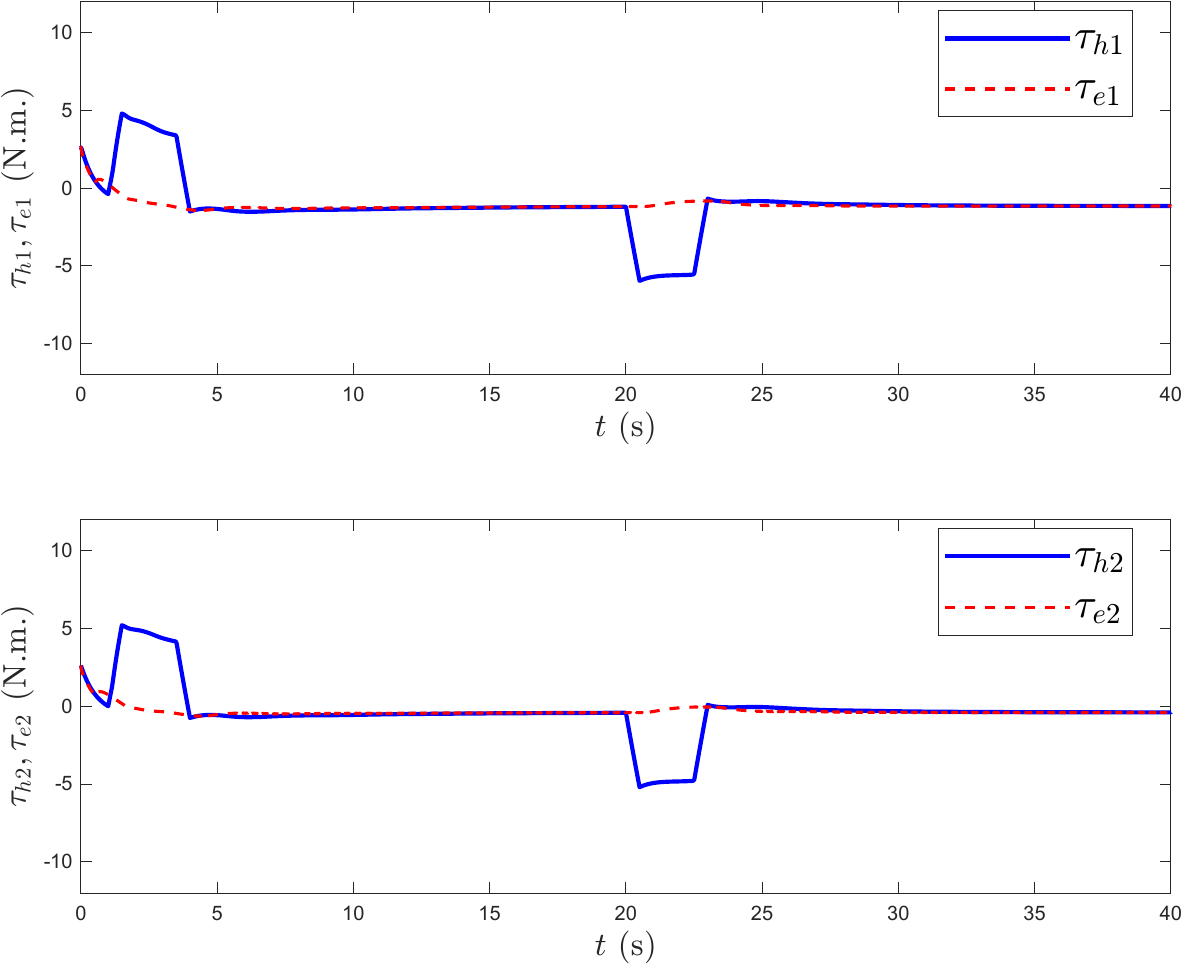}
    \caption{Human--environment interaction torques of master and slave manipulators $\tau_m,\tau_s$.}
    \label{fig:fig6}
\end{figure}
The position synchronization and force
tracking errors respectively shown in Figs.~\ref{fig:fig7} and~\ref{fig:fig8} decrease rapidly and remain
within a small neighborhood of zero, confirming accurate bilateral tracking
performance.
\begin{figure}[pos=!htbp]
    \centering
    \includegraphics[width=\linewidth]{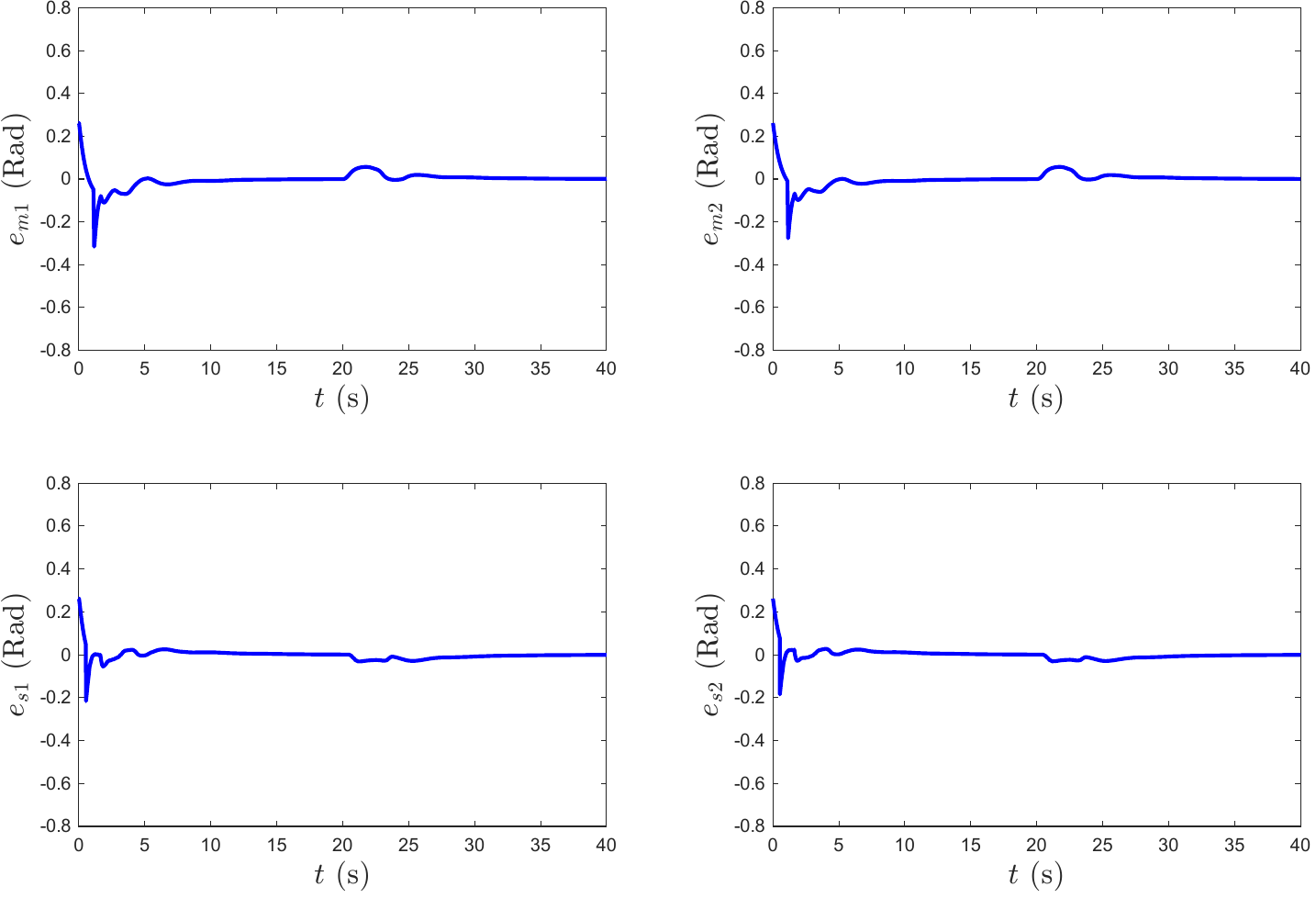}
    \caption{Position tracking error of master and slave manipulators.}
    \label{fig:fig7}
\end{figure}
\begin{figure}[pos=!htbp]
    \centering
    \includegraphics[width=\linewidth]{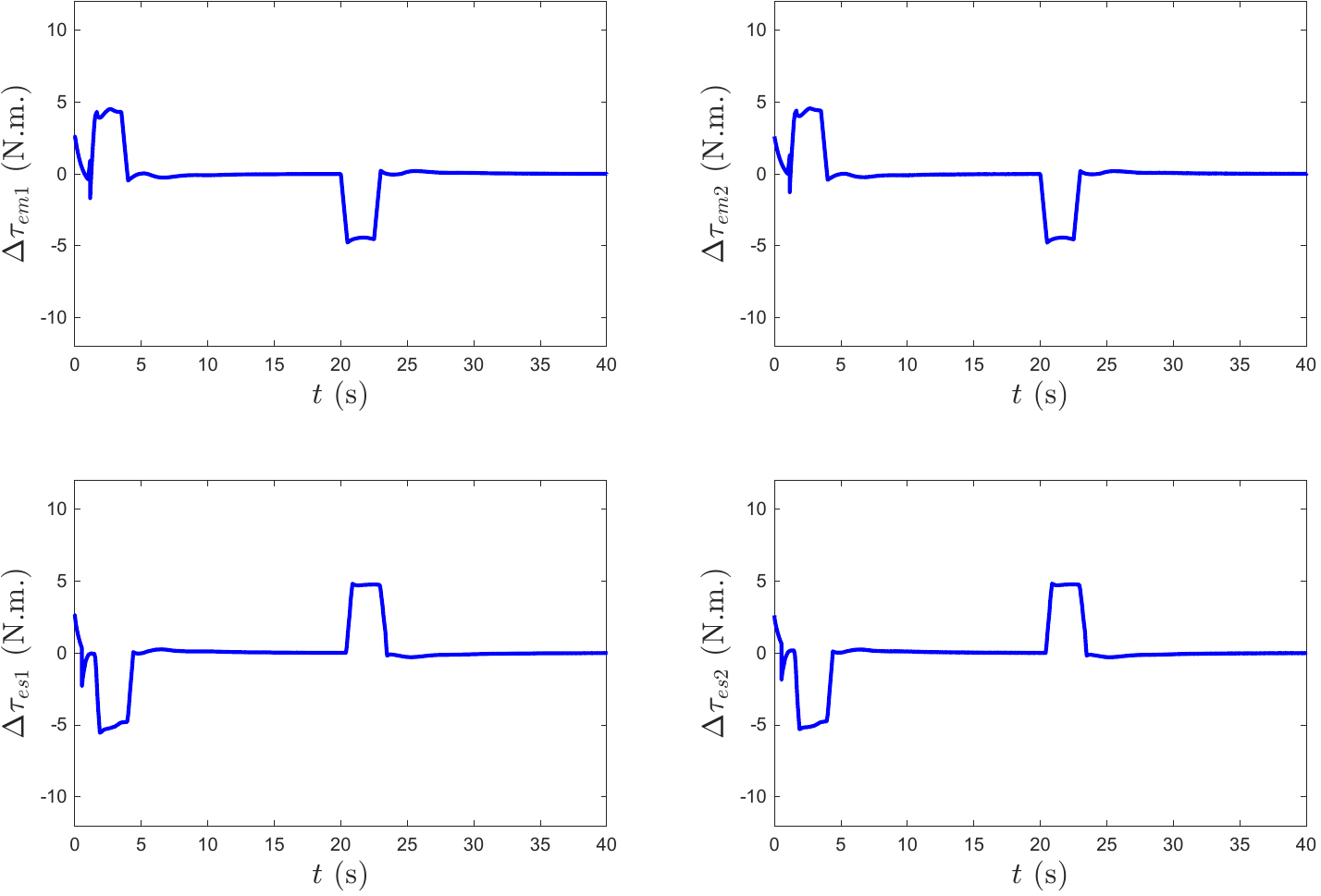}
    \caption{Force tracking error of master and slave manipulators.}
    \label{fig:fig8}
\end{figure}

The behavior of the LSM uncertainty estimators is further examined in Figs.~\ref{fig:fig9} and~\ref{fig:fig10}, which indicate the estimated uncertainties of master and slave
manipulators respectively.
\begin{figure}[pos=!htbp]
    \centering
    \includegraphics[width=0.85\linewidth]{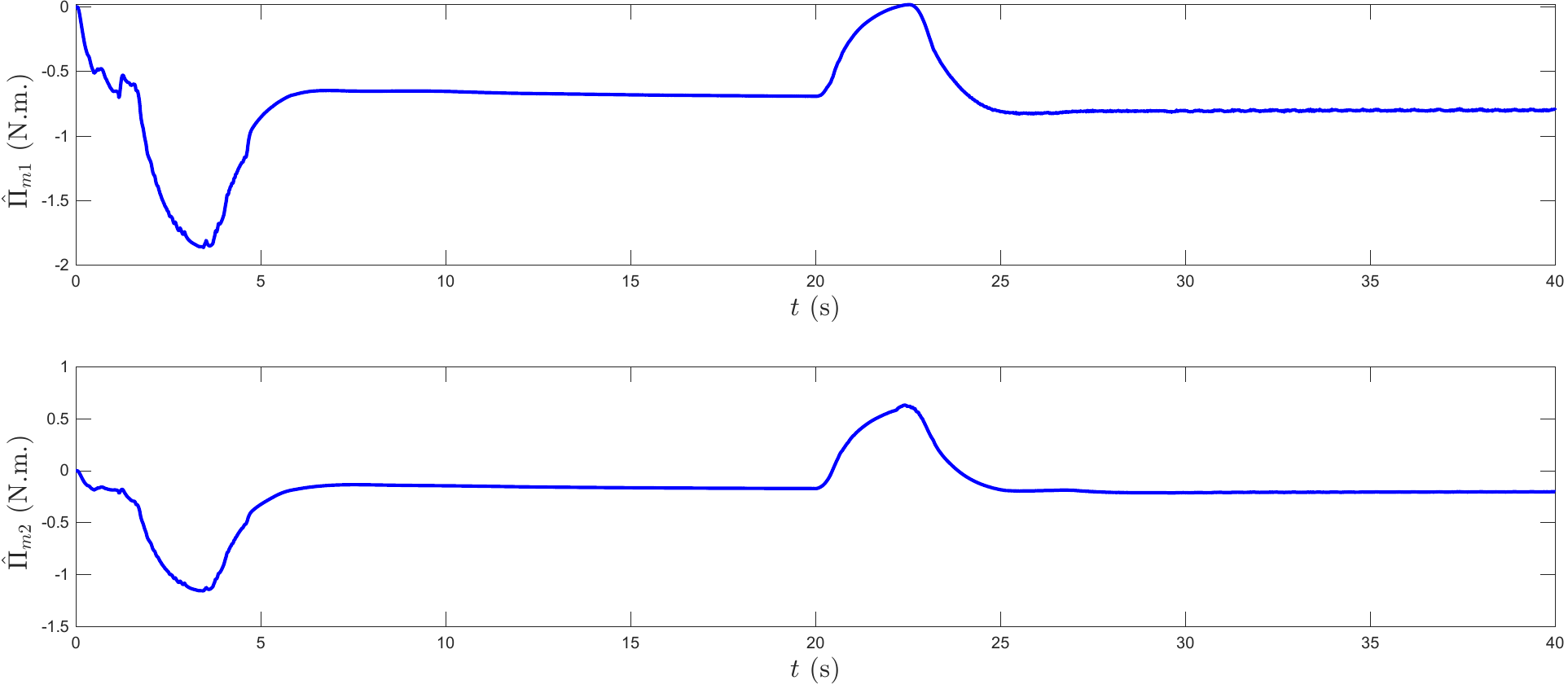}
    \caption{Estimated uncertainties of master manipulator.}
    \label{fig:fig9}
\end{figure}
\begin{figure}[pos=!htbp]
    \centering
    \includegraphics[width=0.85\linewidth]{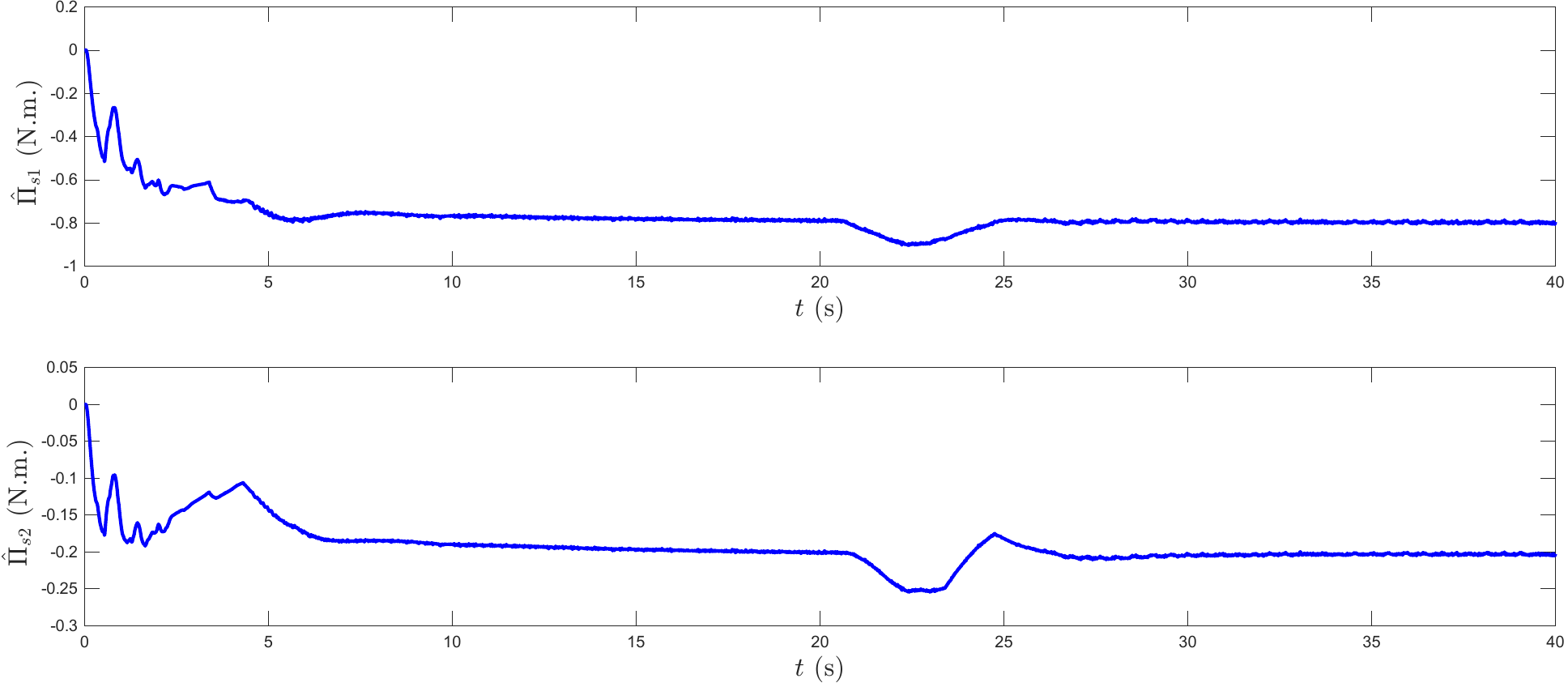}
    \caption{Estimated uncertainties of slave manipulator.}
    \label{fig:fig10}
\end{figure}
The estimated uncertainty torques for the two joints
remain bounded throughout the simulation and vary according to changes in
robot motion, external interaction, and communication conditions. This
indicates that the adaptive LSM readout continuously adjusts the uncertainty
compensation based on the temporal information contained in the liquid states.
Fig.~\ref{fig:fig11} presents the number of active neurons in the master $N_m$ and slave
$N_s$ reservoirs per sample, showing that the neural activity changes with
operating conditions rather than remaining permanently inactive or fully
saturated.
\begin{figure}[pos=!htbp]
    \centering
    \includegraphics[width=0.85\linewidth]{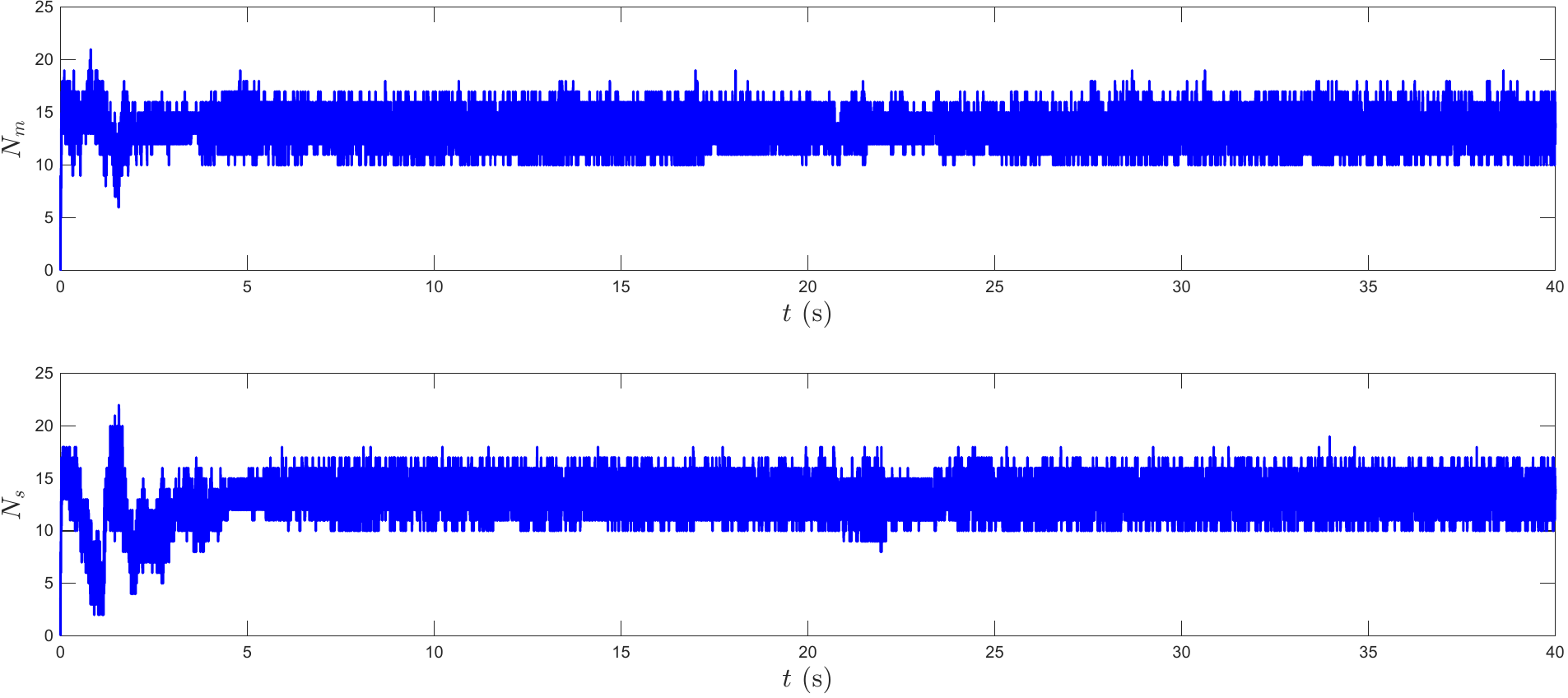}
    \caption{Number of active neurons in master and slave reservoirs $N_m,N_s$.}
    \label{fig:fig11}
\end{figure}

The adaptive readout weights are represented by their norms in Fig.~\ref{fig:fig12}, together
with the remaining adaptive variables in Figs.~\ref{fig:fig13} and~\ref{fig:fig14}.
\begin{figure}[pos=!htbp]
    \centering
    \includegraphics[width=0.85\linewidth]{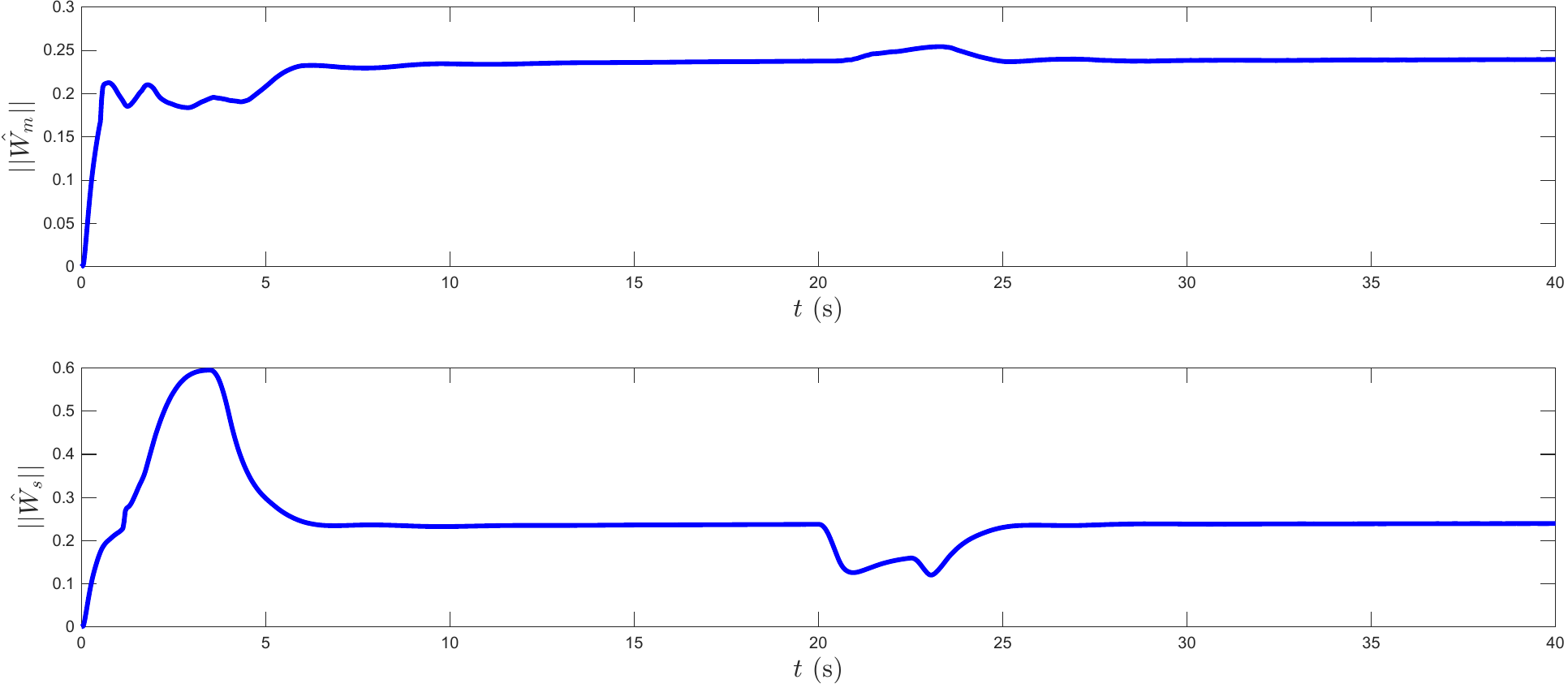}
    \caption{Norm values of adaptive readout layer weights $\hat{W_m},\hat{W_s}$.}
    \label{fig:fig12}
\end{figure}
\begin{figure}[pos=!htbp]
    \centering
    \includegraphics[width=0.85\linewidth]{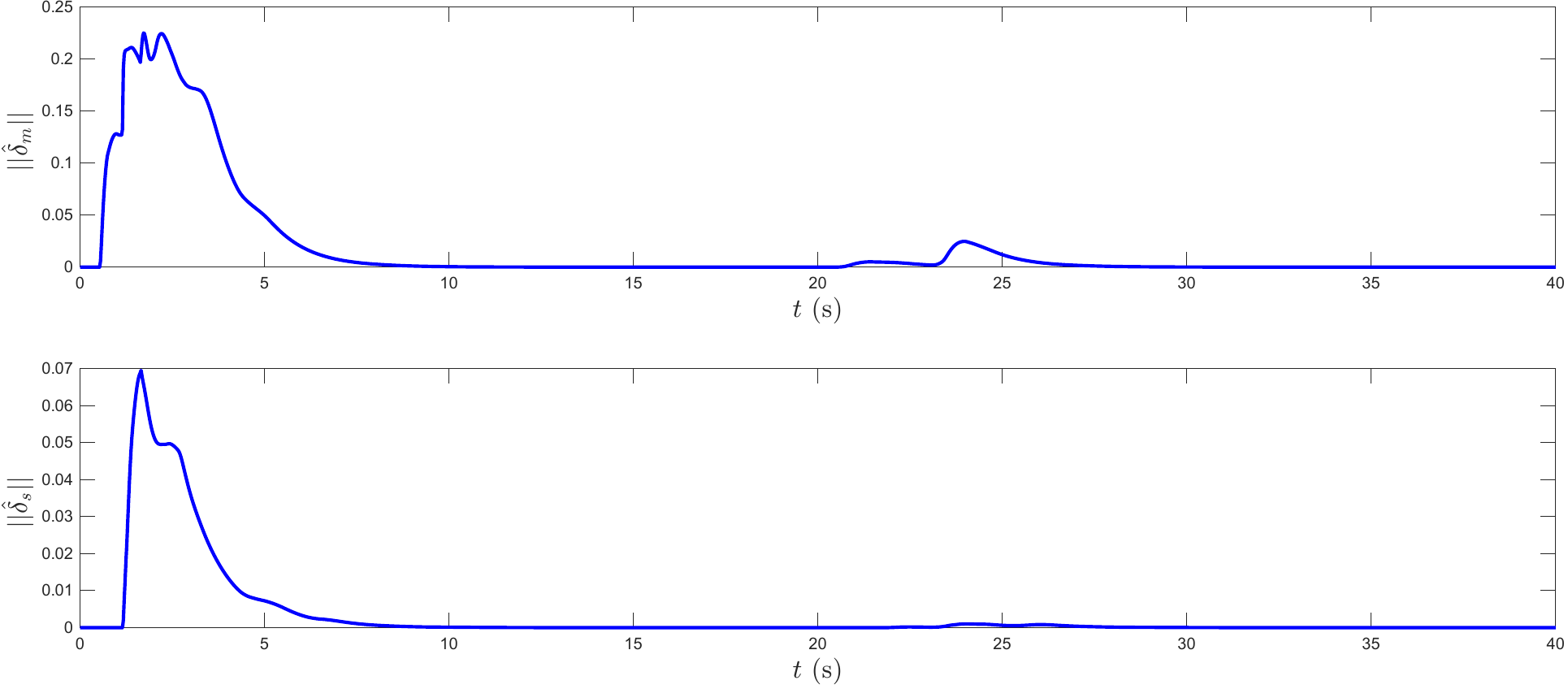}
    \caption{Adaptive estimates of the delay-rate bounds $\hat{\delta}_{j'}$.}
    \label{fig:fig13}
\end{figure}
\begin{figure}[pos=!htbp]
    \centering
    \includegraphics[width=0.85\linewidth]{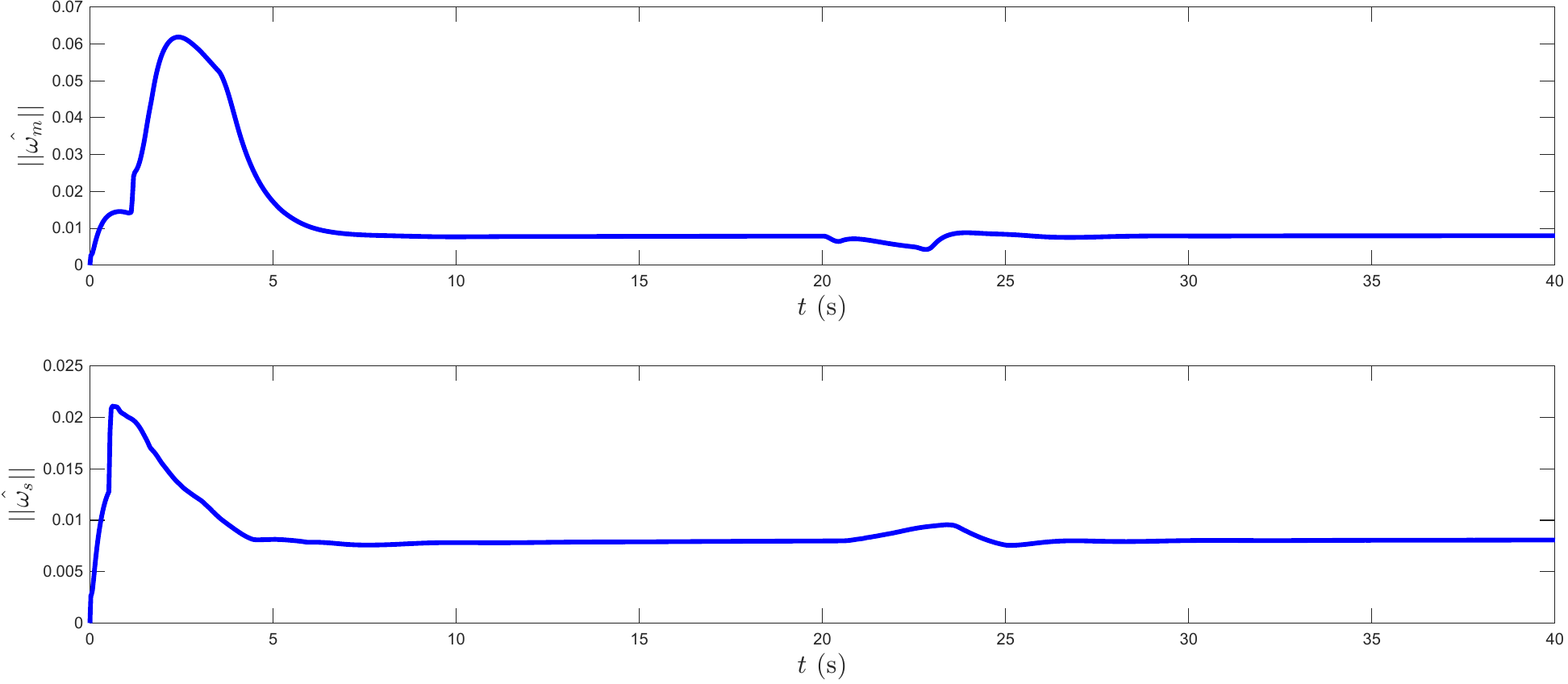}
    \caption{Adaptive parameters $\hat{\omega}_m,\hat{\omega}_s$.}
    \label{fig:fig14}
\end{figure}
All adaptive
signals remain bounded during the simulation. Fig.~\ref{fig:fig15} presents the norms of
the norms of the auxiliary variables $\zeta_m$ and $\zeta_s$ which decrease from their initial
values and remain close to zero, supporting the expected closed-loop
convergence behavior.
\begin{figure}[pos=!htbp]
    \centering
    \includegraphics[width=0.85\linewidth]{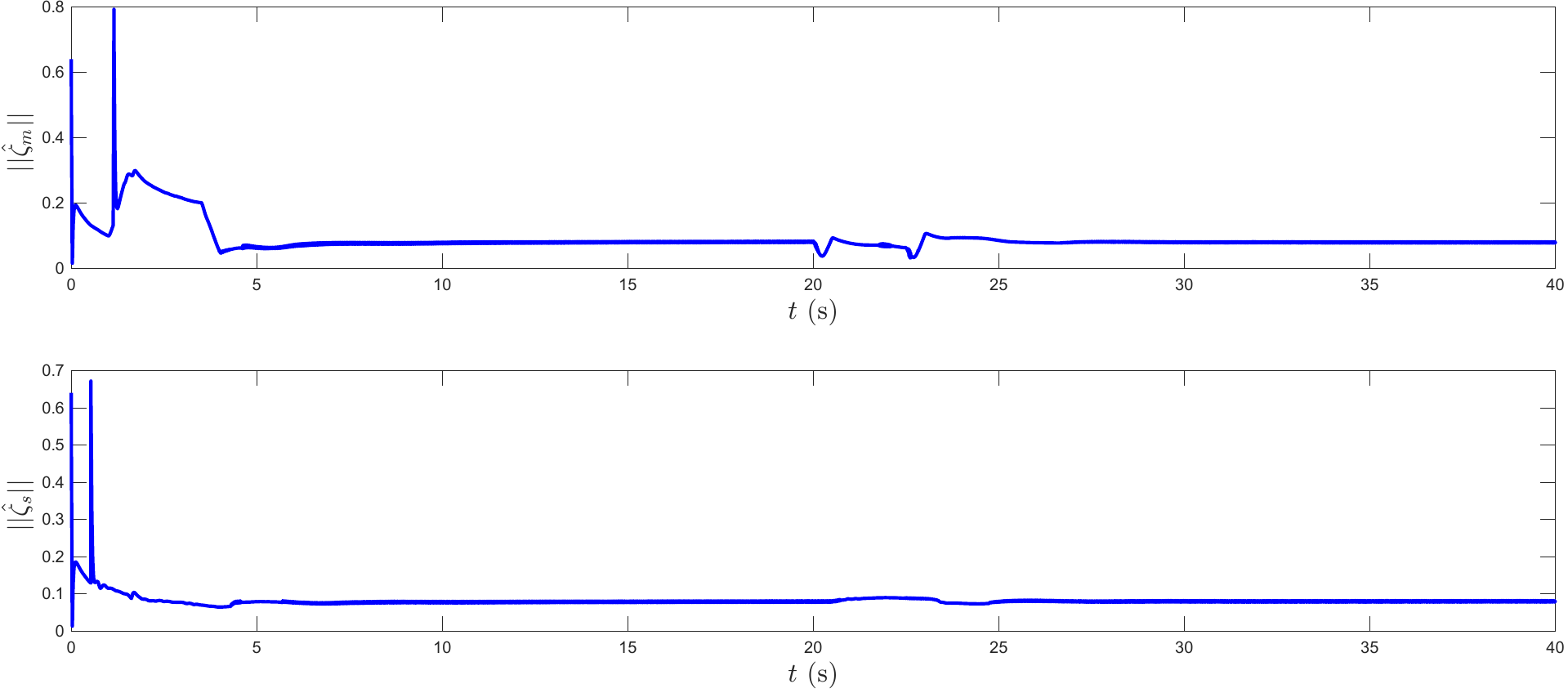}
    \caption{Norm values of $\zeta_m$ , $\zeta_s$.}
    \label{fig:fig15}
\end{figure}
Finally, the master and slave control torques shown in
Fig. 16 remain bounded and do not exhibit sustained high-frequency
oscillations.
\begin{figure}[pos=!htbp]
    \centering
    \includegraphics[width=0.85\linewidth]{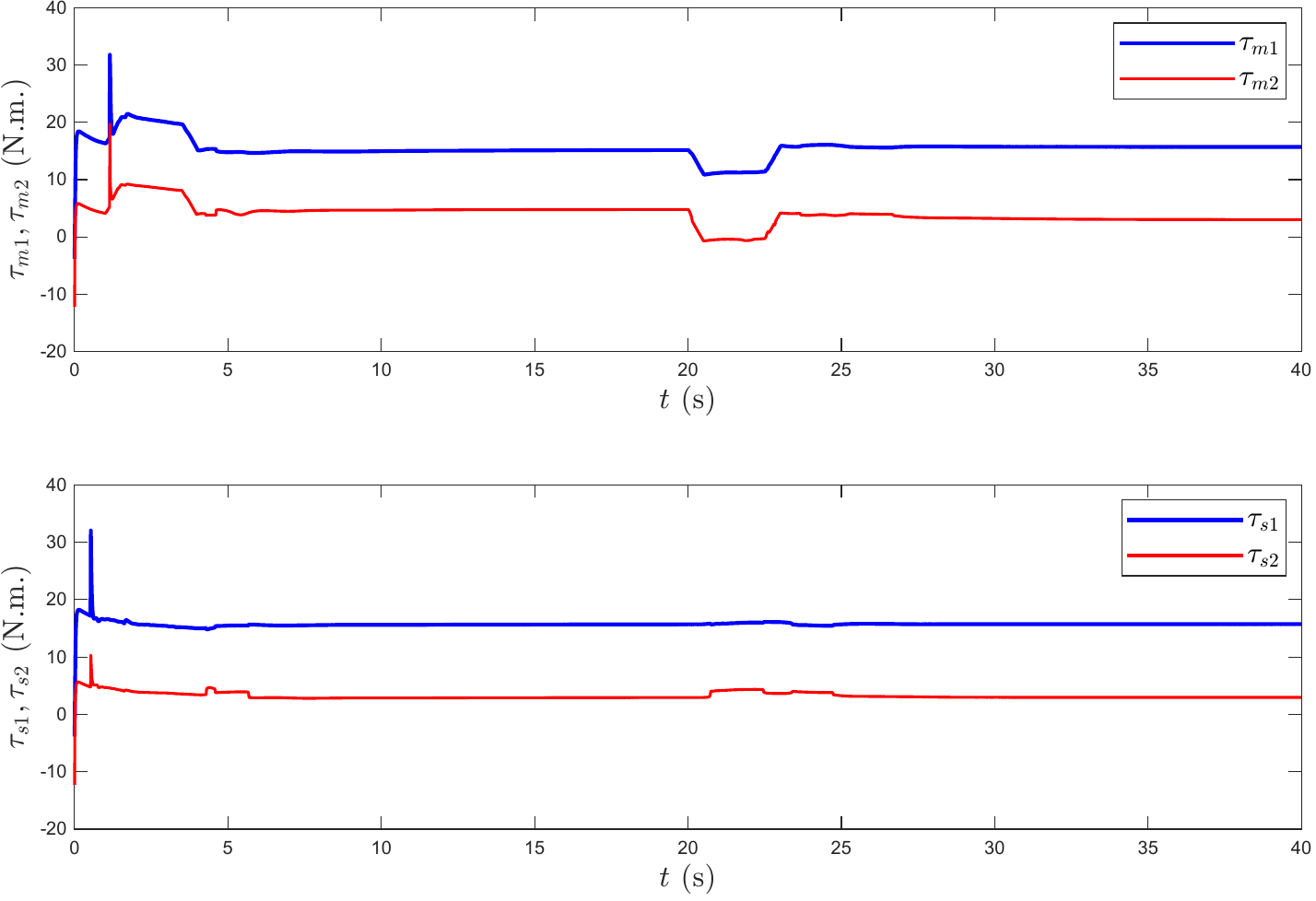}
    \caption{Master and slave manipulators input torques.}
    \label{fig:fig16}
\end{figure}
The results demonstrate stable bilateral operation, satisfactory
position/force tracking, bounded adaptive behavior, and effective online
uncertainty compensation under time-varying delays and uncertain robot
dynamics.

\subsection{Comparison with Other Control Methods}

To evaluate the effectiveness of the proposed LSM-based controller, it is
compared with the RBF-based adaptive method introduced by \cite{ID5} under identical simulation conditions. The comparison focuses on
steady-state tracking accuracy, control effort, and computational efficiency.
\cite{ID5} uses an adaptive RBFNN to estimate and compensate the
unknown robot dynamics. It further combines position/force error filters,
velocity filters, and finite-time auxiliary variables to achieve bilateral
tracking under time-varying delays and uncertainties.

The position and force errors of the proposed method and \cite{ID5}
are shown in Figs.~\ref{fig:fig17} and~\ref{fig:fig18} respectively.
\begin{figure}[pos=!htbp]
    \centering
    \begin{subfigure}{0.48\linewidth}
        \centering
        \includegraphics[width=\linewidth]{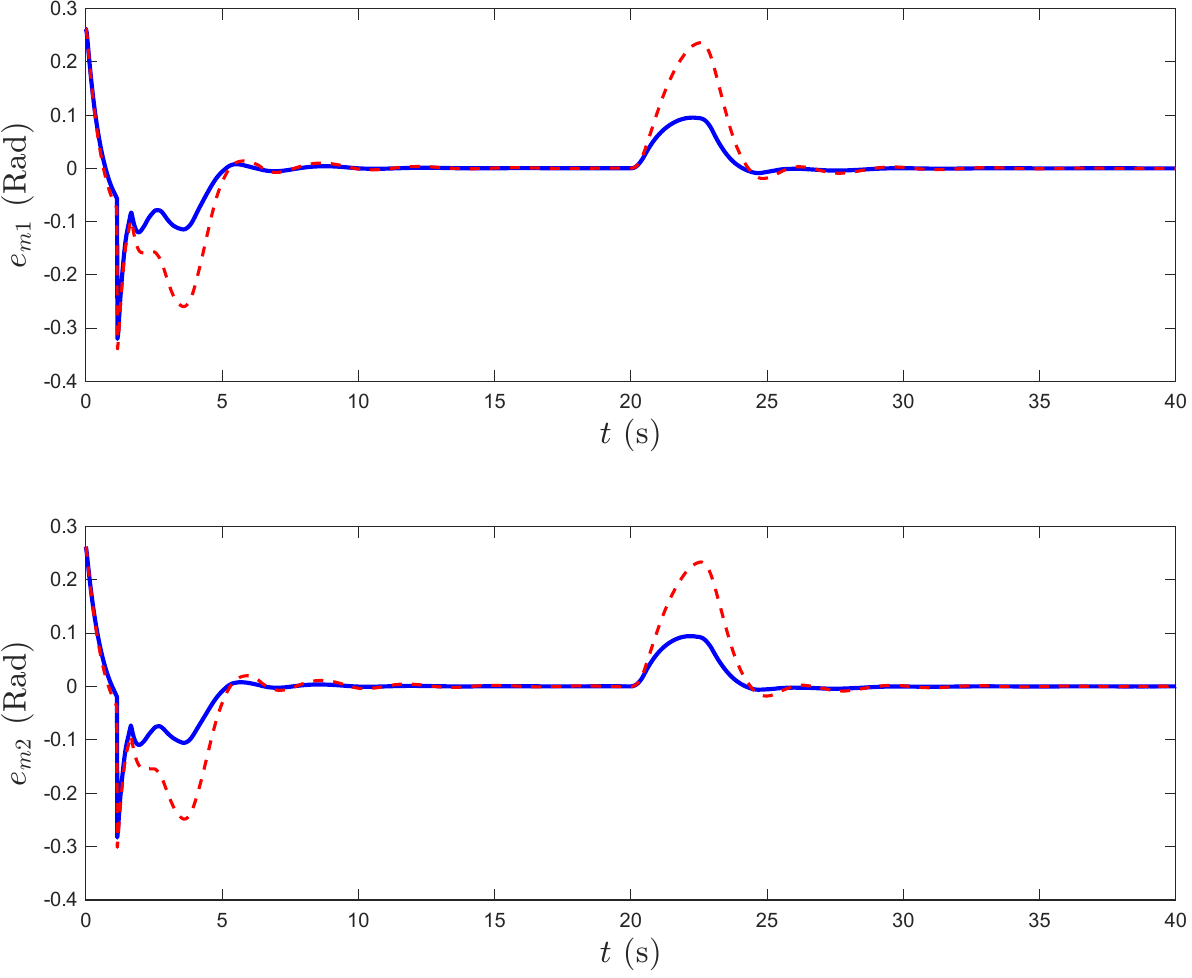}
        \caption{}
        \label{fig:fig17a}
    \end{subfigure}
    \hfill
    \begin{subfigure}{0.48\linewidth}
        \centering
        \includegraphics[width=\linewidth]{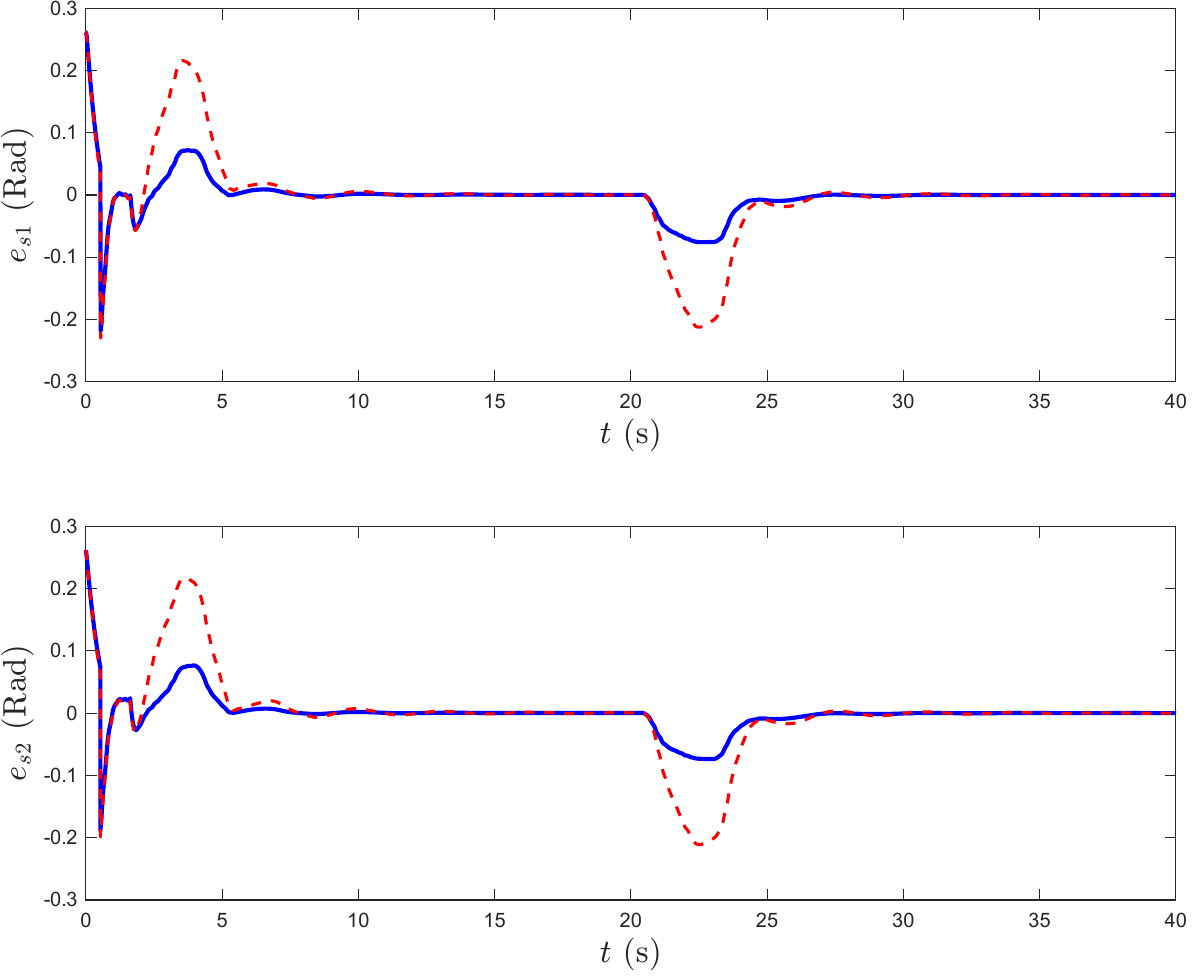}
        \caption{}
        \label{fig:fig17b}
    \end{subfigure}
    \caption{Comparison of position tracking performance with \cite{ID5}.}
    \label{fig:fig17}
\end{figure}
\begin{figure}[pos=!htbp]
    \centering
    \begin{subfigure}{0.48\linewidth}
        \centering
        \includegraphics[width=\linewidth]{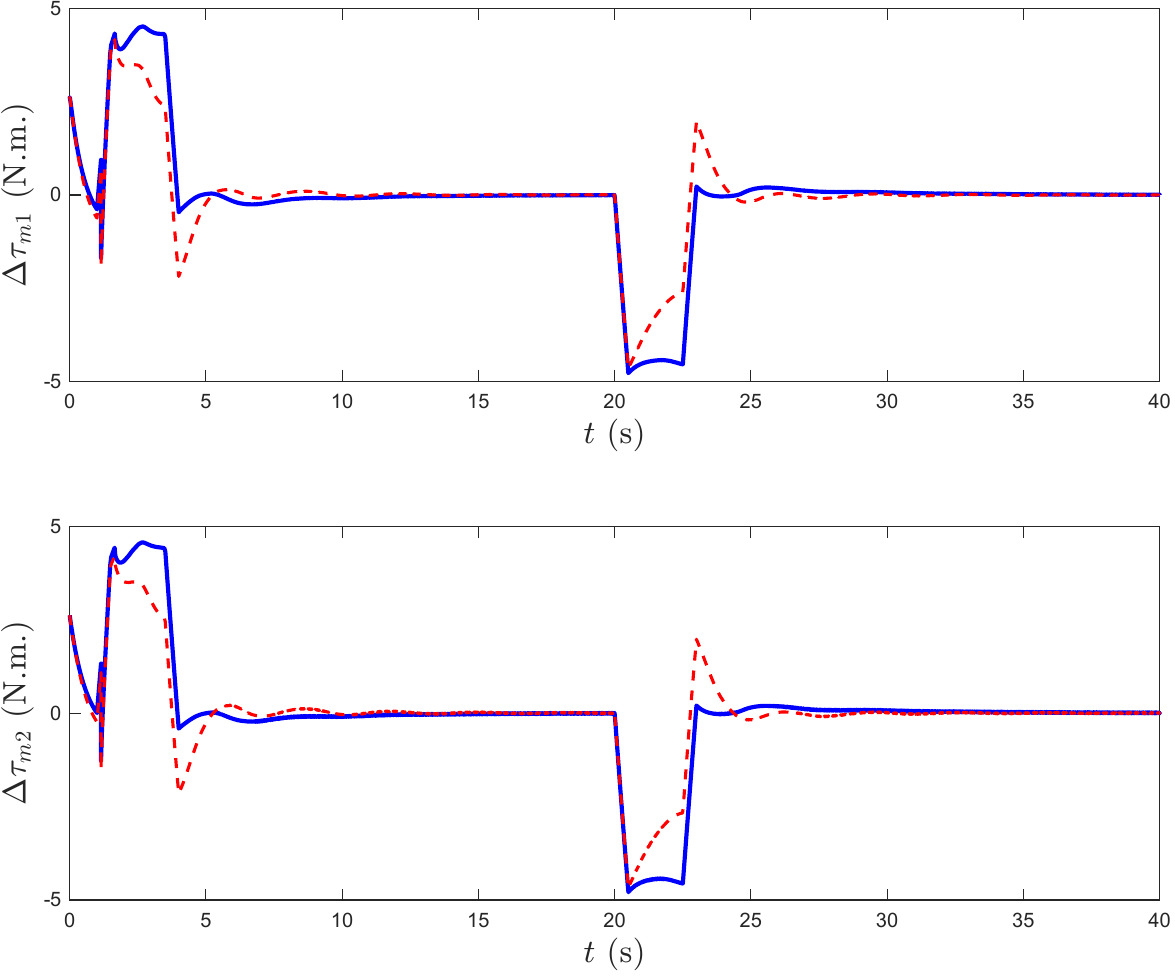}
        \caption{}
        \label{fig:fig18a}
    \end{subfigure}
    \hfill
    \begin{subfigure}{0.48\linewidth}
        \centering
        \includegraphics[width=\linewidth]{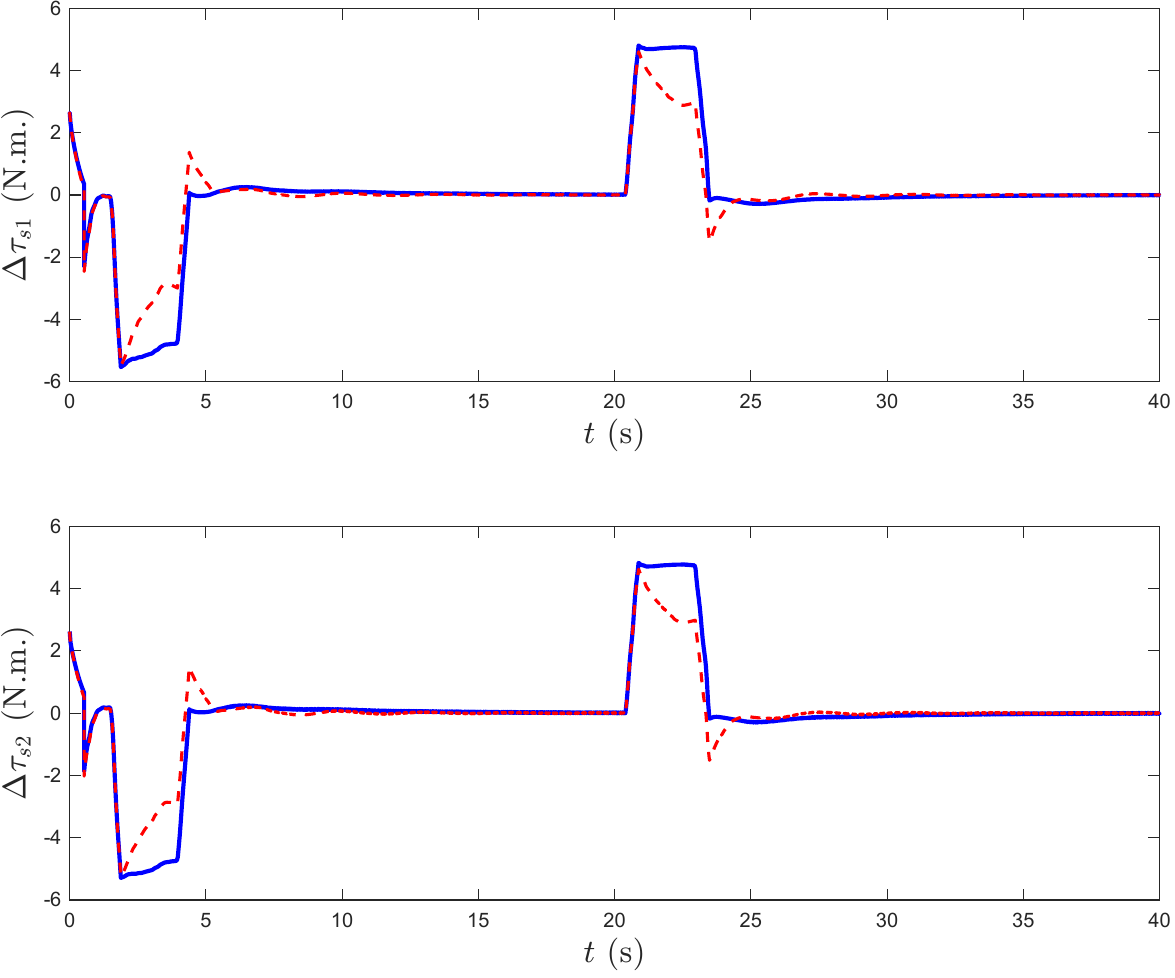}
        \caption{}
        \label{fig:fig18b}
    \end{subfigure}
    \caption{Comparison of force tracking performance with \cite{ID5}.}
    \label{fig:fig18}
\end{figure}
It can be observed from Fig.~\ref{fig:fig17}
that while in both methods error converges to a small neighborhood around zero,
the proposed method exhibits faster error reduction. The steady-state position
and force tracking performances are evaluated using the RMSE over 10-20s,
30-40s, and the entire simulation interval. As summarized in Table~\ref{tab:position_rmse} and
Table~\ref{tab:force_rmse}, the proposed LSM method achieves considerably lower position and force
RMSE values compared to \cite{ID5}, demonstrating improved
steady-state synchronization and force tracking. Because the input and
recurrent reservoir weights are randomly initialized, each simulation may
produce slightly different results. Therefore, to provide a more reliable
evaluation of the proposed controller, the results in Tables~\ref{tab:position_rmse} and~\ref{tab:force_rmse} are
reported as the mean and standard deviation over 10 simulations performed
under identical conditions. The standard deviations further indicate that the
obtained results remain consistent across different random initializations and
do not vary significantly between runs. The corresponding control torques are
shown in Fig.~\ref{fig:fig19}.
\begin{figure}[pos=!htbp]
    \centering
    \begin{subfigure}{0.48\linewidth}
        \centering
        \includegraphics[width=\linewidth]{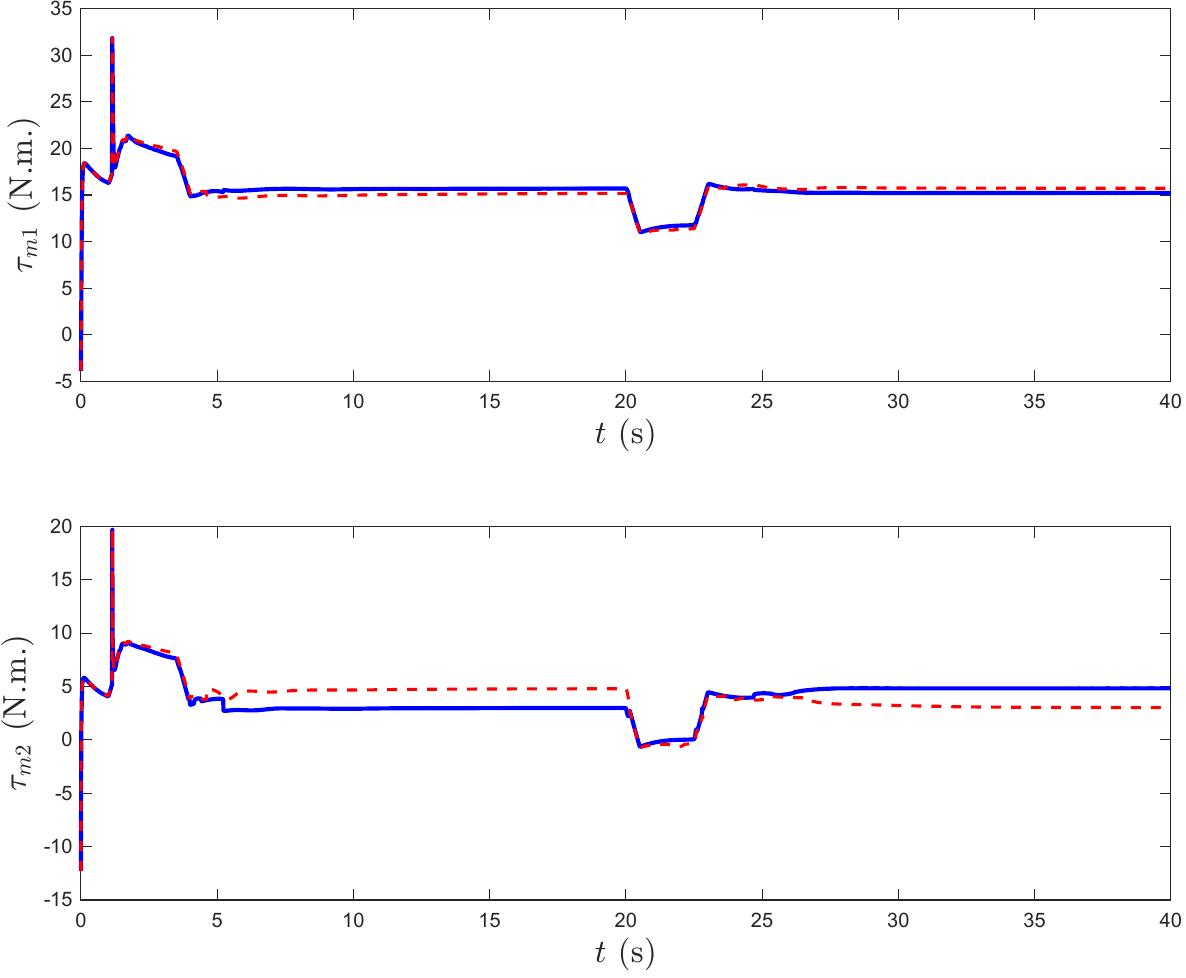}
        \caption{}
        \label{fig:fig19a}
    \end{subfigure}
    \hfill
    \begin{subfigure}{0.48\linewidth}
        \centering
        \includegraphics[width=\linewidth]{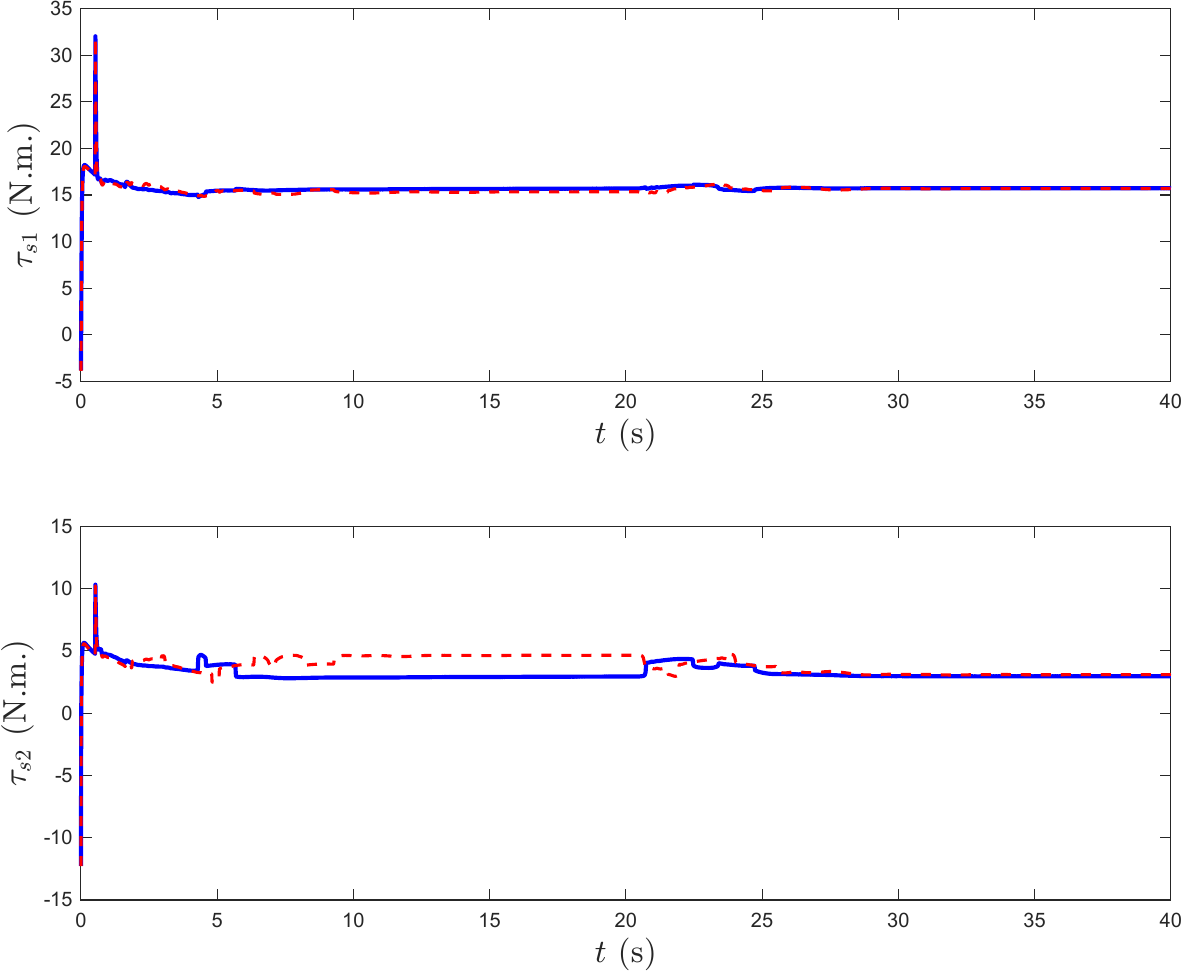}
        \caption{}
        \label{fig:fig19b}
    \end{subfigure}
    \caption{Comparison of input torques with \cite{ID5}.}
    \label{fig:fig19}
\end{figure}
Although both methods achieve satisfactory tracking
performance, the LSM-based controller produces lower joint torques, indicating
that the improved tracking accuracy is obtained with reduced control effort.
Finally, the computational requirements of the two neural approximation
schemes are compared to assess their suitability for online implementation.
All simulations are performed using the same computer, MATLAB/Simulink
configuration and conditions. The simulations are executed on a system
equipped with an Intel Core 17-10870H processor operating at 2.20 GHz, 32 GB
of RAM using MATLAB/Simulink version 2023b Each controller is simulated ten
times under identical conditions. The average execution time is then
calculated. The mean execution time of the proposed LSM-based controller is
2.7205 s, whereas \cite{ID5} requires 3.2803 s. The LSM runs
faster mainly because it relies on simple matrix operations and spike-based
updates, while the RBF network repeatedly evaluates more computationally expensive Gaussian exponential functions. In
addition, only a subset of LSM neurons is typically active at each instant.
Therefore, the proposed LSM achieves superior tracking accuracy with lower
control effort and reduced computational time, indicating greater potential
for online implementation.
\begin{table}[pos=!htbp]
\caption{Comparison of position tracking error's RMSE with \cite{ID5}}
\label{tab:position_rmse}
\begin{tabular}{ccccc}
\toprule
RMSE ($10^{-4}$ Rad) & During & Proposed & RBF Method \\
\midrule
$e_{m1}$ & $10s-20s$ & $5.36 \pm 0.04$ & $14.34$ \\
& $30s-40s$ & $2.24 \pm 0.04$ & $7.35$ \\
& $0-40s$ & $43.78 \pm 0.10$ & $63.54$ \\
\midrule
$e_{m2}$ & $10s-20s$ & $4.07 \pm 0.04$ & $12.67$ \\
& $30s-40s$ & $2.28 \pm 0.04$ & $7.41$ \\
& $0-40s$ & $41.67 \pm 0.10$ & $63.69$ \\
\midrule
$e_{s1}$ & $10s-20s$ & $6.11 \pm 0.03$ & $11.12$ \\
& $30s-40s$ & $1.89 \pm 0.04$ & $6.03$ \\
& $0-40s$ & $19.54 \pm 0.10$ & $40.74$ \\
\midrule
$e_{s2}$ & $10s-20s$ & $4.32 \pm 0.06$ & $9.67$ \\
& $30s-40s$ & $1.12 \pm 0.06$ & $5.89$ \\
& $0-40s$ & $19.27 \pm 0.10$ & $40.82$ \\
\bottomrule
\end{tabular}

\vspace{1mm}

\noindent{\footnotesize mean $\pm$ stddev., $N=10$\par}
\end{table}
\begin{table}[pos=!htbp]
\caption{Comparison of force tracking error's RMSE with \cite{ID5}}
\label{tab:force_rmse}
\begin{tabular}{ccccc}
\toprule
RMSE ($10^{-2}$ N.m) & During & Proposed & RBF Method \\
\midrule
$\Delta\tau_{m1}$ & $10s-20s$ & $0.74 \pm 0.02$ & $2.19$ \\
& $30s-40s$ & $0.42 \pm 0.01$ & $1.97$ \\
& $0-40s$ & $11.11 \pm 0.10$ & $10.57$ \\
\midrule
$\Delta\tau_{m2}$ & $10s-20s$ & $1.14 \pm 0.03$ & $5.51$ \\
& $30s-40s$ & $1.75 \pm 0.04$ & $5.46$ \\
& $0-40s$ & $11.12 \pm 0.20$ & $10.58$ \\
\midrule
$\Delta\tau_{s1}$ & $10s-20s$ & $0.53 \pm 0.02$ & $1.72$ \\
& $30s-40s$ & $0.31 \pm 0.02$ & $1.34$ \\
& $0-40s$ & $11.17 \pm 0.01$ & $12.11$ \\
\midrule
$\Delta\tau_{s2}$ & $10s-20s$ & $0.77 \pm 0.02$ & $3.56$ \\
& $30s-40s$ & $0.67 \pm 0.02$ & $3.32$ \\
& $0-40s$ & $11.19 \pm 0.10$ & $12.16$ \\
\bottomrule
\end{tabular}

\vspace{1mm}

\noindent{\footnotesize mean $\pm$ stddev., $N=10$\par}
\end{table}
\subsection{Temporal-Memory Evaluation}
\label{ssec:tempmemeval}

To further investigate the temporal estimation capability of the proposed LSM, an additional simulation
is conducted using a bounded history-dependent viscoelastic environment. Following the generalized
Maxwell model introduced in \cite{zhang2007} for representing the relaxation and hysteretic behavior of soft
biological tissues, the environment is modified to include internal dynamic states. Such a model provides
a realistic interaction scenario for applications involving soft tissues, compliant materials, and
deformable objects, where the interaction force depends not only on the instantaneous deformation
but also on its previous history. The generalized Maxwell formulation in [1] consists of an equilibrium
spring connected in parallel with several Maxwell branches having different relaxation time scales.

For the considered two-DOF teleoperation system, a two-branch joint-space representation is adopted
as

\begin{equation}
\tau_e = F_e^{*} + D_0\dot{q}_s + K_{\infty}q_s + z_1 + z_2 ,
\label{eq:maxwell_environment}
\end{equation}

where $K_{\infty}=6I$ denotes the relaxed environment stiffness and $z_1,z_2$ represent the internal viscoelastic
torque states. Their dynamics are defined as

\begin{equation}
\dot{z}_1 = K_1\dot{q}_s - \frac{1}{\tau_1}z_1,
\qquad
\dot{z}_2 = K_2\dot{q}_s - \frac{1}{\tau_2}z_2,
\label{eq:maxwell_states}
\end{equation}
where $K_1=2.5I$ and $K_2=1.5I$ determine the contribution of the two viscoelastic branches, while $\tau_1=0.1s$ and $\tau_2=0.5s$ determine their relaxation time scales. Consequently, identical instantaneous values of $q_s$ and $\dot{q}_s$ may correspond to different environment torques depending on the preceding interaction trajectory. This provides a suitable scenario for evaluating the fading-memory capability of the LSM.

To ensure a fair comparison, no modification is made to either the proposed LSM-based controller or the \cite{ID5} compared to simulations before. only the environment model is replaced by the introduced history-dependent viscoelastic model. Therefore, the obtained performance differences directly reflect the ability of the uncertainty approximators to cope with temporally dependent interaction dynamics.

The comparison focuses on steady-state tracking accuracy and control effort. The position and force errors
of the proposed method and \cite{ID5} in the new environment are shown in Figs.~\ref{fig:fig20} and~\ref{fig:fig21}
respectively.
\begin{figure}[pos=!htbp]
    \centering
    \includegraphics[width=\linewidth]{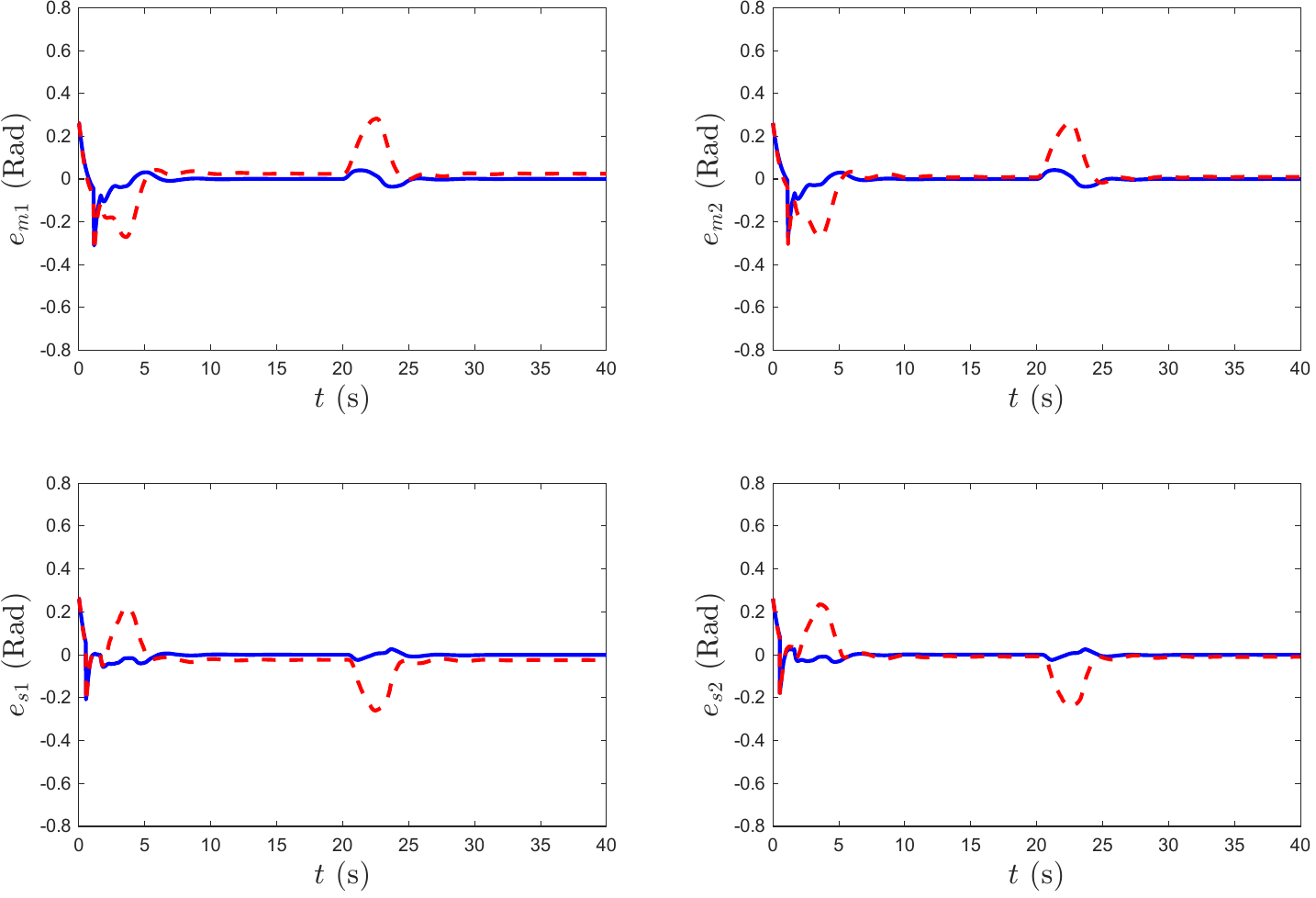}
    \caption{Position tracking error of master and slave manipulators in the new environment.}
    \label{fig:fig20}
\end{figure}
\begin{figure}[pos=!htbp]
    \centering
    \includegraphics[width=\linewidth]{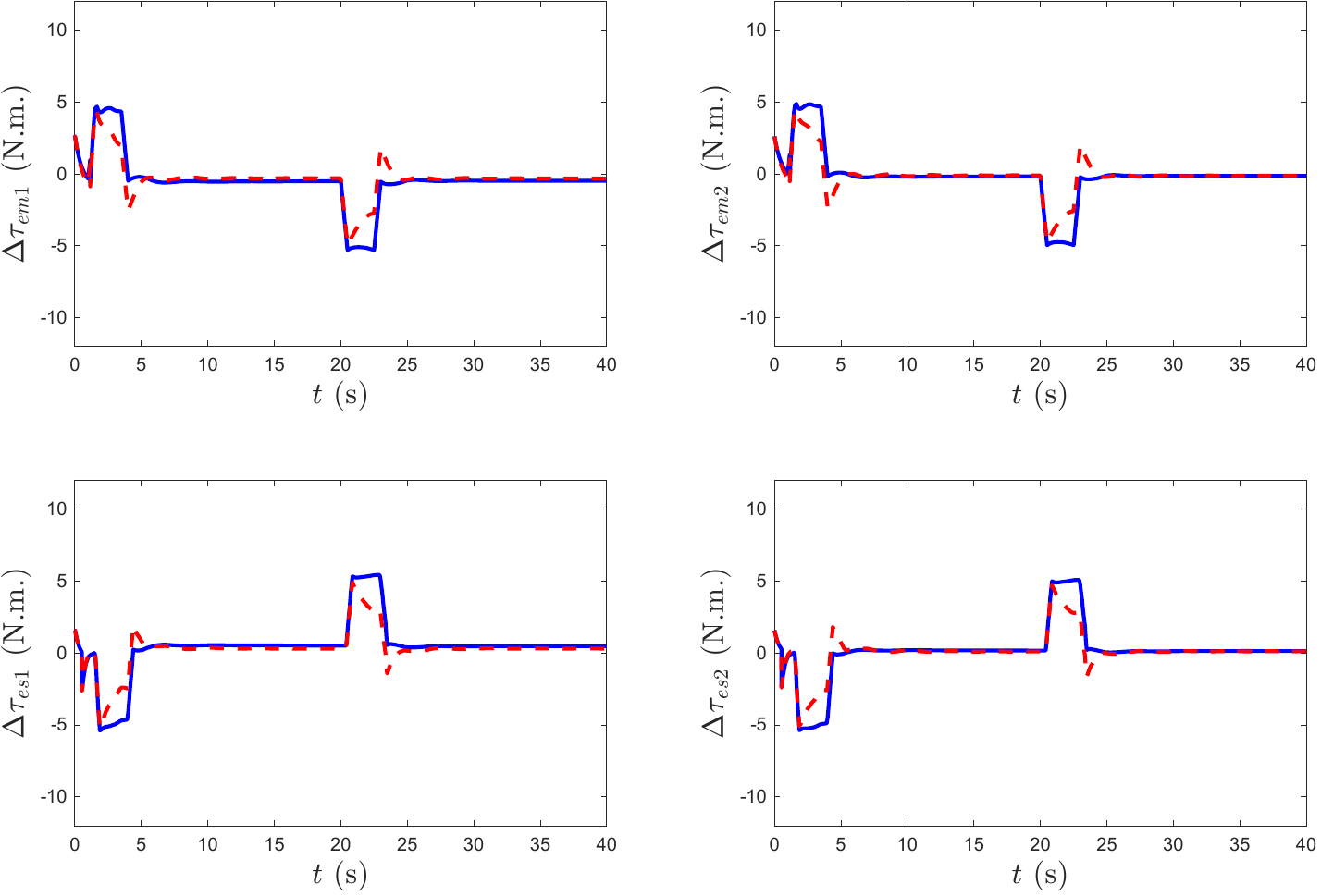}
    \caption{Comparison of force tracking performance with \cite{ID5} in the new environment.}
    \label{fig:fig21}
\end{figure}
Fig.~\ref{fig:fig22} depicts the control signals of the controllers.
\begin{figure}[pos=!htbp]
    \centering
    \begin{subfigure}{0.48\linewidth}
        \centering
        \includegraphics[width=\linewidth]{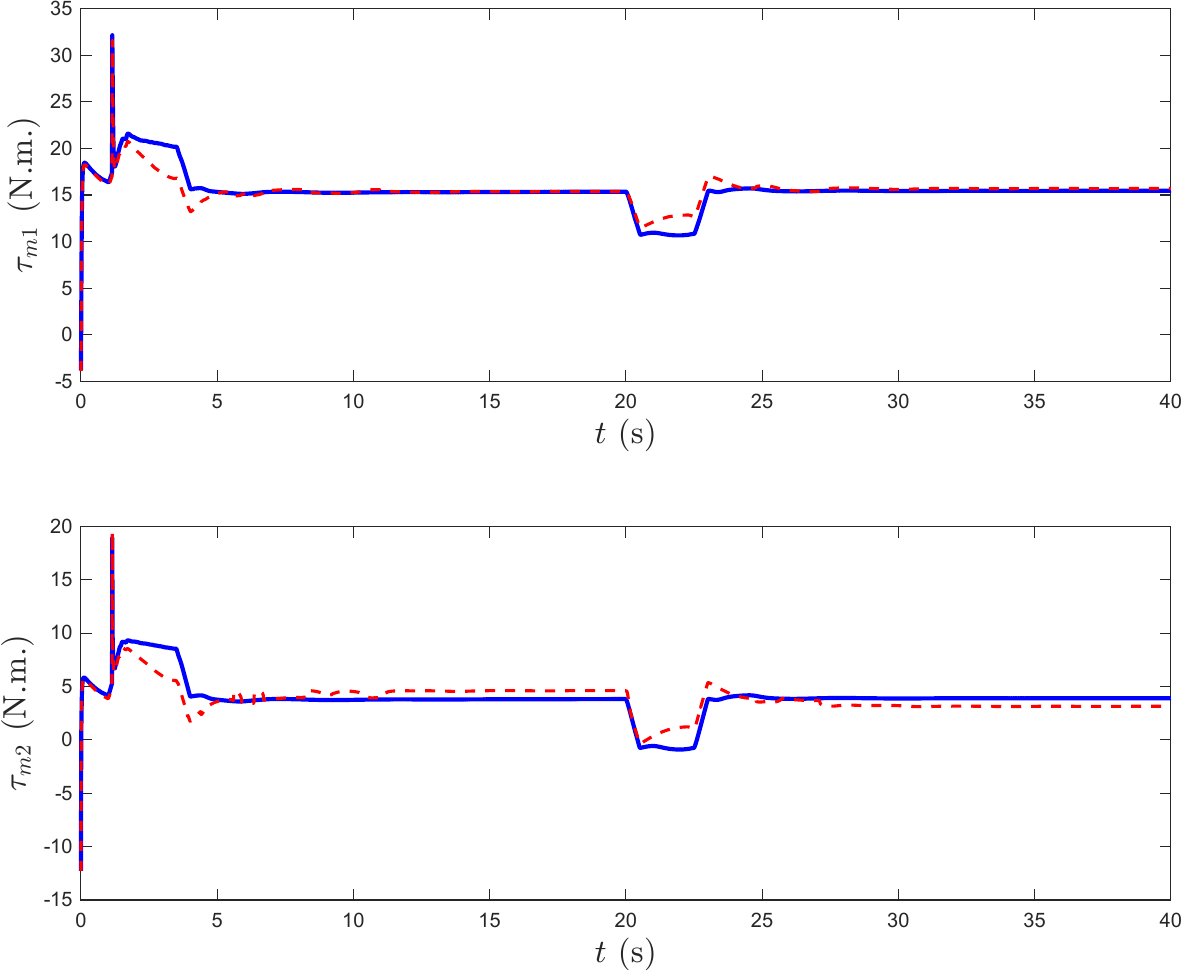}
        \caption{}
        \label{fig:fig22a}
    \end{subfigure}
    \hfill
    \begin{subfigure}{0.48\linewidth}
        \centering
        \includegraphics[width=\linewidth]{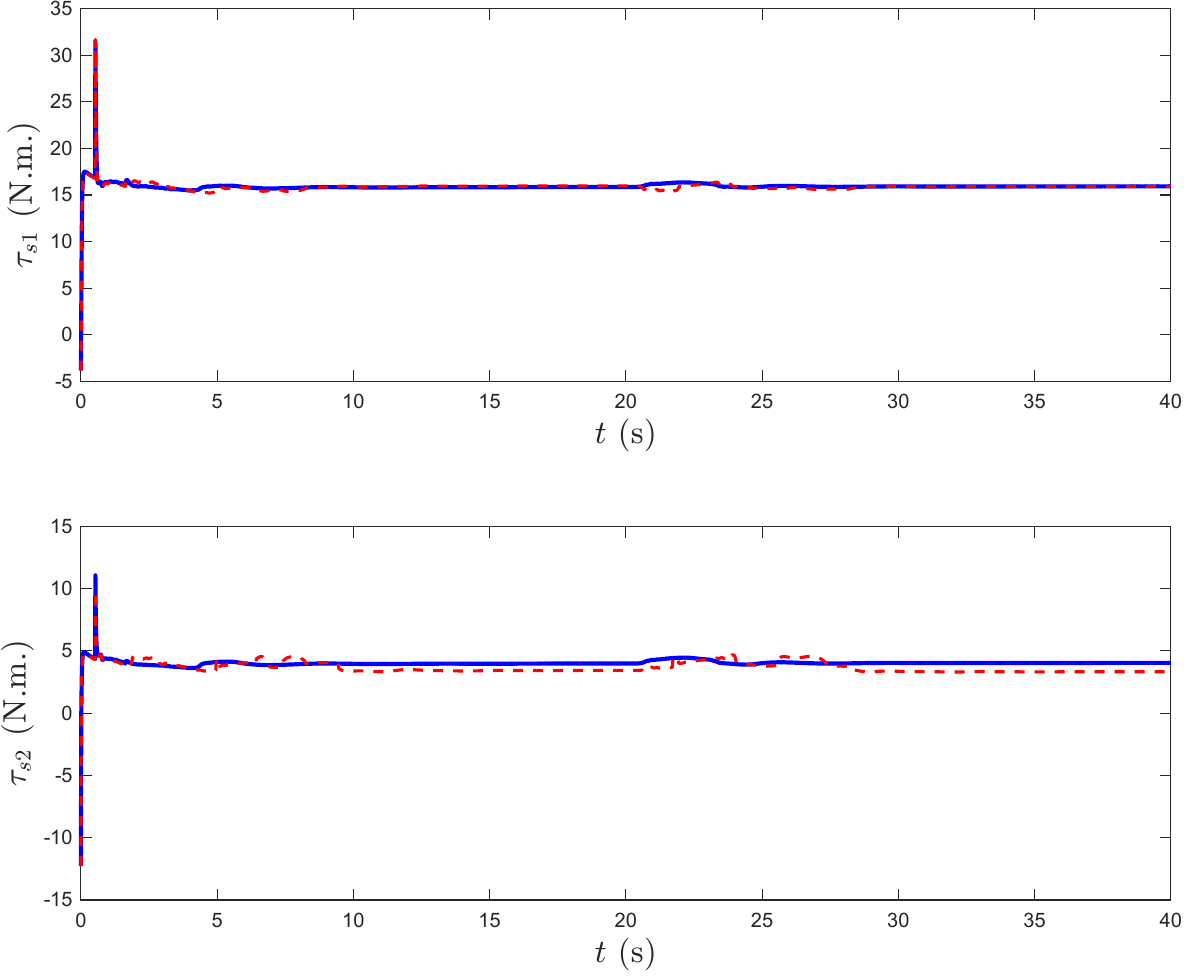}
        \caption{}
        \label{fig:fig22b}
    \end{subfigure}
    \caption{Comparison of input torques with \cite{ID5} in the new environment.}
    \label{fig:fig22}
\end{figure}
The steady-state position and force
tracking performances of the new environment are evaluated using the RMSE over 10-20s, 30-40s, and
the entire simulation interval, as summarized in Table~\ref{tab:position_rmse_new_environment} and Table~\ref{tab:force_rmse_new_environment}.

The results show that the proposed method achieves significantly lower position and force tracking
errors. Notably, these results are obtained without modifying the controller structure or tuning its
parameters, indicating that the proposed method can adapt effectively to the newly introduced history-
dependent environment because the intrinsic fading-memory capability of the LSM, enables it to
capture temporal features of the viscoelastic environment. In contrast, the RBFNN-based method of
\cite{ID5} exhibits degraded tracking performance because RBFNN approximator relies primarily on
the instantaneous information and has limited capability to represent the introduced history-dependent
dynamics of the environment.

\begin{table}[!htbp]
\caption{Comparison of position tracking error's RMSE with \cite{ID5} in the new environment}
\label{tab:position_rmse_new_environment}
\begin{tabular}{cccc}
\toprule
RMSE ($10^{-4}$ Rad) & During & Proposed & RBF Method \\
\midrule
$e_{m1}$ & $10s-20s$ & $7.76 \pm 0.04$ & $251.19$ \\
& $30s-40s$ & $7.75 \pm 0.03$ & $247.67$ \\
& $0-40s$ & $33.14 \pm 0.10$ & $900.08$ \\
\midrule
$e_{m2}$ & $10s-20s$ & $5.45 \pm 0.05$ & $945.29$ \\
& $30s-40s$ & $3.76 \pm 0.04$ & $935.57$ \\
& $0-40s$ & $31.35 \pm 0.10$ & $865.39$ \\
\midrule
$e_{s1}$ & $10s-20s$ & $5.41 \pm 0.03$ & $223.67$ \\
& $30s-40s$ & $6.35 \pm 0.02$ & $564.68$ \\
& $0-40s$ & $17.65 \pm 0.10$ & $966.67$ \\
\midrule
$e_{s2}$ & $10s-20s$ & $6.33 \pm 0.06$ & $842.21$ \\
& $30s-40s$ & $3.47 \pm 0.04$ & $926.10$ \\
& $0-40s$ & $19.19 \pm 0.10$ & $965.63$ \\
\bottomrule
\end{tabular}

\vspace{1mm}

\noindent{\footnotesize mean $\pm$ stddev., $N=10$\par}
\end{table}

\begin{table}[!htbp]
\caption{Comparison of force tracking error's RMSE with \cite{ID5} in the new environment}
\label{tab:force_rmse_new_environment}
\begin{tabular}{cccc}
\toprule
RMSE ($10^{-2}$ N.m) & During & Proposed & RBF Method \\
\midrule
$\Delta\tau_{m1}$ & $10s-20s$ & $0.52 \pm 0.02$ & $130.76$ \\
& $30s-40s$ & $0.47 \pm 0.01$ & $132.57$ \\
& $0-40s$ & $15.35 \pm 0.10$ & $85.01$ \\
\midrule
$\Delta\tau_{m2}$ & $10s-20s$ & $1.79 \pm 0.03$ & $111.13$ \\
& $30s-40s$ & $1.37 \pm 0.04$ & $121.93$ \\
& $0-40s$ & $13.86 \pm 0.10$ & $96.32$ \\
\midrule
$\Delta\tau_{s1}$ & $10s-20s$ & $0.53 \pm 0.02$ & $197.80$ \\
& $30s-40s$ & $0.52 \pm 0.02$ & $131.12$ \\
& $0-40s$ & $14.41 \pm 0.10$ & $101.43$ \\
\midrule
$\Delta\tau_{s2}$ & $10s-20s$ & $0.77 \pm 0.02$ & $103.76$ \\
& $30s-40s$ & $1.75 \pm 0.02$ & $122.68$ \\
& $0-40s$ & $12.88 \pm 0.10$ & $94.43$ \\
\bottomrule
\end{tabular}

\vspace{1mm}

\noindent{\footnotesize mean $\pm$ stddev., $N=10$\par}
\end{table}
\section{Conclusion}
\label{sec:conclusion}
This article addressed the position and force tracking problem in bilateral teleoperation systems subject to asymmetric time-varying communication delays and uncertain nonlinear dynamics. A finite-time adaptive control scheme was developed by integrating hybrid position/force error dynamics, auxiliary variables, and filtering mechanisms with an LSM-based uncertainty estimator. By exploiting the intrinsic temporal processing capability of the LSM, the proposed controller can account for history-dependent uncertainties while retaining a simple online adaptation structure. Closed-loop stability and finite-time convergence of the tracking errors were established through a Lyapunov--Krasovskii analysis.

The effectiveness of the proposed method was evaluated in both spring--damper and history-dependent generalized Maxwell environments. In the spring--damper environment, the proposed method achieved average position and force RMSEs of \(31.09 \times 10^{-4}\) Rad and \(11.14 \times 10^{-2}\) N.m., respectively, compared with \(52.19 \times 10^{-4}\) Rad and \(11.35 \times 10^{-2}\) N.m. for the RBFNN-based method \cite{ID5}. In the generalized Maxwell environment, the corresponding RMSEs were \(25.33 \times 10^{-4}\) Rad and \(14.12 \times 10^{-2}\) N.m. for the proposed method and \(924.44 \times 10^{-4}\) Rad and \(94.29 \times 10^{-2}\) N.m. for the RBFNN-based controller. Moreover, the proposed method achieved a lower mean execution time of \(2.7205\) s, compared with \(3.2803\) s for the RBFNN-based method. In particular, the generalized Maxwell environment highlighted the advantage of the LSM in capturing temporal interaction effects without modifying or retuning the controller structure. Thus, the results demonstrate improved tracking accuracy and computational efficiency and indicate that the proposed LSM-based framework is well suited to teleoperation tasks involving memory-dependent and dynamically evolving interaction uncertainties.

Future work will proceed in several directions. Extending the design to a fixed-time formulation would render the settling-time bound independent of the initial condition. The control and filter gains and reservoir hyperparameters were selected manually; treating the conditions of Theorem~\ref{thm:main} as constraints and searching the remaining freedom by genetic algorithms \cite{Arabsorkhi2024} or reinforcement learning \cite{Attarzadeh2026} may further improve tracking accuracy. The sparse per-step activity of the estimator makes a neuromorphic implementation a natural next step, which would also allow the reported computational advantage to be assessed on embedded hardware. Broader validation across hysteretic and contact-rich environment models, and against rate-based reservoirs such as echo state networks, would clarify the extent to which the observed gains depend on spiking dynamics rather than reservoir structure in general.

\FloatBarrier


\nocite{*} 
\bibliographystyle{elsarticle-num-names}
\bibliography{combined_references}

@article{NEWAPP1,
  author  = {Theodora Kastritsi and Theofanis Prapavesis Semetzidis
             and Zoe Doulgeri},
  title   = {Passive Bilateral Surgical Teleoperation With {RCM}
             and Spatial Constraints in the Presence of Time Delays},
  journal = {IEEE Transactions on Robotics},
  year    = {2025},
  volume  = {41},
  pages   = {612--627},
  doi     = {10.1109/TRO.2024.3502221}
}

@article{NEWAPP2,
  author  = {Edgar M. Hidalgo and Mats Isaksson
             and Mariadas Capsran Roshan and Thomas H. Marwick
             and Leah Wright and Gavin Lambert},
  title   = {Evaluating the Impacts of Network Latency, Haptics,
             and Ergonomics in a Haptically-Enabled Robot for
             Teleoperated Echocardiography},
  journal = {Computers in Biology and Medicine},
  year    = {2025},
  volume  = {195},
  pages   = {110450},
  doi     = {10.1016/j.compbiomed.2025.110450}
}

@article{NEWAPP3,
  author  = {He Li and Shuxiang Guo and Ruijie He and Hanze Wang
             and Masahiko Kawanishi},
  title   = {A Home-Based Upper Limb Rehabilitation System via
             Cloud-Based Teleoperation and {sEMG}-Driven Bilateral Control},
  journal = {IEEE Internet of Things Journal},
  year    = {2025},
  volume  = {12},
  number  = {17},
  pages   = {35510--35521},
  doi     = {10.1109/JIOT.2025.3578857}
}

@article{NEWAPP4,
  author  = {Beatrice Luciani and Alex van den Berg and Matti Lang
             and Alexandre L. Ratschat and Laura Marchal-Crespo},
  title   = {Intuitive Therapist Robot Patient Physical Interaction
             Is Worth a Thousand Words},
  journal = {Scientific Reports},
  year    = {2026},
  note    = {Published online 21 July 2026},
  doi     = {10.1038/s41598-026-63191-x}
}

@article{NEW1,
  author  = {R. J. Anderson and M. W. Spong},
  title   = {Bilateral Control of Teleoperators with Time Delay},
  journal = {IEEE Transactions on Automatic Control},
  year    = {1989},
  volume  = {34},
  number  = {5},
  pages   = {494--501},
  month   = may,
  doi     = {10.1109/9.24201}
}

@article{NEW2,
  author  = {G. Niemeyer and J.-J. E. Slotine},
  title   = {Stable Adaptive Teleoperation},
  journal = {IEEE Journal of Oceanic Engineering},
  year    = {1991},
  volume  = {16},
  number  = {1},
  pages   = {152--162},
  month   = jan,
  doi     = {10.1109/48.64895}
}

@article{NEW3,
  author  = {S. Munir and W. J. Book},
  title   = {Internet-Based Teleoperation Using Wave Variables with Prediction},
  journal = {IEEE/ASME Transactions on Mechatronics},
  year    = {2002},
  volume  = {7},
  number  = {2},
  pages   = {124--133},
  month   = jun,
  doi     = {10.1109/TMECH.2002.1011249}
}

@article{NEW4,
  author  = {Z. Chen and F. H. Huang and W. Song and S. Q. Zhu},
  title   = {A Novel Wave-Variable Based Time-Delay Compensated Four-Channel Control Design for Multilateral Teleoperation System},
  journal = {IEEE Access},
  year    = {2018},
  volume  = {6},
  pages   = {25506--25516},
  doi     = {10.1109/ACCESS.2018.2829601}
}

@article{NEW5,
  author  = {Emanuel Slawi{\'n}ski and Vicente Mut},
  title   = {{PD}-Like Controllers for Delayed Bilateral Teleoperation of Manipulators Robots},
  journal = {International Journal of Robust and Nonlinear Control},
  year    = {2015},
  volume  = {25},
  number  = {12},
  pages   = {1801--1815},
  month   = aug,
  doi     = {10.1002/rnc.3177}
}

@article{NEW6,
  author  = {Saba Al-Wais and Suiyang Khoo and Tae Hee Lee and Lakshmanan Shanmugam and Saeid Nahavandi},
  title   = {Robust {$H_\infty$} Cost Guaranteed Integral Sliding Mode Control for the Synchronization Problem of Nonlinear Tele-Operation System with Variable Time-Delay},
  journal = {ISA Transactions},
  year    = {2018},
  volume  = {72},
  pages   = {25--36},
  month   = jan,
  doi     = {10.1016/j.isatra.2017.10.009}
}

@article{NEW7,
  author  = {Y. N. Yang and C. C. Hua and J. P. Li and X. P. Guan},
  title   = {Finite-Time Output-Feedback Synchronization Control for Bilateral Teleoperation System via Neural Networks},
  journal = {Information Sciences},
  year    = {2017},
  volume  = {406--407},
  pages   = {216--233},
  month   = sep,
  doi     = {10.1016/j.ins.2017.04.034}
}

@article{NEW8,
  author  = {Y. C. Liu and M. H. Khong},
  title   = {Adaptive Control for Nonlinear Teleoperators with Uncertain Kinematics and Dynamics},
  journal = {IEEE/ASME Transactions on Mechatronics},
  year    = {2015},
  volume  = {20},
  number  = {5},
  pages   = {2550--2562},
  month   = oct,
  doi     = {10.1109/TMECH.2015.2388555}
}

@article{NEW9,
  author  = {Z. Chen and B. Liang and T. Zhang and X. Wang and B. Zhang},
  title   = {Adaptive Bilateral Control for Nonlinear Uncertain Teleoperation with Guaranteed Transient Performance},
  journal = {Robotica},
  year    = {2016},
  volume  = {34},
  number  = {10},
  pages   = {2205--2222},
  month   = oct,
  doi     = {10.1017/S0263574714002847}
}

@article{NEW10,
  author  = {D.-H. Zhai and Y. Q. Xia},
  title   = {Adaptive Control of Semi-Autonomous Teleoperation System with Asymmetric Time-Varying Delays and Input Uncertainties},
  journal = {IEEE Transactions on Cybernetics},
  year    = {2017},
  volume  = {47},
  number  = {11},
  pages   = {3621--3633},
  month   = nov,
  doi     = {10.1109/TCYB.2016.2573798}
}

@article{NEW11,
  author  = {Z. J. Li and X. Q. Cao and N. Ding},
  title   = {Adaptive Fuzzy Control for Synchronization of Nonlinear Teleoperators with Stochastic Time-Varying Communication Delays},
  journal = {IEEE Transactions on Fuzzy Systems},
  year    = {2011},
  volume  = {19},
  number  = {4},
  pages   = {745--757},
  month   = aug,
  doi     = {10.1109/TFUZZ.2011.2143417}
}

@article{NEW12,
  author  = {D. Sun and Q. F. Liao and H. L. Ren},
  title   = {Type-2 Fuzzy Modeling and Control for Bilateral Teleoperation System with Dynamic Uncertainties and Time-Varying Delays},
  journal = {IEEE Transactions on Industrial Electronics},
  year    = {2018},
  volume  = {65},
  number  = {1},
  pages   = {447--459},
  month   = jan,
  doi     = {10.1109/TIE.2017.2719604}
}

@article{NEW13,
  author  = {C. G. Yang and X. J. Wang and Z. J. Li and Y. A. Li and C.-Y. Su},
  title   = {Teleoperation Control Based on Combination of Wave Variable and Neural Networks},
  journal = {IEEE Transactions on Systems, Man, and Cybernetics: Systems},
  year    = {2017},
  volume  = {47},
  number  = {8},
  pages   = {2125--2136},
  month   = aug,
  doi     = {10.1109/TSMC.2016.2615061}
}

@article{NEW14,
  author  = {H. Q. Wang and P. X. Liu and S. C. Liu},
  title   = {Adaptive Neural Synchronization Control for Bilateral Teleoperation Systems with Time Delay and Backlash-Like Hysteresis},
  journal = {IEEE Transactions on Cybernetics},
  year    = {2017},
  volume  = {47},
  number  = {10},
  pages   = {3018--3026},
  month   = oct,
  doi     = {10.1109/TCYB.2016.2644656}
}

@article{NEW15,
  author  = {Y. N. Yang and C. C. Hua and X. P. Guan},
  title   = {Adaptive Fuzzy Finite-Time Coordination Control for Networked Nonlinear Bilateral Teleoperation System},
  journal = {IEEE Transactions on Fuzzy Systems},
  year    = {2014},
  volume  = {22},
  number  = {3},
  pages   = {631--641},
  month   = jun,
  doi     = {10.1109/TFUZZ.2013.2269694}
}

@article{NEW16,
  author  = {D.-H. Zhai and Y. Q. Xia},
  title   = {Finite-Time Control of Teleoperation Systems with Input Saturation and Varying Time Delays},
  journal = {IEEE Transactions on Systems, Man, and Cybernetics: Systems},
  year    = {2017},
  volume  = {47},
  number  = {7},
  pages   = {1522--1534},
  month   = jul,
  doi     = {10.1109/TSMC.2016.2631601}
}

@article{NEW17,
  author  = {Y. N. Yang and C. C. Hua and X. P. Guan},
  title   = {Finite-Time Synchronization Control for Bilateral Teleoperation under Communication Delays},
  journal = {Robotics and Computer-Integrated Manufacturing},
  year    = {2015},
  volume  = {31},
  pages   = {61--69},
  month   = feb,
  doi     = {10.1016/j.rcim.2014.07.001}
}

@article{NEW18,
  author  = {Y. N. Yang and C. C. Hua and X. P. Guan},
  title   = {Finite Time Control Design for Bilateral Teleoperation System with Position Synchronization Error Constrained},
  journal = {IEEE Transactions on Cybernetics},
  year    = {2016},
  volume  = {46},
  number  = {3},
  pages   = {609--619},
  month   = mar,
  doi     = {10.1109/TCYB.2015.2410785}
}

@article{NEW19,
  author  = {Z. W. Wang and Z. Chen and B. Liang and B. Zhang},
  title   = {A Novel Adaptive Finite Time Controller for Bilateral Teleoperation System},
  journal = {Acta Astronautica},
  year    = {2018},
  volume  = {144},
  pages   = {263--270},
  month   = mar,
  doi     = {10.1016/j.actaastro.2017.12.046}
}

@article{NEW20,
  author  = {H. C. Zhang and A. G. Song and S. B. Shen},
  title   = {Adaptive Finite-Time Synchronization Control for Teleoperation System with Varying Time Delays},
  journal = {IEEE Access},
  year    = {2018},
  volume  = {6},
  pages   = {40940--40949},
  doi     = {10.1109/ACCESS.2018.2857802}
}

@article{NEW21,
  author  = {Z. W. Wang and Z. Chen and Y. M. Zhang and X. Y. Yu and X. Wang and B. Liang},
  title   = {Adaptive Finite-Time Control for Bilateral Teleoperation Systems with Jittering Time Delays},
  journal = {International Journal of Robust and Nonlinear Control},
  year    = {2019},
  volume  = {29},
  number  = {4},
  pages   = {1007--1030},
  month   = mar,
  doi     = {10.1002/rnc.4423}
}

@article{NEW22,
  author  = {F. Hashemzadeh and M. Tavakoli},
  title   = {Position and Force Tracking in Nonlinear Teleoperation Systems under Varying Delays},
  journal = {Robotica},
  year    = {2015},
  volume  = {33},
  number  = {4},
  pages   = {1003--1016},
  month   = may,
  doi     = {10.1017/S026357471400068X}
}

@article{NEW23,
  author  = {S. Ganjefar and S. Rezaei and F. Hashemzadeh},
  title   = {Position and Force Tracking in Nonlinear Teleoperation Systems with Sandwich Linearity in Actuators and Time-Varying Delay},
  journal = {Mechanical Systems and Signal Processing},
  year    = {2017},
  volume  = {86},
  pages   = {308--324},
  month   = mar,
  doi     = {10.1016/j.ymssp.2016.09.023}
}

@article{NEW24,
  author  = {H. Amini and S. M. Rezaei and M. Zareinejad and H. Ghafarirad},
  title   = {Enhanced Time Delayed Linear Bilateral Teleoperation System by External Force Estimation},
  journal = {Transactions of the Institute of Measurement and Control},
  year    = {2013},
  volume  = {35},
  number  = {5},
  pages   = {637--647},
  month   = jul,
  doi     = {10.1177/0142331212464643}
}

@article{NEW25,
  author  = {L. P. Chan and F. Naghdy and D. Stirling},
  title   = {Position and Force Tracking for Non-Linear Haptic Telemanipulator under Varying Delays with an Improved Extended Active Observer},
  journal = {Robotics and Autonomous Systems},
  year    = {2016},
  volume  = {75},
  pages   = {145--160},
  month   = jan,
  doi     = {10.1016/j.robot.2015.10.007}
}

@article{NEW26,
  author  = {D. Sun and Q. F. Liao and T. Stoyanov and A. Kiselev and A. Loutfi},
  title   = {Bilateral Telerobotic System Using Type-2 Fuzzy Neural Network Based Moving Horizon Estimation Force Observer for Enhancement of Environmental Force Compliance and Human Perception},
  journal = {Automatica},
  year    = {2019},
  volume  = {106},
  pages   = {358--373},
  month   = aug,
  doi     = {10.1016/j.automatica.2019.04.033}
}

@article{NEW27,
  author  = {Y. Yuan and Y. J. Wang and L. Guo},
  title   = {Force Reflecting Control for Bilateral Teleoperation System under Time-Varying Delays},
  journal = {IEEE Transactions on Industrial Informatics},
  year    = {2019},
  volume  = {15},
  number  = {2},
  pages   = {1162--1172},
  month   = feb,
  doi     = {10.1109/TII.2018.2822670}
}

@article{ID1,
  author  = {Siyu Lu and Yuxi Ban and Xia Zhang and Bo Yang and Shan Liu and Lirong Yin and Wenfeng Zheng},
  title   = {Adaptive Control of Time Delay Teleoperation System with Uncertain Dynamics},
  journal = {Frontiers in Neurorobotics},
  year    = {2022},
  volume  = {16},
  pages   = {928863},
  doi     = {10.3389/fnbot.2022.928863}
}

@article{ID3,
  author  = {Kesong Chen and Haochen Zhang},
  title   = {Design of Synchronization Tracking Adaptive Control for Bilateral Teleoperation System with Time-Varying Delays},
  journal = {Sensors},
  year    = {2022},
  volume  = {22},
  number  = {20},
  pages   = {7798},
  doi     = {10.3390/s22207798}
}

@article{ID4,
  author  = {Marco Moran-Armenta and Kleber Pati{\~{n}}o and Emanuel Slawi{\'n}ski and Francisco G. Rossomando and Vicente Mut and Javier Moreno-Valenzuela},
  title   = {Adaptive Compensation for Bilateral Teleoperation Systems With Asymmetric Time-Varying Delays Using Neural Network in {P+d} Control Framework},
  journal = {IEEE Access},
  year    = {2025},
  volume  = {13},
  pages   = {62633--62647},
  doi     = {10.1109/ACCESS.2025.3558119}
}

@article{ID5,
  author  = {Haochen Zhang and Aiguo Song and Huijun Li and Dapeng Chen and Liqiang Fan},
  title   = {Adaptive Finite-Time Control Scheme for Teleoperation With Time-Varying Delay and Uncertainties},
  journal = {IEEE Transactions on Systems, Man, and Cybernetics: Systems},
  year    = {2022},
  volume  = {52},
  number  = {3},
  pages   = {1552--1566},
  doi     = {10.1109/TSMC.2020.3032295}
}

@article{ID7,
  author  = {Fanghao Huang and Wei Zhang and Zheng Chen and Jianzhong Tang and Wei Song and Shiqiang Zhu},
  title   = {{RBFNN}-Based Adaptive Sliding Mode Control Design for Nonlinear Bilateral Teleoperation System Under Time-Varying Delays},
  journal = {IEEE Access},
  year    = {2019},
  volume  = {7},
  pages   = {11905--11912},
  doi     = {10.1109/ACCESS.2019.2891887}
}

@inproceedings{ID10,
  author    = {Zhaohui An and Gaochen Min and Xiangzuo Jiang and Xinbo Yu and Wei He and Carlos Silvestre},
  title     = {Adaptive {RBF}-Neural-Network Control with Force Observer for Teleoperation Robotic System},
  booktitle = {Cognitive Computation and Systems: First International Conference, {ICCCS} 2022, Revised Selected Papers},
  series    = {Communications in Computer and Information Science},
  volume    = {1732},
  pages     = {315--330},
  publisher = {Springer},
  year      = {2023},
  doi       = {10.1007/978-981-99-2789-0_27}
}

@article{ID11,
  author  = {Jingwen Wang and Jiawei Tian and Xia Zhang and Bo Yang and Shan Liu and Lirong Yin and Wenfeng Zheng},
  title   = {Control of Time Delay Force Feedback Teleoperation System With Finite Time Convergence},
  journal = {Frontiers in Neurorobotics},
  year    = {2022},
  volume  = {16},
  pages   = {877069},
  doi     = {10.3389/fnbot.2022.877069}
}

@article{ID13,
  author  = {Shuang Zhang and Shuo Yuan and Xinbo Yu and Linghuan Kong and Qing Li and Guang Li},
  title   = {Adaptive Neural Network Fixed-Time Control Design for Bilateral Teleoperation With Time Delay},
  journal = {IEEE Transactions on Cybernetics},
  year    = {2022},
  volume  = {52},
  number  = {9},
  pages   = {9756--9769},
  doi     = {10.1109/TCYB.2021.3063729}
}

@article{ID52,
  author  = {Phuong Nam Dao and Van Tinh Nguyen and Yen-Chen Liu},
  title   = {Finite-Time Convergence for Bilateral Teleoperation Systems with Disturbance and Time-Varying Delays},
  journal = {IET Control Theory \& Applications},
  year    = {2021},
  volume  = {15},
  number  = {13},
  pages   = {1736--1748},
  doi     = {10.1049/cth2.12155}
}

@article{ID56,
  author  = {Tinh-Van Nguyen and Yen-Chen Liu},
  title   = {Advanced Finite-Time Control for Bilateral Teleoperators With Delays and Uncertainties},
  journal = {IEEE Access},
  year    = {2021},
  volume  = {9},
  pages   = {141951--141960},
  doi     = {10.1109/ACCESS.2021.3119578}
}

@article{ID64,
  author  = {Lei Xi and Haochen Zhang and Jianrong Liu and Wenxu Zhang and Yangming Fan},
  title   = {A Novel {P+d} Control Scheme for Time-Delayed Telerobotic Systems with Damping Adjustment},
  journal = {Processes},
  year    = {2025},
  volume  = {13},
  number  = {3},
  pages   = {611},
  doi     = {10.3390/pr13030611}
}

@article{ID69,
  author  = {Erchao Li and Ze Wang and Haochen Zhang},
  title   = {Zeroing Neural Dynamics-Based Adaptive Interval Type-2 Fuzzy Neural Network Control for Bilateral Teleoperation Systems},
  journal = {Journal of Control, Automation and Electrical Systems},
  year    = {2026},
  volume  = {37},
  number  = {4},
  pages   = {787--805},
  doi     = {10.1007/s40313-026-01284-8}
}

@article{ID72,
  author  = {Yana Yang and Huixin Jiang and Lu Gan and Changchun Hua and Junpeng Li},
  title   = {Fixed-Time Composite Neural Learning Control of Flexible Telerobotic Systems},
  journal = {IEEE Transactions on Cybernetics},
  year    = {2024},
  volume  = {54},
  number  = {6},
  pages   = {3602--3614},
  doi     = {10.1109/TCYB.2023.3325425}
}

@article{ID74,
  author  = {Hang Li and Wusheng Chou},
  title   = {Adaptive {FNN} Backstepping Control for Nonlinear Bilateral Teleoperation With Asymmetric Time Delays and Uncertainties},
  journal = {International Journal of Control, Automation and Systems},
  year    = {2023},
  volume  = {21},
  number  = {9},
  pages   = {3091--3104},
  doi     = {10.1007/s12555-022-0158-9}
}

@article{ID75,
  author  = {Longnan Li and Zhengxiong Liu and Zhiqiang Ma and Xing Liu and Jianhui Yu and Panfeng Huang},
  title   = {Adaptive Neural Learning Finite-Time Control for Uncertain Teleoperation System with Output Constraints},
  journal = {Journal of Intelligent \& Robotic Systems},
  year    = {2022},
  volume  = {105},
  number  = {4},
  pages   = {76},
  doi     = {10.1007/s10846-022-01675-4}
}

@article{ID76,
  author  = {Longnan Li and Zhengxiong Liu and Zhiqiang Ma and Panfeng Huang and Shaofan Guo},
  title   = {Adaptive Neural Learning Fixed-Time Control for Uncertain Teleoperation System with Time-Delay and Time-Varying Output Constraints},
  journal = {International Journal of Robust and Nonlinear Control},
  year    = {2022},
  volume  = {32},
  number  = {16},
  pages   = {8912--8931},
  doi     = {10.1002/rnc.6319}
}

@article{ID81,
  author  = {Naveen Kumar and Niharika Thakur and Yogita Gupta},
  title   = {Time Delay Compensated Disturbance Observer-Based Sliding Mode Slave Controller and Neural Network Model for Bilateral Teleoperation System},
  journal = {Intelligent Service Robotics},
  year    = {2024},
  volume  = {17},
  number  = {4},
  pages   = {931--943},
  doi     = {10.1007/s11370-024-00546-1}
}

@article{ID87,
  author  = {Parham Mohsenzadeh Kebria and Abbas Khosravi and Saeid Nahavandi and Dongrui Wu and Fernando Bello},
  title   = {Adaptive Type-2 Fuzzy Neural-Network Control for Teleoperation Systems With Delay and Uncertainties},
  journal = {IEEE Transactions on Fuzzy Systems},
  year    = {2020},
  volume  = {28},
  number  = {10},
  pages   = {2543--2554},
  doi     = {10.1109/TFUZZ.2019.2941173}
}

@inproceedings{ID88,
  author    = {Parham M. Kebria and Abbas Khosravi and Saeid Nahavandi},
  title     = {Neural Network Control of Teleoperation Systems with Delay and Uncertainties Based on Multilayer Perceptron Estimations},
  booktitle = {2020 International Joint Conference on Neural Networks ({IJCNN})},
  year      = {2020},
  publisher = {IEEE},
  pages     = {1--7},
  doi       = {10.1109/IJCNN48605.2020.9207035}
}

@inproceedings{ID105,
  author    = {A. Khanzadeh and S. Ganjefar and M. Ghaemifar},
  title     = {An Adaptive Controller for Cooperative Teleoperation System with Time-Varying Delay and Formation},
  booktitle = {2023 11th RSI International Conference on Robotics and Mechatronics ({ICRoM})},
  address   = {Tehran, Iran},
  year      = {2023},
  publisher = {IEEE},
  pages     = {48--53},
  doi       = {10.1109/ICRoM60803.2023.10412586}
}

@article{NEWLSM1,
  author  = {Wolfgang Maass and Thomas Natschl{\"a}ger and Henry Markram},
  title   = {Real-Time Computing Without Stable States: A New Framework for Neural Computation Based on Perturbations},
  journal = {Neural Computation},
  year    = {2002},
  volume  = {14},
  number  = {11},
  pages   = {2531--2560},
  month   = nov,
  doi     = {10.1162/089976602760407955}
}

@article{NEWLSM2,
  author  = {Mantas Luko{\v{s}}evi{\v{c}}ius and Herbert Jaeger},
  title   = {Reservoir Computing Approaches to Recurrent Neural Network Training},
  journal = {Computer Science Review},
  year    = {2009},
  volume  = {3},
  number  = {3},
  pages   = {127--149},
  month   = aug,
  doi     = {10.1016/j.cosrev.2009.03.005}
}

@incollection{NEWLSM3,
  author    = {Wolfgang Maass},
  title     = {Liquid State Machines: Motivation, Theory, and Applications},
  booktitle = {Computability in Context: Computation and Logic in the Real World},
  editor    = {S. Barry Cooper and Andrea Sorbi},
  year      = {2011},
  pages     = {275--296},
  publisher = {Imperial College Press},
  address   = {London, U.K.},
  doi       = {10.1142/9781848162778_0008}
}

@inproceedings{NEWLSM4,
  author    = {Michael R. Smith and Aaron J. Hill and Kristofor D. Carlson and Craig M. Vineyard and Jonathon Donaldson and David R. Follett and Pamela L. Follett and John H. Naegle and Conrad D. James and James B. Aimone},
  title     = {A Novel Digital Neuromorphic Architecture Efficiently Facilitating Complex Synaptic Response Functions Applied to Liquid State Machines},
  booktitle = {2017 International Joint Conference on Neural Networks ({IJCNN})},
  year      = {2017},
  pages     = {2421--2428},
  address   = {Anchorage, AK, USA},
  month     = may,
  publisher = {IEEE},
  doi       = {10.1109/IJCNN.2017.7966150}
}

@article{manna2023plsm,
  author  = {Siladittya Manna and Dipayan Das and Saumik Bhattacharya and Umapada Pal and Sukalpa Chanda},
  title   = {{PLSM}: A Parallelized Liquid State Machine for Unintentional Action Detection},
  journal = {IEEE Transactions on Emerging Topics in Computing},
  year    = {2023},
  volume  = {11},
  number  = {2},
  pages   = {474--484},
  doi     = {10.1109/TETC.2022.3211011}
}

@article{eshraghian2023,
  author  = {Jason K. Eshraghian and Max Ward and Emre O. Neftci and Xinxin Wang and Gregor Lenz and Girish Dwivedi and Mohammed Bennamoun and Doo Seok Jeong and Wei D. Lu},
  title   = {Training Spiking Neural Networks Using Lessons From Deep Learning},
  journal = {Proceedings of the IEEE},
  year    = {2023},
  volume  = {111},
  number  = {9},
  pages   = {1016--1054},
  month   = sep,
  doi     = {10.1109/JPROC.2023.3308088}
}

@book{spong2020,
  author    = {Mark W. Spong and Seth Hutchinson and M. Vidyasagar},
  title     = {Robot Modeling and Control},
  edition   = {2},
  publisher = {John Wiley \& Sons},
  year      = {2020},
  isbn      = {978-1-119-52399-4}
}

@article{zhang2007,
  author  = {Wei Zhang and Henry Y. Chen and Ghassan S. Kassab},
  title   = {A Rate-Insensitive Linear Viscoelastic Model for Soft Tissues},
  journal = {Biomaterials},
  year    = {2007},
  volume  = {28},
  number  = {24},
  pages   = {3579--3586},
  month   = aug,
  doi     = {10.1016/j.biomaterials.2007.04.040}
}

@misc{Attarzadeh2026,
  author        = {Armin Attarzadeh and Mohammad Ali Ghaemifar
                   and Alireza Khanzadeh and Soheil Ganjefar},
  title         = {Deep Reinforcement Learning for Adaptive Gain Tuning
                   in Control of Teleoperation Manipulators with Joint
                   Flexibility and Time-Varying Delays},
  year          = {2026},
  eprint        = {2607.21145},
  archivePrefix = {arXiv},
  primaryClass  = {eess.SY},
  doi           = {10.48550/arXiv.2607.21145}
}

@article{Khosravi2024,
  author  = {Kosar Khosravi and Mohammadali Ghaemifar and Saeed Ebadollahi},
  title   = {Enhancing Spatial Awareness: A Survey of Camera-Based
             Frontal View to Bird's-Eye-View Conversion},
  journal = {SSRN Electronic Journal},
  year    = {2024},
  note    = {Available at SSRN 4829554},
  doi     = {10.2139/ssrn.4829554}
}

@inproceedings{Arabsorkhi2024,
  author    = {Hamid Arabsorkhi and Mohammadali Ghaemifar
               and Saeed Ebadollahi and Saina Moradi},
  title     = {{GA}-Tuned Ensemble Learning for Improving the Performance
               of {Wi-Fi RSS}-Based Indoor Positioning},
  booktitle = {2024 10th International Conference on Web Research ({ICWR})},
  year      = {2024},
  pages     = {349--354},
  publisher = {IEEE},
  doi       = {10.1109/ICWR61162.2024.10533318}
}

@misc{Ghaemifar2026,
  author        = {Mohammadali Ghaemifar and Arshia Goshtasbi
                   and Arian Hajizadeh and Armin Attarzadeh
                   and Erfan Riazati},
  title         = {Adaptive {RBFNN} Control of Uncertain Bilateral
                   Teleoperation Systems with Delay-Dependent {LMI}
                   Stability Conditions},
  year          = {2026},
  eprint        = {2608.20182},
  archivePrefix = {arXiv},
  primaryClass  = {eess.SY},
  doi           = {10.48550/arXiv.2608.20182}
}


\end{document}